\documentclass[12pt,onecolumn,draftclsnofoot,journal]{IEEEtran}
\usepackage{graphicx}
\DeclareGraphicsExtensions{.pdf}

\usepackage{amsfonts}
\usepackage{rotating}
\usepackage[caption=false]{subfig}
\usepackage[cmex10]{amsmath} 
\usepackage{array}
\usepackage{cite}
\usepackage{amssymb}
\usepackage{pifont}
\usepackage{amsfonts}
\usepackage{amsmath}
\usepackage{stackrel}
\usepackage{booktabs,multirow}
\usepackage{arydshln}
\usepackage{slashbox}
\usepackage{amsthm}
\usepackage{listings}
\usepackage{algorithmicx}
\usepackage{algorithm}
\usepackage{algpseudocode}
\usepackage{tablefootnote}
\usepackage{threeparttable}
\usepackage{hhline}

\DeclareMathAlphabet{\mathbcal}{OMS}{cmsy}{b}{n}
\DeclareMathAlphabet{\mathcal}{OMS}{cmsy}{b}{n}
\DeclareMathAlphabet{\mathfrak}{OMS}{cmsy}{b}{n}
\newtheorem{lem}{Lemma}
\newtheorem{rem}{Remark}
\newtheorem{theo}{Theorem}
\newtheorem{mydef}{Definition}
\newtheorem{cor}{Corollary}
\newtheorem{pro}{Proposition}
\newtheorem{ex}{Example}

\newfloat{routine}{htbp}{loa}
\floatname{routine}{Routine} 
\makeatletter
\newcommand{\algmargin}{\the\ALG@thistlm}
\makeatother
\newlength{\forwidth}
\algdef{SE}[parFOR]{parFor}{EndparFor}[1]
  {\parbox[t]{\dimexpr\linewidth-\algmargin}{
     \hangindent\forwidth\strut\algorithmicfor\ #1\ \algorithmicdo\strut}}{\algorithmicend\ \algorithmicfor}
\algnewcommand{\parState}[1]{\State
  \parbox[t]{\dimexpr\linewidth-\algmargin}{\strut #1\strut}}

\newlength{\ifwidth}
\algdef{SE}[parIF]{parIf}{EndparIf}[1]
  {\parbox[t]{\dimexpr\linewidth-\algmargin}{
     \hangindent\ifwidth\strut\algorithmicif\ #1\ \algorithmicdo\strut}}{\algorithmicend\ \algorithmicif}
\graphicspath {{figures/}}
\begin{document}
\title{Error Floor Analysis of LDPC Row Layered Decoders}

\author{Ali~Farsiabi and~Amir~H.~Banihashemi,~\IEEEmembership{Senior Member,~IEEE}% 
        }

\maketitle

\begin{abstract}
In this paper, we analyze the error floor of quasi-cyclic (QC) low-density parity-check (LDPC) codes decoded by the sum-product algorithm (SPA) 
with row layered message-passing scheduling. For this, we develop a linear state-space model of trapping sets (TSs) which incorporates the layered nature of scheduling. We demonstrate that the contribution of each TS to the error floor is not only a function of the topology of the TS, but also depends on the row layers in which different check nodes of the TS are located. This information, referred to as TS layer profile (TSLP), plays an important role in the harmfulness of a TS. As a result, the harmfulness of a TS in particular, and the error floor of the code in general, 
can significantly change by changing the order in which the information of different layers, corresponding to different row blocks of the parity-check matrix, is updated.%permuting the row blocks of the parity-check matrix. 
We also study the problem of finding a layer ordering that minimizes the error floor, and obtain 
row layered decoders with error floor significantly lower than that of their flooding counterparts. 
As part of our analysis, we make connections between the parameters of the state-space model for a row layered schedule and those of the flooding schedule.  
%Such connections are used to compare the harmfulness of a given TS under the two schedules. 
Simulation results are presented to show the accuracy of analytical error floor estimates.  
\end{abstract}


% \fi
\IEEEpeerreviewmaketitle
\section{Introduction}
\IEEEPARstart {F}{inite}-length low-density parity-check (LDPC) codes under iterative message passing algorithms suffer from error floor,
i.e., as the channel quality improves, at some point, the error rate does not decrease as fast as its initial rate of decrease.
The error floor problem of LDPC codes has been the topic of extensive research in recent years. In~\cite{danesh, Ryan2,Ryan1, Kyung, Zhang1,Lee-2018,homayoon2020,TB,zhang_quasiuniform,TSbreaking}, decoders with improved error floor were devised. 
For quantized message-passing decoders, it is well-known that clipping messages, in general, causes an error floor~\cite{Zhao}. The error floor generally worsens as
the dynamic range of messages is reduced~\cite{zhang_quasiuniform}.  
Different techniques to improve the error floor by adjusting the dynamic range of (some of) the messages were presented in~\cite{zhang_quasiuniform,TB,Lee-2018}.
LDPC codes with low error floor were constructed in \cite{Ivkovic, Asvadi, Khaz, Nguyen, mao2001heuristic,xiao2004improved,Sima-CL1,Sima-CL2,Bashir-TCOM,Tian-2004,Peg,zheng2010constructing,Tao-2018,Bashir-irreg,Sima-TCOM}. Characterization and enumeration of structures responsible for error floor were studied in \cite{yoones2015, hashemireg, hashemiireg, mehdi2014, mehdi2012,Wang}, and techniques to estimate the error floor were developed in~\cite{richardson,Cole, LaraIS, daneshrad,Lara_SP,Xiao,Sun_phd, Schleg, But_SS, Homayoon_SP,Sina,Ontology,Hu_magneticIS,XB-2007,Ali-TCOM,Raveendran,Zhu,Peyman}.

There are two general categories of techniques to estimate the error floor of LDPC codes. The first category is based on importance sampling techniques, and requires the full knowledge of the parity-check matrix or Tanner graph of the code to estimate the error floor~\cite{richardson,Cole, LaraIS, daneshrad,Lara_SP,Xiao,Hu_magneticIS,XB-2007}. The second category, on the other hand, is code-independent, in that, rather than the full knowledge of the code's Tanner graph, these techniques only require the multiplicity and topology of harmful substructures of the Tanner graph, referred to as {\em trapping sets} (TSs), and possibly the degree distributions of the graph, to estimate the error floor~\cite{Sun_phd, Schleg, But_SS,AS_threshold,Homayoon_SP,Ontology,Ali-TCOM}. In particular, in~\cite{Sun_phd, Schleg, But_SS}, to evaluate the performance of an LDPC code over 
the additive white Gaussian noise (AWGN) channel, a linear state-space model is used to represent the 
dynamics of the sum-product algorithm (SPA) in the vicinity of a TS of interest with the inputs to the model generated using density evolution (DE)~\cite{Urbank}. 
Based on this model, the failure probability of the TS, which is only a function of the topology of the TS and the degree distribution of the Tanner graph, is calculated.  
The error floor is then estimated as a weighted sum of these failure probabilities over dominant TSs of the code, with the weights being the multiplicities of 
different TS structures. 

Among TSs, the most harmful ones are those with only degree-$1$ and degree-$2$ check nodes in their induced subgraphs~\cite{Milen,Ryan1}.
Such TSs are called {\em elementary} (ETS). The  degree-$1$ and degree-$2$ check nodes are referred to as {\em unsatisfied} and {\em missatisfied}
check nodes, respectively. In \cite{Sun_phd}, Sun proposed a linear state-space model to analyze the dynamics of ETSs in 
the error floor region over the AWGN channel. This model was based on the assumption that the decoder behavior 
outside an ETS can be well approximated by DE. Schlegel and Zhang~\cite{Schleg} proposed an improved 
linear state-space model in which an iteration-dependent linear gain was added to the model to account for the impact 
of external messages of missatisfied check nodes on the internal messages of an ETS. More recently, Butler and Siegel \cite{But_SS} 
refined and extended the linear state-space model and used it to analyze the effect of log-likelihood ratio (LLR) saturation on the 
error floor performance of LDPC codes with fixed variable-node degrees decoded by floating-point SPA.

To the best of our knowledge, all the existing work on the theoretical analysis of error floor, including~\cite{Sun_phd, Schleg, But_SS}, is limited to 
two-phase message passing algorithms, also known as {\em flooding} or {\em parallel schedule}. In flooding schedule, each decoding iteration is divided into two parts. 
In the first (second) part of an iteration, all the variable (check) nodes compute their messages and pass them to their adjacent  check (variable) nodes simultaneously. 
There are however a variety of message passing schedules which are advantagous to flooding in terms of performance, complexity or convergence speed~\cite{L1,L2,L3,L4,L5,L6,L7,L8,L9,XiaoSchedule,NouhSchedule,InfromedDynamic_wesel_2010,M2I2}.
An important category of schedules are {\em layered or serial schedules}~\cite{L1,L2,L5,L6,L8,L7}. In such schedules, each iteration of message-passing consists of multiple 
sub-iterations performed serially. This allows a more frequent updating of the reliabilities in each iteration compared to the flooding schedule, which consequently results in a higher convergence speed. Moreover, due to the reuse of the same hardware for the implementation of different sub-iterations, the hardware resources required for the implementation of layered decoders are substantially lower than those of their flooding counterparts. Due to these attractive features, layered decoders are often used in practical applications along with quasi-cyclic (QC) LDPC codes. In such a combination, the row or column blocks of the parity-check matrix of the QC-LDPC code correspond to different layers. The layered decoder is then referred to as {\em row (horizontal) layered decoder}~\cite{L1,L2,L3,L4,L5,L6} or {\em column (vertical) layered decoder}~ \cite{L7}, respectively. In \cite{L8}, it has been shown that the convergence speed of both types of layered decoder can be twice that of  a decoder with flooding schedule.

The vast majority of research on layered decoders is devoted to the issues concerning convergence speed and efficient implementations.
In particular, the study of the error floor of such decoders has been mainly limited to empirical results~\cite{InfromedDynamic_wesel_2010,M2I2,Angarita_MS_2014,BackTrack_hard_2011,
IDS_kim_2012let,vasic_horizontal}. In~\cite{InfromedDynamic_wesel_2010} and \cite{M2I2}, the authors proposed a dynamic scheduling and 
a schedule diversity, respectively, that reduced the error floor. More recently, Raveendran and Vasic~\cite{vasic_horizontal} studied 
flooding and row layered Gallager-B algorithms applied to the $(155, 64)$ Tanner code over the binary symmetric channel, and demonstrated that while 
the former decoder gets trapped in $(5,3)$ ETSs, the latter does not, and as a result the layered decoder has a superior error floor performance compared to the flooding decoder.

Motivated by the wide spread application of layered decoders and the fact that the behavior of such decoders in the error floor region is still not well understood, 
in this paper, we aim at the theoretical analysis of the error floor of row layered decoders. In fact, to the best of our knowledge, this is the first work in 
which the dynamics of a soft layered decoder in the error floor region is theoretically analyzed. We start by developing a linear state-space model for ETSs that incorporates the 
layered nature of the decoding algorithm. We then use this model to study the dynamics of a saturating SPA in the vicinity of the TS over the AWGN channel. 
Our analysis shows that the harmfulness of a given TS, as well as the error floor of the code, can significantly change by changing the order in which the information of different layers, corresponding to different row blocks of the parity-check matrix, are updated.\footnote{Our experiments show that the waterfall performance of the codes is rather insensitive to the updating order of layers.} We then study the problem of finding the layer ordering that minimizes the error floor. As a result, we find orderings that result in error floors substantially lower than those of the decoder with flooding schedule. Connections are also made between the model parameters of  the layered decoder and those of the flooding one. 
%These connections are used to compare the harmfulness of a given TS under the two schedules. 
Finally, we compare our theoretical estimates of the error floor with Monte Carlo simulations for QC-LDPC codes, 
both variable-regular and irregular, and demonstrate a good match between the two.

The rest of the paper is organized as follows: In Section~\ref{sec1}, we present 	some preliminaries. This is followed by a review of  the linear state-space model of an ETS for a decoder with the flooding schedule in Section~\ref{sec2}. We then develop the linear state-space model of an ETS for a row layered decoder in Section~\ref{sec3}. In the same section, we also establish connections between the model parameters of the row layered and flooding schedules. In Section~\ref{Analysis_Opt_Section}, we analyze the effect of different row block permutations on some important parameters affecting the error floor performance of  the layered decoder,
and discuss the optimization of the layer ordering to minimize the error floor. In Section~\ref{sec4}, we present the simulation results to evaluate the accuracy of the theoretical results in estimating the error floor. We end the paper with some concluding remarks in Section~\ref{con}.

\section{Preliminaries} \label{sec1}
\subsection{Notations}
In this paper, matrices and vectors are denoted by boldfaced upper case and lower case letters, respectively. The only exceptions are the LLR vectors 
where the symbol $\L$ with different subscripts or superscripts is utilized for representation. All the vectors are assumed to be column vectors. 
A list of the notations and symbols used in this paper is provided in Tables \ref{notationTableLayered} and \ref{notationTableLayeredDec}.
\begin{table}[]
\centering
\scalebox{0.75}{%
\begin{threeparttable}
\caption{Notations and Symbols (Part A)}
\setlength{\tabcolsep}{0.8pt}
\renewcommand{\arraystretch}{1.2}
\label{notationTableLayered}
\begin{tabular}{||c|cccc|| }
\hline
%\hline
\multicolumn{1}{||c|}{ Notations}& \multicolumn{1}{c||}{ Descriptions}\\
\hline
\hline
%$\mathcal{C}$ &  \multicolumn{1}{c}{ A binary LDPC code}\\
%\hline
$\mathbf{H}$ & \multicolumn{1}{c||}{ $m\times n$ parity check matrix}\\
\hline
$\mathbf{H_b}$ & \multicolumn{1}{c||}{ $m_b \times n_b$ base matrix}\\
\hline
%$p$ & \multicolumn{1}{c}{  Lifting degree of a QC-LDPC code}\\
%\hline
$\mathbf{d}$ &  \multicolumn{1}{c||}{A codeword}\\
\hline
$\hat{\mathbf{d}}$ & \multicolumn{1}{c||}{ Estimated codeword}\\
\hline
%$G$ &  \multicolumn{1}{c}{ A bipartite Tanner graph}\\
%\hline
%$V$ &  \multicolumn{1}{c}{ Set of VNs of a Tanner graph}\\
%\hline
%$C$ &  \multicolumn{1}{c}{ Set of CNs of a Tanner graph}\\
%\hline
%$E$ & \multicolumn{1}{c}{ Set of edges of a Tanner graph}\\
%\hline
%$d_i$ &  \multicolumn{1}{c}{ $i$th codeword bit}\\
%\hline
%$u_i$ &  \multicolumn{1}{c}{  $i$th transmitted bit using BPSK}\\
%\hline
%$y_i$ &  \multicolumn{1}{c}{ $i$th received symbol from the channel}\\
%\hline
%$n_i$ &  \multicolumn{1}{c}{ $i$th element of the noise vector}\\
%\hline
%${\sigma_{ch}^2}$&  \multicolumn{1}{c}{ The channel noise variance}\\
%\hline
$I_{max}$& \multicolumn{1}{c||}{Maximum number of iterations}\\
\hline
$L^{ch}_i$&  \multicolumn{1}{c||}{ Channel LLR}\\
\hline
$L_{\ell}^{[i \leftarrow j]}$&  \multicolumn{1}{c||}{ $j$th CN to $i$th VN message at iteration $\ell$}\\
\hline
$L_{\ell}^{[i \rightarrow j]}$ &  \multicolumn{1}{c||}{ $i$th VN to $j$th CN message at iteration $\ell$}\\
\hline
${\tilde{L}_{\ell}^{[i]}}$ &  \multicolumn{1}{c||}{ Total LLR of $i$th VN at iteration $\ell$}\\
%\hline
%$\hat{d}_i$ &  \multicolumn{1}{c||}{ $i$th bit of the estimated codeword}\\
\hline
$m_s$ & \multicolumn{1}{c||}{ Number of state variables}\\
\hline
$\mathcal{S}$ &  \multicolumn{1}{c||}{ Set of VNs of a TS, $|\mathcal{S}|=a$ }\\
\hline
$\Gamma{(\mathcal{S})}$&  \multicolumn{1}{c||}{ Set of CNs of a TS}\\
\hline
$\Gamma_o{(\mathcal{S})}$&  \multicolumn{1}{c||}{ Set of odd-degree CNs of a TS, $|\Gamma_o{(\mathcal{S})}|=b$}\\
\hline
$\Gamma_e{(\mathcal{S})}$&  \multicolumn{1}{c||}{ Set of even-degree CNs of an ETS, $|\Gamma_e{(\mathcal{S})}|=\frac{m_s}{2}$}\\
\hline
\multirow{ 2}{*}{$\mathbf{x}^{(\ell)}$}
& \multicolumn{1}{c||}{ $m_s\times 1$ state vector of a LETS in flooding decoder }\\
&  \multicolumn{1}{c||}{at the $\ell$th iteration} \\
\hline
\multirow{ 2}{*}{$\bar{g}'_\ell$}
& \multicolumn{1}{c||}{ Missatisfied CN multiplicative gains in flooding decoder }\\
&  \multicolumn{1}{c||}{at the $\ell$th iteration}\\
\hline
$\mathbf{A}$&  \multicolumn{1}{c||}{ $m_s\times m_s$ transition matrix of a LETS in flooding decoder}\\
\hline
\multirow{ 2}{*}{$\mathbf{B}$}
& \multicolumn{1}{c||}{ $m_s\times a$ matrix determining the channel input contributions }\\
&  \multicolumn{1}{c||}{to the  state variables in flooding decoder}\\
\hline
\multirow{ 2}{*}{$\mathbf{B}_{ex}$}
& \multicolumn{1}{c||}{ $m_s\times b$ matrix determining the contribution of unsatisfied   }\\
&  \multicolumn{1}{c||}{CN inputs to state variables in flooding decoder} \\
\hline
\multirow{ 2}{*}{$\mathbf{C}$}
& \multicolumn{1}{c||}{ $a\times m_s$ matrix determining the relation of state variables   }\\
&  \multicolumn{1}{c||}{with the total LLR vector}\\
\hline
\multirow{ 2}{*}{$\mathbf{D}_{ex}$}
& \multicolumn{1}{c||}{ $a\times b$ matrix determining the contribution of unsatisfied   }\\
&  \multicolumn{1}{c||}{CN inputs to the total LLR vector}\\
\hline
\multirow{ 1}{*}{$\L$}
& \multicolumn{1}{c||}{  $a\times 1$ channel input vector in the LETS linear model  }\\
\hline
\multirow{ 2}{*}{$\L^{(\ell)}_{ex}$}
& \multicolumn{1}{c||}{ $b\times 1$ unsatisfied CN input vector in the LETS  }\\
&   \multicolumn{1}{c||}{linear model at the $\ell$th iteration}\\
\hline
\multirow{ 1}{*}{$\tilde{\L}^{(\ell)}$}
& \multicolumn{1}{c||}{ $a\times 1$ total LLR vector in the LETS linear model  }\\
\hline
$\mathbf{P}$ & \multicolumn{1}{c||}{A permutation matrix}\\
\hline
$\rho(\mathbf{M})$ & \multicolumn{1}{c||}{Spectral radius of a matrix $\mathbf{M}$}\\
\hline
$r$ & \multicolumn{1}{c||}{Dominant eigenvalue of the flooding transition matrix $\mathbf{A}$}\\
\hline
$\mathbf{w}_1^T,\mathbf{u}_1$ & \multicolumn{1}{c||}{Left and right eigenvectors corresponding to $r$}\\
\hline
$\mathbb{C}$& \multicolumn{1}{c||}{Set of complex numbers}\\
\hline
$\beta'_\ell$ &\multicolumn{1}{c||}{ Error indicator function of flooding decoder}\\
\hline
$P_e\{\mathcal{S}\}$ & \multicolumn{1}{c||}{Probability of failure of TS $\mathcal{S}$}\\
\hline
$L_i$& \multicolumn{1}{c||}{$i$th row layer of a QC-LDPC code}\\
\hline
$J$ & \multicolumn{1}{c||}{Number of missatisfied CN layers in a TS}\\
\hline
$(.)_k$ & \multicolumn{1}{c||}{$k$th element of a vector} \\
\hline
$(m_{ch},\sigma^2_{ch})$ & \multicolumn{1}{c||}{Mean and variance of the channel noise}\\
\hline
\multirow{ 2}{*}{$(m_{ex}^{(\ell)},\sigma^{2^{(\ell)}}_{ex})$}
 & \multicolumn{1}{c||}{Mean and variance of the unsatisfied CN inputs at} \\
 & \multicolumn{1}{c||}{iteration $\ell$ of the flooding decoder}\\
 \hline
 \end{tabular}
%\begin{tablenotes}
        %\item[\textdagger] ``$\sim$" means the corresponding clipping threshold is not applied.
  %\end{tablenotes}
\end{threeparttable}
}
\end{table}
\begin{table}[]
\centering
\scalebox{0.75}{%
\begin{threeparttable}
\caption{Notations and Symbols (Part B)}
\setlength{\tabcolsep}{0.8pt}
\renewcommand{\arraystretch}{1.2}
\label{notationTableLayeredDec}
\begin{tabular}{||c|cccc|| }
\hline
%\hline
\multicolumn{1}{||c|}{ Notations}& \multicolumn{1}{c||}{ Descriptions}\\
\hline
\hline
 $n_{L_j}$& \multicolumn{1}{c||}{Number of state variables of a LETS within layer $j$}\\
 \hline
$\mathcal{A}_j$&  \multicolumn{1}{c||}{$m_s\times m_s$ transition matrix of layer $j$}\\
\hline
\multirow{ 2}{*}{$\mathcal{B}_j$}
& \multicolumn{1}{c||}{ $m_s \times a$ matrix indicating the contribution of the channel  }\\
&  \multicolumn{1}{c||}{LLRs in the calculation of the state variables of the $j$th layer}\\
\hline
\multirow{ 2}{*}{$\overset{\triangleleft}{\mathfrak{B}}_{ex,j},\overset{\triangleright}{\mathfrak{B}}_{ex,j}$}&  \multicolumn{1}{c||}{$m_s\times b$ matrices indicating the contribution of $\L_{ex}^{(\ell-1)}$} \\
& \multicolumn{1}{c||}{and $\L_{ex}^{(\ell)}$ in updating the state variables within the $j$th layer}\\
\hline
\multirow{ 2}{*}{$\mathbf{G}^{(\ell)}$}& \multicolumn{1}{c||}{$m_s\times m_s$ matrix whose diagonal entries are the} \\
& \multicolumn{1}{c||}{gains corresponding to the $m_s$ state variables}\\
\hline
\multirow{ 2}{*}{$\mathfrak{G}_j^{(\ell)}$}&\multicolumn{1}{c||}{$m_s\times m_s$ gain matrix corresponding to}\\
& \multicolumn{1}{c||}{the $j$th layer of iteration $\ell$}\\
\hline
$\tilde{\mathbf{x}}^{(\ell,j)}$& \multicolumn{1}{c||}{Layered decoder state vector of layer $j$ at iteration $\ell$} \\
\hline
\multirow{ 2}{*}{$\tilde{\mathbf{B}}^{(\ell)}$}& \multicolumn{1}{c||}{ $m_s \times a$ matrix indicating the contribution of the channel  LLRs in }\\
&  \multicolumn{1}{c||}{calculation of the state variables at the $\ell$th iteration of the layered decoder}\\
\hline
\multirow{ 2}{*}{$\overset{\triangleleft}{\mathbf{ B}}_{ex}^{(\ell)}$ , $\overset{\triangleright}{\mathbf{ B}}_{ex}^{(\ell)}$}& \multicolumn{1}{c||}{$m_s\times b$ matrices illustrating the contribution of $\L_{ex}^{(\ell-1)}$ and $\L_{ex}^{(\ell)}$} \\
& \multicolumn{1}{c||}{in updating the state variables at the $\ell$th iteration of the layered decoder}\\
\hline
\multirow{ 2}{*}{$\hat{\psi}^{[k'\rightarrow j]}_\ell$}&  \multicolumn{1}{c||}{Probability distribution of the messages from}\\
& \multicolumn{1}{c||}{virtual VN $v_{k'}$ to missatisfied CN $c_j$ at iteration $\ell$}\\
\hline
\multirow{ 2}{*}{$P_{inv,\ell}^{[k'\rightarrow j]}$}& \multicolumn{1}{c||}{Probability of polarity inversion in the messages from }\\
&  \multicolumn{1}{c||}{virtual VN $v_{k'}$ to missatisfied CN $c_j$ at iteration $\ell$}\\
\hline
\multirow{ 2}{*}{$\bar{g}^{(\ell)}_{c_j},\bar{g}'^{(\ell)}_{c_j}$}&\multicolumn{1}{c||}{Average gains of missatisfied CN $c_j$ at the $\ell$th iteration of layered } \\
&   \multicolumn{1}{c||}{decoder before and after adding the polarity inversion, respectively}\\
\hline
\multirow{ 2}{*}{$\tilde{\mathbf{A}}_{J \rightarrow 1}$}& \multicolumn{1}{c||}{Transition matrix of a LETS with $J$ layers}\\ &\multicolumn{1}{c||}{in layered decoding}\\
\hline
$\tilde{r}$& \multicolumn{1}{c||}{Dominant eigenvalue of $\tilde{\mathbf{A}}_{J \rightarrow 1}$}\\
\hline
$\tilde{\mathbf{w}}_1^T,\tilde{\mathbf{u}}_1$& \multicolumn{1}{c||}{Left and right eigenvectors corresponding to $\tilde{r}$ }\\
\hline
\multirow{ 2}{*}{$\tilde{\mathbf{A}}$} & \multicolumn{1}{c||}{The only irreducible diagonal block of the Frobenius}\\
& \multicolumn{1}{c||}{normal form of $\tilde{\mathbf{A}}_{J \rightarrow 1}$ for LETSs that are not simple cycles}\\
\hline
\multirow{ 1}{*}{$\tilde{\boldsymbol{\omega}}_1^T,\tilde{\boldsymbol{\nu}}_1$} & \multicolumn{1}{c||}{Dominant left and right eigenvectors of $\tilde{\mathbf{A}}$}\\
\hline
\multirow{ 2}{*}{$D_{l}$}& \multicolumn{1}{c||}{A LETS digraph of the layered decoder }\\
& \multicolumn{1}{c||}{with nodes $V_l$ and edges $E_l$}\\
\hline
\multirow{ 2}{*}{$D_{f}$}& \multicolumn{1}{c||}{A LETS digraph of the flooding decoder }\\
& \multicolumn{1}{c||}{with nodes $V_f$ and edges $E_f$}\\
\hline
{$n_z$}& \multicolumn{1}{c||}{Number of zero columns of the transition matrix of the layered decoder}\\
\hline
\multirow{ 2}{*}{$(\mathbf{m}_{ex}^{(\ell)},\mathbf{\Sigma}_{ex}^{(\ell)})$}& \multicolumn{1}{c||}{$b\times 1$ mean vector and $b\times b$ covariance matrix of}  \\
 &  \multicolumn{1}{c||}{the inputs from unsatisfied CNs of a LETS at iteration $\ell$}\\
 \hline
 $\tilde{\beta}^{(\ell)}$ &\multicolumn{1}{c||}{Error indicator function of layered decoder}\\
 \hline
 $\Upsilon_i$ & \multicolumn{1}{c||}{Size of the $i$th TS group}\\
 \hline
 %\multirow{ 2}{*}
 {$\Pi_J$}& \multicolumn{1}{c||}{Set of all the $J!$ permutations of the layers of a LETS}\\
% & \multicolumn{1}{c||}{involved in an LETS}\\
 %\hline
 %$\pi$ &\multicolumn{1}{c||}{An instance of the permutations of the layers}\\
 \hline
 \multirow{ 2}{*}{$\stackrel{\leftarrow}{\bar{\psi}_{\ell}^j}$} & \multicolumn{1}{c||}{Average distribution of CN to VN messages} \\
& \multicolumn{1}{c||}{at layer $j$ of iteration $\ell$} \\
 \hline
 \multirow{ 2}{*}{$\stackrel{\rightarrow}{\bar{\psi}_{\ell}^j}$} & \multicolumn{1}{c||}{Average distribution of VN to CN messages} \\
& \multicolumn{1}{c||}{at layer $j$ of iteration $\ell$}\\
\hline 
\end{tabular}
%\begin{tablenotes}
        %\item[\textdagger] ``$\sim$" means the corresponding clipping threshold is not applied.
  %\end{tablenotes}
\end{threeparttable}
}
\end{table}
\subsection{LDPC codes, SPA, flooding and layered schedules and TSs } 
%In channel coding, the main idea is to add redundancy to the transmitted information bits in such a way that the corrupted data at the receiver can be recovered. 
Consider a binary LDPC code $\mathcal{C} $ with  parity-check matrix $\mathbf{H}$. %The entries of $\mathbf{H}$ are elements of a particular binary field, $\mathbb{F}_2$. 
A codeword of $\mathcal{C}$ is denoted by $\bold{d}$, and satisfies $\mathbf{H}\bold{d} = \bold{0}$. 
%What distinguishes an LDPC code from the other type of block codes is that its associated $\bold{H}$ matrix is sparse meaning that the non-zero entries form a low percentage of the elements of matrix $\mathbf{H}$.
Let $G=(V\cup C,E)$ be the Tanner graph representing $\mathcal{C} $, where $V=\{ v_1,v_2,\dots,v_n \}$ and $C=\{ c_1,c_2,\dots,c_m \}$ are the 
sets of variable nodes (VNs) and check nodes (CNs), respectively, and $E=\{ e_1,e_2,\dots,e_k \}$ is the set of edges.
Suppose that $\mathcal{C}$ is used for transmission over an AWGN channel using a binary phase shift keying (BPSK) modulation, where 
codeword bits $d_i$ are mapped to modulated symbols $u_i = (-1)^{d_i}$. At the channel output, we thus have $y_i=u_i+n_i$, where $\{n_i\}$ represents the noise
and is a zero-mean Gaussian random process with independent and identically distributed (i.i.d.) values, each with variance $\sigma_{ch}^2$.

%Any LDPC code may be represented by a \emph{bipartite} Tanner graph $G=(V,C,E)$. The term bipartite refers to the partition of the \emph{nodes} of the graph in to two different sets, the variable nodes (VNs) $V=\{ v_1,v_2,\dots,v_n \}$ and the check nodes (CNs) $C=\{ c_1,c_2,\dots,c_m \}$ that are connected via the \emph{edges} represented by the set $E=\{ e_1,e_2,\dots,e_k \}$. A VN, $v_i$, is connected to a CN, $c_j$, via an edge if and only if its corresponding element in parity check matrix is one, $\mathbf{H}(j,i)=1$.

%Assume that we want to transmit the data from a source to a destination through a communication channel such as additive white Gaussian noise (AWGN) channel. In this regard, first each block of binary data is coded by a codeword from $\mathcal{C} $ and after mapping to the constellation points of a specific modulation, is transmitted through the communication channel. For example, assuming the binary phase shift keying (BPSK) modulation, the binary codeword's bits, $d_i\in\{0,1\}$, are mapped to the constellation points $\{+1,-1\}$, respectively. Let $u_i \in \{-1,1\} $ represents the $i$th transmitted bit through the AWGN channel. At the channel output, a noisy version of the code symbols, $y_i=u_i+n_i$, are received where $y_i$ and $n_i$ represent the $i$th received symbol and $i$th component of the zero mean Gaussian noise, respectively. 

For decoding, we assume SPA is used with the following channel LLR values as the input:
%Given the noisy blocks of data at the receiver side, due to the large length of the practical codewords, maximum likelihood (ML) sequence decoding is not feasible. As a result, the decoding operation is done according to the message passing (MP) algorithms. The sum product algorithm (SPA), as a widely used MP algorithm, by performing iteratively on the Tanner graph of the code, can be utilized to decode an LDPC codeword. The input of the SPA is the log-likelihood ratio (LLR) of the received symbols. Let ${v_i}$ represents the $i$th VN corresponding to the $i$th symbol of the received block. The LLR of this VN in the binary AWGN channel is defined as follows
\begin{equation}
{L^{ch}_i=2y_i/\sigma_{ch}^2}\:.
\label{CH_LLR}
\end{equation}
%where ${\sigma_{ch}^2}$ is the channel noise variance.
The message sent from VN $v_i$ to CN $c_j$ at iteration $\ell$ is given by
\begin{equation}
{{L_{\ell}^{[i \rightarrow j]}}=L^{ch}_i+\sum_{k \in M(i)\backslash j}{L_{\ell-1}^{[i \leftarrow k]}}}\:,
\label{v2c_message}
\end{equation}
where ${M(i)}\backslash j$ represents the set of CNs adjacent to VN  ${v_i}$ excluding $c_j$, and ${L_{\ell-1}^{[i \leftarrow k]}}$ denotes the message sent from CN $c_k$ to VN $v_i$ at iteration $\ell -1$. 
%We use the notation ${M(i)}$ ($N(j)$) to represent the set of CNs (VNs) adjacent to VN (CN) ${v_i}$ ($c_j$). %Similarly, we indicate all variable nodes adjacent to check node ${j}$, excluding variable node ${i}$, by ${N(j)\backslash{i}}$. 
%Denoting the messages, at the $\ell$th iteration, sent from variable node ${v_i}$ to check node ${c_j}$ and vice-versa by ${L_{\ell}^{[i \rightarrow j]}}$ and ${L_{\ell}^{[i \leftarrow j]}}$, the messages sent from VNs to CNs are given by
At the first iteration of the algorithm, all the messages sent from check nodes to variable nodes at iteration $\ell-1$ are assumed to be zero in (\ref{v2c_message}).
The CN to VN messages in SPA  are computed as

\begin{equation}
{L_{\ell}^{[i \leftarrow j]}=2\tanh^{-1}\Bigg[{\prod_{k \in N(j)\backslash i }\tanh\frac{{L_{\ell}^{[k \rightarrow j]}}}{2}}\Bigg]}\:,
\label{SPA_CN}
\end{equation} 
where $N(j)\backslash i$ represents the set of VNs adjacent to CN $c_j$ excluding $v_i$.
At the end of each iteration, for each $i \in \{1,\ldots,n\}$, first, the total LLR is calculated by 
\begin{equation}
  {{\tilde{L}_{\ell}^{[i]}}=L^{ch}_i+\sum_{k \in M(i)}{L_{\ell}^{[i \leftarrow k]}}},
  \label{TotLLR}
\end{equation}
and then, a hard decision is made by
\begin{equation}
\hat{d_i}=[\text{sign}({\tilde{L}_{\ell}^{[i]}})+1]/2.
\end{equation}

If the decoded block, $\hat{\bold{d}}$, at the end of iteration ${\ell} \leq I_{max}$, is a codeword, i.e., if $\mathbf{H}\hat{\bold{d}}=\bold{0}$, then the decoding is terminated successfully. ($I_{max}$ is the maximum number of iterations.) Otherwise, if iteration $I_{max}$ is completed and still no codeword is found, then a decoding failure is declared. 

To circumvent the numerical errors in the calculation of $\tanh(x)$ in (\ref{SPA_CN}), similar to \cite{butler_numerical,zhang_quasiuniform}, 
we use the following equivalent operation to calculate CN to VN messages:
%With respect to equation (\ref{SPA_CN}), finite precision implementation of $\tanh(x)$ accompanies with numerical saturation problems. As a solution, in \cite{butler_numerical,zhang_quasiuniform},             it is suggested to use the following equivalent of the CN to VN operation 
\begin{equation}
{L_{\ell}^{[i \leftarrow j]}=\underset{k \in N(j)\backslash i }{\boxplus}}{L_{\ell}^{[k \rightarrow j]}},
\label{boxplus1}
\end{equation}
in which the pairwise box-plus operator, $\boxplus$, is defined as
\begin{equation}
\begin{split}
x_1 \boxplus x_2 & = \ln\Bigg(\dfrac{1+e^{x_1+x_2}}{e^{x_1}+e^{x_2}}\Bigg) \\  & =\text{sign}(x_1)\text{sign}(x_2).\min(|x_1|,|x_2|)+s(x_1,x_2)\:.
\end{split}
\label{boxplus_pair}
\end{equation}
The term $s(x_1,x_2)$ in (\ref{boxplus_pair}) %, is called the correction term and 
is given by 
\begin{equation}
{s(x_1,x_2)=\ln\Bigg(1+e^{-|x_1+x_2|}
\Bigg)-\ln\Bigg(1+e^{-|x_1-x_2|}
\Bigg).}
\label{s_correction}
\end{equation}

It should be noted that Equations \eqref{v2c_message}, \eqref{SPA_CN}, and \eqref{boxplus1} are edgewise operations, and can thus 
be executed within various scheduling frameworks. In decoders with \emph{flooding schedule}, at the first half of each iteration, Equation \eqref{v2c_message} 
is executed for all the VN to CN messages. In the second half of the iteration, all the CN to VN messages are updated based on \eqref{SPA_CN}. 
%Then, using \eqref{TotLLR}, the LLRs of all the VNs are updated together. 
For decoders with \emph{row layered schedule}, the CNs are partitioned into different subgroups (layers). 
Within each CN subgroup, messages are only generated on the edges between the CNs in the subgroup and their adjacent VNs while the other edges in the graph remain inactive. After the messages within a subgroup are generated, the updated reliabilities are used in the following layers. 
In Algorithm \ref{Layered_alg}, the steps of the row layered SPA is presented for a QC-LDPC code. In this paper, we consider 
QC-LDPC codes whose parity-check matrices ${\bf H}$ consist of an $m_b \times n_b$ array of circulant permutation matrices (CPMs) of size $p \times p$ and zero matrices of the same size. The Tanner graph of such codes can be considered as a cyclic $p$-lifting of a bipartite {\em base graph} with $n_b$ VNs and $m_b$ CNs.
Each layer of the row layered decoder in this case corresponds to one row block of ${\bf H}$, and thus, there are $m_b$ layers.
%The parameters $n_b$, $m_b$ and $p$ are the number of columns of $H_b$, the number of rows of $H_b$ and the size of circulant known as lifting degree, respectively. The number of row layers in horizontal layered decoding is also the same as the number of rows in base matrix, $m_b$. 
It should be noted that the equation in Line 7 of Algorithm \ref{Layered_alg} is derived by combining \eqref{v2c_message} and \eqref{TotLLR}. 
As can be seen, the total LLR of each VN may be updated several times within an iteration. So, for simplicity, the iteration index, $\ell$, 
is removed from the symbol ${\tilde{L}^{[i]}}$.
%%%%%%%%%%%%%%%%%testtesttest
 
 \begin{algorithm}
\caption{Row Layered Box-Plus Decoding Algorithm }
\label{Layered_alg}
 \begin{algorithmic} [1] 
 \State \textbf{Input:} Channel LLRs for all the VNs. 
 \State {\textbf{Initialization:} All the total LLRs are initialized with the channel LLRs, ${\tilde{L}^{[i]}}=L^{ch}_i$ for VNs $i=1:n_bp$. }
 \For{iteration $\ell=1,\dots,I_{max}$}
 \For{The CN group (layer) number $z=1:m_b$}
 \For{CNs $j=(z-1)p+1,\dots,zp$, in layer $z$}
 \ForAll{VNs $i \in N(j)$} 
 \State{${{L_{\ell}^{[i \rightarrow j]}}={\tilde{L}^{[i]}}-{L_{\ell-1}^{[i \leftarrow j]}}}$.}
 \EndFor
 \ForAll{VNs $i \in N(j)$}
 \State{${L_{\ell}^{[i \leftarrow j]}=\underset{k \in N(j)\backslash i }{\boxplus}}{L_{\ell}^{[k \rightarrow j]}}$,}
 \State{${\tilde{L}^{[i]}}={{L_{\ell}^{[i \rightarrow j]}}}+L_{\ell}^{[i \leftarrow j]}$.}
 \EndFor
 \EndFor
 \EndFor
 \par{Hard-decision:}
 \ForAll{VNs $i=1,\dots,n_bp$}
 \State{$\hat{d_i}=[\text{sign}({\tilde{L}^{[i]}})+1]/2$.}
 \EndFor
 \If{$\bold{H}\hat{\bold{d}}=\bold{0}$}
 \State{Break.}
 \EndIf
 \EndFor 
\State \textbf{Output: $\hat{\bold{d}}$} 
%\State \textbf{Output:} $\mathcal{L}^p$.
 \end{algorithmic}
 \end{algorithm}
 %%%%%%%%%%%%%%%%%
%\subsection{Definition of graph and digraph, ...}    
\subsection{Trapping sets (TSs)} 
It is well-known that in the error floor region, the majority of the errors of iterative decoding algorithms occur as a result of the decoder getting trapped within a subset $\mathcal{S}$ of VNs, i.e., although all the bits outside $\mathcal{S}$ have correct values for sufficiently large number of iterations, all the bits inside $\mathcal{S}$ are in error.
In this case, the set $\mathcal{S}$ is called a {\em trapping set (TS)}. 
Let $\mathcal{S}$ be a TS, and $\Gamma{(\mathcal{S})}$ be the set of neighboring CNs of $\mathcal{S}$ in the Tanner graph $G$. The \textit{induced subgraph} of $\mathcal{S}$ in $G$, denoted by $G(\mathcal{S})$, is a graph whose nodes and edges are $\mathcal{S} \cup \Gamma{(\mathcal{S})}$ and $\{v_i c_j \in E : v_i \in \mathcal{S}, c_j \in \Gamma{(\mathcal{S})}\}$, respectively. The set of CNs, $\Gamma{(\mathcal{S})}$, can be partitioned into even-degree CNs, $\Gamma_{e}{(\mathcal{S})}$, and odd-degree CNs, $\Gamma_{o}{(\mathcal{S})}$. The members of $\Gamma_{o}{(\mathcal{S})}$ and $\Gamma_{e}{(\mathcal{S})}$ are referred to as \textit{unsatisfied check nodes} and \textit{missatisfied check nodes}, respectively. If all the CNs in $\Gamma{(\mathcal{S})}$ have degrees $1$ or $2$, 
the TS is called an \textit{elementary TS (ETS)}. An ETS $\mathcal{S}$ is referred to as a {\em leafless ETS (LETS)} if each variable node in $\mathcal{S}$
is connected to at least two missatisfied CNs. Leafless ETSs are known to be the most harmful TSs over the AWGN channel~\cite{But_SS,Milen}. 
Similar to \cite{But_SS}, we thus limit our discussions in this paper to LETSs. In the rest of the paper, for simplicity, sometimes, we use the term ``TS'' instead of ``LETS.''

Trapping sets are often identified by their size, $|\mathcal{S}|=a$, and the number of unsatisfied CNs in their subgraph, $|\Gamma_{o}{(\mathcal{S})}| = b$. 
In this case, the TS is said to belong to the $(a,b)$ {\em class}, or is referred to as an $(a,b)$ TS. 
\section{Linear State-Space Model of LETSs for SPA with Flooding Schedule} \label{sec2}

\subsection{The Model} 
%A \emph{state-space} is a mathematical model including a set of input, output and state variables related by differential or difference equations which represent a physical system. State-space refers to the space in which the state of the system can be represented. The axes of this space are the state variables.

%The message passing algorithm inside a trapping set can be expressed based on a state-space model. This model, due to the non-linearities involved is so complex to analyse. However, in the error floor region when an LETS failure happens, following the dynamic of the failure mentioned before, the model can be simplified to a linear state-space model.  

%By treating an LETS as a physical system and the messages flowing in different edges of the LETS subgraph as state variables, a linear model for analysis of the LETSs has been introduced by Butler, \cite{But_SS}, which is essentially a refined and extended version of the proposed models by Sun \cite{Sun_phd} and Schlegel and Zhang \cite{Schleg}. In the following, this linear model is briefly explained. 

%The state-space model representing an ${(a, b)}$ LETS contains two input vectors. The first vector is ${\L^{(\ell)}_{ex}}$ with ${b}$ entries whose elements are the messages from the unsatisfied (degree one) CNs at iteration ${\ell}$. The second input vector denoted by ${\L}$ is the intrinsic information provided by the channel whose elements are the log-likelihood ratios. For AWGN channel, this vector includes ${a}$ independent and identically distributed ${(i.i.d.)}$ Gaussian random variables.  
The linear state-space model of a LETS is given by~\cite{But_SS}
\begin{IEEEeqnarray}{lCl"s}\label{ss_eq_flood1}
\mathbf{x}^{(0)}=\mathbf{B}\L\\\label{ss_eq_flood2}
\mathbf{x}^{(\ell)}=\bar{g}'_\ell\mathbf{A}\mathbf{x}^{(\ell-1)}+\mathbf{B}\L+\mathbf{B}_{ex}\L^{(\ell)}_{ex}, & \text{\ for} \ \ell\geq 1\:,\\\label{ss_eq_flood3}
\tilde{\L}^{(\ell)}=\bar{g}'_\ell\mathbf{C}\mathbf{x}^{(\ell-1)}+\L+\mathbf{D}_{ex}\L^{(\ell)}_{ex}, & \text{\ for} \ \ell\geq 1\:.
\end{IEEEeqnarray}
In the above equations, the vector $\mathbf{x}^{(\ell)}$ represents the state vector at iteration $\ell$. The elements of this
vector, called state variables, are the LLR messages passed over different edges of the LETS subgraph towards the missatisfied CNs. 
For an  ${(a,b)}$ LETS, the number of state variables is equal to
${m_s=\sum_{i=1}^{a} d_{v_i}-b}$, where $d_{v_i}$ is the degree of the $i\text{th}$ variable node in the LETS. 
% The major part of the above equations is the updating mechanism of the state vector $\mathbf{x}^{(\ell)}$. The elements of $\mathbf{x}^{(\ell)}$ are the LLR messages flowing in different edges of the TS subgraph toward the mis-satisfied CNs. Hence, the number of state variables for an ${(a,b)}$ LETS is ${m_s=\sum_{i=1}^{a}d_{v_i}-b}$ where $d_{v_i}$ is the $i\text{th}$ variable node degree. The state variables are updated per each iteration of the MP algorithm. 
The state variables are initialized in (9), and then updated in each iteration using (10). The relationship between state variables in consecutive iterations are established through the $ m_s\times m_s $ matrix \textbf{A}, referred to as the \emph{transition matrix}. The model has two input vectors ${\L^{(\ell)}_{ex}}$ and ${\L}$ with sizes 
$b$ and $a$, respectively. The elements of these vectors are  the messages from the unsatisfied (degree-one) CNs at iteration ${\ell}$ and the channel LLRs, respectively.
Matrices  $\bf{B}$ and ${\bf B}_{ex}$ are responsible for the contribution of channel LLRs and unsatisfied CN messages to state variables, respectively.
The $a \times 1$ output vector ${\tilde{\L}^{(\ell)}}$ represents the total LLR values of the LETS variable nodes, given in \eqref{TotLLR}. The contributions of 
the state variables and unsatisfied CN messages to this vector is accounted for through matrices $\mathbf{C}$ and $\mathbf{D}_{ex}$, respectively.
%
% The contribution of the $ a\times 1 $ channel LLR vector, $\L$, and the $ b\times 1 $ unsatisfied CNs input vector at each iteration, $\L^{(\ell)}_{ex}$, in calculating the state variables are determined by the $ m_s\times a $ matrix \textbf{B} and $ m_s\times b $ matrix $\textbf{B}_{ex}$, respectively.  The column vector ${\tilde{\L}_\ell}$ is the model output whose entries are the soft information given in \eqref{TotLLR}. Also, the $a\times m_s$ matrix $\mathbf{C}$ and the $a\times b$ matrix $\mathbf{D}_{ex}$ determine the relation of state variables and unsatisfied CN inputs with the total LLR vector, ${\tilde{\L}_\ell}$, respectively.
%
 The parameter $\bar{g}'_\ell$ is a multiplicative gain to account for the effect of the external messages entering the missatisfied CNs. 
 The calculation of $\bar{g}'_\ell$ is discussed later in Subsection~\ref{subsec3.3}.

%The following example is given to clarify the matrices involved in the model.
\begin{figure}
\centering
\includegraphics[width=1.7in]{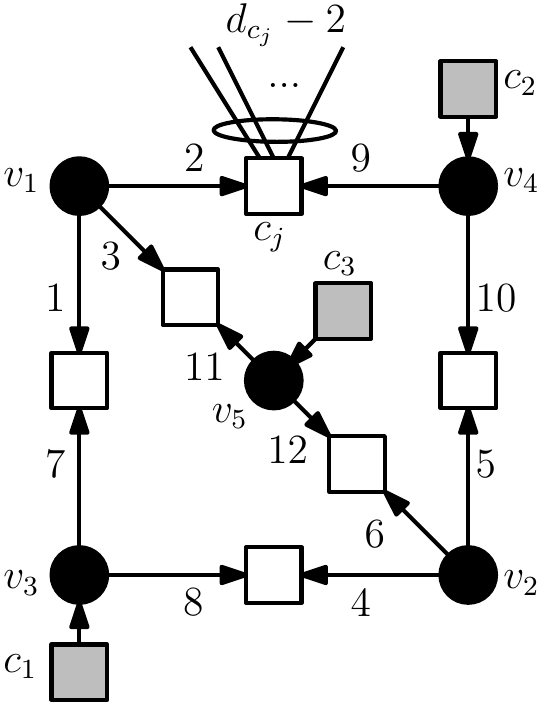}
\caption{A $(5,3)$ LETS. The VNs, missatisfied CNs and unsatisfied CNs are shown by black circles, white squares and gray squares, respectively. 
The external connections for one of the missatisfied CNs are also shown.}
\label{(5,3)external}
\end{figure} 

\begin{ex}\label{Example_53_flood}
In Fig. \ref{(5,3)external}, the structure of $(5,3)$ LETSs of Tanner (155, 64) code is shown. 
The matrices in the linear state-space model for this LETS structure, corresponding to the edge and node labels in Fig. \ref{(5,3)external}, are the followings:
%In order to obtain the matrices of the model, all the edges of the LETS subgraph are labelled, first. Then, it is assumed that the sent messages from the VNs are not affected by the mis-satisfied CNs and directly passed to their neighbouring VNs. The model matrices related to Fig. \ref{(5,3)external} is given in the following. Regarding the matrix $\mathbf{A}$, since the effect of mis-satisfied CNs are ignored, one can write the VN to mis-satisfied CN messages in the $l_{th}$ iteration as the addition of the VN to mis-satisfied CN messages of the neighbouring VNs in the previous iteration. So, for instance, the message 1 can be calculated by adding messages 9 and 11. As a result, the 9th and 11th column of the first row of $\mathbf{A}$ become 1 and so on. The same approach is also taken to obtain the other matrices. For example, the channel LLR of the VN $v_1$ is used to obtain messages number 1, 2 and 3 meaning that the rows number 1, 2 and 3 in the first column of $\mathbf{B}$ are equal to 1. Also, the ones at the rows number 6 and 7 of the first column of $\mathbf{B}_{ex}$ implies that the first unsatisfied CN messages contributes to the states number 7 and 8.  

\begin{small}

\[ \mathbf{A}=\left[
\begin{array}{c c c c c c c c c c c c}
0& 0 &0 &0 &0& 0 &0& 0& 1 &0& 1 &0 \\
0& 0& 0 &0& 0&0 &1& 0& 0& 0& 1& 0\\
0& 0& 0 &0& 0& 0 &1 &0 &1 &0 &0& 0\\
0 &0& 0& 0& 0& 0 &0& 0& 0 &1& 0& 1\\

0& 0 &0 &0 &0& 0& 0 &1& 0 &0 &0 &1\\
0& 0 &0& 0& 0 &0 & 0&1 &0& 1 &0& 0\\
0 &0& 0& 1 &0 &0& 0 &0 &0&0 &0& 0\\
1 &0 &0& 0& 0& 0 &0 &0& 0 &0 &0 &0\\

0& 0 &0 &0& 1& 0 &0 &0 &0& 0 &0 &0\\
0& 1 &0 &0& 0& 0 &0& 0& 0& 0& 0& 0\\
0 &0& 0 &0&0& 1& 0& 0 &0 &0 &0 &0\\
0& 0& 1& 0& 0 &0 &0 &0 &0& 0& 0& 0\\
\end{array}
\right],  \mathbf{B}=\left[ \begin{array}{ccccc}
1 & 0 & 0 & 0 & 0  \\
1 & 0 & 0 & 0 & 0  \\
1 & 0 & 0 & 0 & 0  \\
0 & 1 & 0 & 0 & 0  \\
0 & 1 & 0 & 0& 0   \\
0 & 1 & 0 & 0 & 0  \\
0 & 0 & 1 & 0 & 0  \\
0 & 0 & 1 & 0 & 0  \\
0 & 0 & 0 & 1 & 0  \\
0 & 0 & 0 & 1 & 0  \\
0 & 0 & 0 & 0 & 1  \\
0 & 0 & 0 & 0 & 1  \\
\end{array} \right],
\mathbf{B}_{ex}=\left[ \begin{array}{ccc}
0 & 0 & 0 \\
0 & 0 & 0 \\
0 & 0 & 0 \\
0 & 0 & 0 \\
0 & 0 & 0 \\
0 & 0 & 0 \\
1 & 0 & 0 \\
1 & 0 & 0 \\
0 & 1 & 0 \\
0 & 1 & 0 \\
0 & 0 & 1 \\
0 & 0 & 1 \\

\end{array} \right],\]
\[
\mathbf{C}=\left[\begin{array}{cccccccccccc}
0 & 0 & 0 & 0 & 0 & 0 & 1 & 0 & 1 & 0 & 1 & 0 \\
0 & 0 & 0 & 0 & 0 & 0 & 0 & 1 & 0 & 1 & 0 & 1 \\
1 & 0 & 0 & 1 & 0 & 0 & 0 & 0 & 0 & 0 & 0 & 0 \\
0 & 1 & 0 & 0 & 1 & 0 & 0 & 0 & 0 & 0 & 0 & 0 \\
0 & 0 & 1 & 0 & 0 & 1 & 0 & 0 & 0 & 0 & 0 & 0 \\

\end{array} \right],
\mathbf{D}_{ex}=\left[\begin{array}{ccc}
0 & 0 & 0 \\
0 & 0 & 0 \\
1 & 0 & 0 \\
0 & 1 & 0 \\
0 & 0 & 1 \\
\end{array} \right]\:.
\]
\end{small}
As an example, the first row of matrix $\mathbf{A}$ has a `1' in columns $9$ and $11$. By (\ref{ss_eq_flood2}), this means that the 
first state variable ($x_1$) at iteration $\ell$ is a function of state variables $x_9$ and $x_{11}$ at iteration $\ell -1$. This relationship between the state variables 
at consecuative iterations can also be seen in Fig.~\ref{(5,3)external}, where in this case, $x_1$, as an outgoing message from $v_1$, is a function of the extrinsic state variables that are incoming messages to the missatisfied check nodes connected to $v_1$. As another example, from Fig.~\ref{(5,3)external}, the channel LLR of the VN $v_1$ contributes to messages $x_1$, $x_2$ and $x_3$. This, based on (\ref{ss_eq_flood2}), means that the elements in the first column of $\mathbf{B}$ and rows $1$, $2$ and $3$  
are equal to $1$. As the final example, the ones in the first column of $\mathbf{B}_{ex}$, which are located in rows $7$ and $8$,  imply that 
the first unsatisfied CN message contributes to state variables $x_7$ and $x_8$. This can also be seen from Fig.~\ref{(5,3)external}.
\label{Example1_53flood}  
\end{ex}

\subsection{Application of Density Evolution in the State-Space Model}  
In order to analyze TS failures based on the linear state-space model, one needs to obtain the probability distribution of the messages entering the TS subgraph via unsatisfied CNs at different iterations, $\L^{(\ell)}_{ex}$, as well as the distribution of the messages from the external VNs connected to the missatisfied CNs. 
The latter distribution is then used to obtain the linear gains $\bar{g}'_\ell$.

To obtain such distributions, the authors of \cite{Sun_phd, Schleg, But_SS} use density evolution (DE)~\cite{Urbank} with the underlying assumptions that 
the all-zero codeword is transmitted, and that the surrounding neighborhood of the TS is tree-like. This implies that all the messages entering the TS through unsatisfied and missatisfied CNs are independent in each iteration as well as in subsequent iterations.

\subsection{Missatisfied CN Gain Model}
\label{subsec3.3}
%In Fig. \ref{(5,3)external}, the external connections of the mis-satisifed CN $c_j$ between the VNs $v_1$ and $v_4$ are shown. In the first state space model, proposed by Sun \cite{Sun_phd}, the external connections are neglected meaning that, for instance, the messages being sent from ${v_1}$ (or ${v_4}$) are directly passed to ${v_4}$ (or ${v_1}$) without being affected by the ${d_{c_j}-2}$ messages coming from the variable nodes outside of the LETS. This is justified by the assumption that the LLR values outside the TS has correct sign and their growth rate are much higher than the internal LLRs. However, in \cite{Schleg, But_SS}, it was shown that these messages, in particular at the early iterations, can impact the internal messages. Therefore, they improved the model by proposing the linear multiplicative gain, ${0 \leq \bar{g}'_\ell \leq 1}$, to account for the influence of the external VNs incident to the degree-two CNs at different iterations. 

To calculate the missatisfied CN gains, without loss of generality, consider the missatisfied CN $c_j$ in Fig. \ref{(5,3)external}.
The message passed from $c_j$ to $v_4$ is given by
%Considering the VNs $v_1$ and $v_4$ in Fig. \ref{(5,3)external} and the CN $c_j$ in between, the $c_j$ to $v_4$ message can be calculated as
\begin{equation}
{L_{\ell}^{[4 \leftarrow j]}=2\tanh^{-1}\Bigg[g^{[j]}_\ell\tanh\frac{{L_{\ell}^{[1 \rightarrow j]}}}{2}\Bigg]},
\label{cjv4}
\end{equation}  
where $g^{[j]}_\ell$ is defined as
\begin{equation}
g^{[j]}_\ell={\prod_{k \in N(j)\backslash \{1,4\} }\tanh\frac{{L_{\ell}^{[k \rightarrow j]}}}{2}}.
\label{g_j_l}
\end{equation} 
For relatively small values of ${L_{\ell}^{[1 \rightarrow j]}}$ compared to $g^{[j]}_\ell$, which typically happens in the event of a TS failure, 
the Taylor expansion of order two for \eqref{cjv4} around zero is given by
\begin{equation}
{L_{\ell}^{[4 \leftarrow j]}}\approx g^{[j]}_\ell L_{\ell}^{[1 \rightarrow j]}.
\end{equation}
(Note that the square term in the Taylor expansion is equal to zero.)
The above model is further simplified in \cite{Schleg, But_SS} by replacing $g^{[j]}_\ell $ with an average gain, $\bar{g}_\ell$, over all possible channel noise realizations as well as all the CNs. For an LDPC code with regular CN degree of $d_c$, the expected gain is
\begin{equation}
 \bar{g}_\ell=\mathbb{E}_{\mathbf{n},j,k}\Bigg[\tanh\frac{{L_{\ell}^{[k \rightarrow j]}}}{2}\Bigg]^{d_c-2}\:,
 \label{g_bar}
\end{equation}
where $\mathbf{n}$ is the channel noise vector. The average, $\mathbb{E}_{\mathbf{n},j,k}$, is taken over all CNs $j$, all $k \in N(j)$ and all noise realizations. 
To calculate (\ref{g_bar}), one can simply calculate the expected value with respect to the probability distribution of VN to CN messages at iteration $\ell$, where this distribution is a function of degree distributions of the code as well as the channel noise distribution.  

In \cite{Schleg, But_SS}, the gain model of Equation \eqref{g_bar} is extended by incorporating the polarity inversions of the LETS internal messages to the model. These inversions are caused by erroneous LLRs entering the missatisfied CNs from the $d_c-2$ external VNs. In other words, assuming the all-zero codeword is transmitted, whenever an odd number of external messages, out of $d_c-2$, have negative signs, the polarity of the internal messages from missatisfied CNs to their neighboring VNs will be inverted. The probability of inversion, $P_{inv,\ell}$, is thus calculated by
\begin{equation}
P_{inv,\ell}=\sum_{k \text{\ odd}}\binom {d_c-2} k P_{e,\ell}^k(1-P_{e,\ell})^{d_c-2-k}\:,
\end{equation}     
where $P_{e,\ell}$ is the probability of error in each of the VN to CN messages at iteration $\ell$, and is calculated using DE.
In~\cite{But_SS}, to model the inversion, the authors suggested using the following average gain instead of \eqref{g_bar}:
\begin{equation}
\bar{g}'_\ell=(1-P_{inv,\ell})\bar{g}_\ell.
\end{equation}
\subsection{Error Probability of a LETS Structure} 
Before the calculation of error probability, we discuss some of the properties of the transition matrix $\mathbf{A}$ relevant to the calculation.

%\begin{mydef}
For a given square matrix $\mathbf{M}$ and a permutation matrix $\mathbf{P}$, the matrix $\mathbf{PMP}^T$ is called the {\em symmetric permutation} of $\mathbf{M}$.   
%\end{mydef}
If $\mathbf{M}$ represents the adjacency matrix of a graph, symmetric permutation results in an isomorphic graph whose nodes are relabeled. 
The eigenvalues of $\mathbf{M}$ and its symmetric permutation are the same. Also, their eigenvectors, up to a permutation, are equal \cite{Meyer}. 
\begin{mydef}\label{perm_irred_mydef}
A non-negative $n \times n$ real matrix $\mathbf{M}$ %where all the entries $m_{ij} \geq 0$ for $1 \leq i,j \leq n$ 
is \emph{irreducible} if it cannot be symmetrically permuted by any permutation matrix $\mathbf{P}$ into a block upper triangular matrix, i.e., 
\begin{equation}
\label{Perm_irreducible_def}
\mathbf{PMP}^T \neq \left [\begin{array}{c|c}
\mathbf{M}_a &\mathbf{M}_{c}\\
\hline
\mathbf{0}& \mathbf{M}_b
\end{array} \right ],
\end{equation} 
where $\mathbf{M}_a$ and $\mathbf{M}_b$ are square matrices with sizes greater than 0 (non-trivial).
The matrix $\mathbf{M}$ is \emph{reducible} if it is not irreducible.
\end{mydef}
%\begin{mydef}
The \emph{spectral radius} of a square matrix $\mathbf{M}$ is defined as the largest absolute value of its eigenvalues, and is denoted by $\rho (\mathbf{M})$.
%\end{mydef}
\begin{theo}[Perron-Frobenius theorem of non-negative irreducible matrices]\label{Perron_Frobenius_theo}
Let $\mathbf{M}$ be a non-negative irreducible matrix. Then,
\begin{itemize}
\item[(a)]$\mathbf{M}$ has a positive real eigenvalue equal to its spectral radius $\rho(\mathbf{M})$.
\item[(b)]corresponding to $\rho(\mathbf{M})$, matrix $\mathbf{M}$ has a positive eigenvector $\mathbf{x}$.\footnote{This is valid for both left and right eigenvectors.}
\item[(c)]$\rho(\mathbf{M})$ is a simple (multiplicity equal to 1) eigenvalue of $\mathbf{M}$.
\item[(d)]$\mathbf{x}$ is the only non-negative eigenvector of $\mathbf{M}$.
\end{itemize}

\end{theo}
%\begin{rem}
Almost all the LETS structures, with the exception of simple cycles,  have a non-negative irreducible transition matrix $\mathbf{A}$~\cite{But_SS}.
Hence, based on Theorem \ref{Perron_Frobenius_theo}, there exists a positive dominant eigenvalue of $\mathbf{A}$, $r$, 
whose corresponding left and right eigenvectors, $\mathbf{w}_1^T$ and $\mathbf{u}_1$, are both positive. It is shown in \cite{But_SS}
that for a simple cycle, $r=1$, and for LETSs that are not simple cycles, $r > 1$.
%\end{rem}

%Given the irreducible non-negative square matrix $\mathbf{M}$, the \textit{Perron-Frobenius} theory implies that the spectral radius $\rho(\mathbf{M})=r$ in which $r>0 $ is an eigenvalue of $\mathbf{M}$. Also, the left and right eigenvectors, $\mathbf{w}_1$ and $\mathbf{u}_1$, associated with $r$ are positive and there are not any other left or right eigenvectors of $\mathbf{M}$ with non-negative entries \cite{horn,varga}.
%\begin{mydef}
%
%The irreducible non-negative matrices are divided into two main categories. If such a matrix $\mathbf{M}$ has only one eigenvalue with magnitude equal to $\rho(\mathbf{M})$, then $\mathbf{M}$ is called a \emph{primitive} matrix. On the other hand, if the number of eigenvalues of $\mathbf{M}$ with magnitude $\rho(\mathbf{M})$ is $h>1$, then $\mathbf{M}$ is said to be \emph{imprimitive} with the index of imprimitivity equal to $h$. In both cases, only one of the eigenvalues with magnitude $\rho(\mathbf{M})$ is real and positive.
%
%\end{mydef}
 
%In order to obtain the probability of an LETS failure, in \cite{But_SS}, first a non-recursive form of the state variables is calculated, then, the non-negative matrix theory is utilized to obtain an error indicator function. Finally, the probability of TS failure is estimated by Gaussian approximation. In the following, these steps are briefly explained. However, our explanation of the method is close to Sun's approach \cite{Sun_phd} and the similarity of that with the Butler's method will be discussed.
 
To obtain the probability of a LETS failure, in \cite{But_SS}, the authors first derived a non-recursive equation for the state vector:
\begin{IEEEeqnarray*}{l}\label{nonRecursiveSS_flood}
\mathbf{x}^{(\ell)}=\mathbf{A}^\ell\mathbf{B}\L\prod_{j=1}^\ell\bar{g}'_j
%\\ 
 +\sum_{i=1}^\ell\mathbf{A}^{\ell-i}\big(\mathbf{B}\L+\mathbf{B}_{ex}\L^{(i)}_{ex}\big)\prod_{j=i+1}^\ell\bar{g}'_j\:.
\IEEEyesnumber
\label{state_flood_nonreq}
\end{IEEEeqnarray*}
%Now we have an equation showing the state of the messages in different branches of the TS subgraph at the ${l}$th iteration with respect to the model inputs from the channel as well as unsatisfied CNs.
They then showed that the projection of the state vector on the positive left eigenvector associated with the dominant eigenvalue $r$, given by
\begin{IEEEeqnarray*}{l}
\label{indicator_flood}
\mathbf{w}^T_1\mathbf{x}^{(\ell)}=r^\ell\mathbf{w}^T_1\mathbf{B}\L\prod_{j=1}^\ell\bar{g}'_j
%\\
+\sum_{i=1}^\ell r^{\ell-i}\mathbf{w}^T_1\big(\mathbf{B}\L+\mathbf{B}_{ex}\L_{ex}^{(i)}\big)\prod_{j=i+1}^\ell\bar{g}'_j\:,
\IEEEyesnumber
\end{IEEEeqnarray*}
can be used as an indicator of TS failure. In fact, the following scaled version of (\ref{indicator_flood}) is used in~\cite{But_SS} as the error indicator function:
\begin{equation}\label{indicator_Butler}
{\beta'_\ell=\mathbf{w}^T_1\mathbf{B}\L+\sum_{i=1}^\ell\frac{\mathbf{w}^T_1\big(\mathbf{B}\L+\mathbf{B}_{ex}\L_{ex}^{(i)}\big)}{r^i\prod_{j=1}^i\bar{g}'_j}}\:,
\end{equation}
with the probability of error for the corresponding LETS $\mathcal{S}$ given by
\begin{equation}\label{BetaProbFail}
{P_e(\mathcal{S})=\lim_{\ell\to\infty}\text{Pr}\{\beta'_\ell<0\}=\lim_{l\to\infty}Q\bigg(\frac{\mathbb{E}[\beta'_\ell]}{\sqrt{\mathbb{VAR}[\beta'_\ell]}}\bigg),}
\end{equation}
where $Q(x)=\frac{1}{\sqrt{2 \pi}}\int_{x}^{\infty}\exp\big (-\frac{u^2}{2}\big ) du$. In the derivation of (\ref{BetaProbFail}), it is assumed that 
the random variable $\beta'_\ell$ is Gaussian.

\section{Linear State-Space Model of LETSs for SPA with Row Layered Schedule}\label{sec3}
In this section, we develop a linear state-space model of LETSs for SPA with row layered schedule, and use the model to calculate 
the error probability of LETSs. In Subsection~\ref{subsec4.1}, by a proper labeling of state variables, we establish a relationship between the matrices $\mathbf{A}$, $\mathbf{B}$ and $\mathbf{B}_{ex}$ of the model for flooding schedule in (\ref{ss_eq_flood1})-(\ref{ss_eq_flood3}) and the corresponding matrices needed in the model for the layered schedule. The recursive and non-recursive equations for the state vector of the layered schedule are then derived in Subsection~\ref{subsec4.2}. In Subsection~\ref{subsec4.3}, we present the application of DE within the linear state-space model of the layered decoder. Within this subsection, we also introduce the concept of TS layer profile which plays an important role in the proper application of DE in the model, and in identifying TSs that have the same topology but may have different harmfulness. We then derive the gain values of the missatisfied CNs for a layered schedule in Subsection~\ref{gain_lay_section}. The spectral properties of the LETS system matrices in layered decoders are analyzed in Subsection \ref{SpectralProp_SubSec}. 
%Finally, in Subsection~\ref{ProbErrLay}, we calculate the error probability of a LETS using the linear state-space model.

\subsection{Relationship between model matrices of flooding and row layered schedules} \label{subsec4.1}

%As it is mentioned before, the main difference between the flooding and layered decoders is the updating order of the message passing algorithm. While in the standard schedule the reliabilities are updated once per iteration and used in the next iteration, the check node groups (layers) of the row layered decoders are activated one by one within an iteration meaning that the updated reliabilities can be utilized within the same iteration by the following layers. 
To establish a relationship between the matrices that appear in the two models, it is helpful to label the state variables in a certain order.
Consider an $(a,b)$ LETS ${\cal S}$ with $m_s$ state variables $x_1,\ldots,x_{m_s}$.  Suppose that the missatisfied CNs of ${\cal S}$ are from $J$ different layers of the 
parity-check matrix ${\bf H}$, where $J \leq m_b$. We denote these layers by $L_1, \ldots, L_J$, where an smaller index for a layer implies that the CNs in that layer are updated earlier in an iteration. In the following, we say that such a LETS has $J$ layers. We use the notation $n_{L_j}$ to denote the number of state variables that are updated in layer $L_j$, for $1 \leq j \leq J$. We thus have 
$n_{L_1}+\cdots+n_{L_J}=m_s$. To assign the state variables to different internal messages of ${\cal S}$, we start with the messages that are updated in $L_1$, and assign to them variables $x_1, \ldots, x_{n_{L_1}}$. We then move on to the messages that are updated in $L_2$, and assign to them variables $x_{n_{L_1}+1}, \ldots,
x_{n_{L_2}+n_{L_1}}$. We will continue this process all the way to $L_J$ until all the $m_s$ messages have their state variables assigned to them. We call this labeling of state variables {\em systematic labeling}.

Based on the systematic labeling, the matrix ${\bf A}$ of flooding schedule will be a $J \times J$ array of matrices ${\bf A}_{i,j}, 1 \leq i \leq J, 1 \leq j \leq J$,
where the size of the matrix ${\bf A}_{i,j}$ is $n_{L_i} \times n_{L_j}$, and the $J$ diagonal matrices are all-zero, i.e., ${\bf A}_{i,i}={\bf 0}, 1 \leq i \leq J$.  Based on the 
partitioning of the state variables according to their layer, the rows of matrices ${\bf B}$ and ${\bf B}_{ex}$ can also be partitioned into $J$ blocks, with the
$i$th row block containing $n_{L_i}$ rows.  We refer to this representation of matrices as {\em systematic form}.

%\begin{figure}
%\centering
%\includegraphics[width=1.8in]{53TannerLayersorderedEX.pdf}
%\caption{A $(5,3)$ LETS of the Tanner $(155, 64)$ code whose edges are labeled in a systematic form. The edges of different layers are distinguished by different colors and line types.}
%\label{(5,3)lay_col}
%\end{figure}

\begin{ex}
\label{ex2}
Consider the $(5,3)$ LETSs of the Tanner $(155,64)$ code discussed in Example~\ref{Example_53_flood}. The Tanner code has three row layers, i.e., $m_b =3$.
Each of the $(5,3)$ LETSs has $3$ unsatisfied and $6$ missatisfied check nodes. Each of the three unsatisfied check nodes belongs to a different 
row layer. Out of $6$ missatisfied check nodes, each set of two belongs to a different row layer ($J=3$). This is shown in  Fig. \ref{(5,3)lay_col}.
To distinguish the layers, different colors and line types are used in Fig. \ref{(5,3)lay_col}. As can be seen, the selection of edge labels, which reflects the 
indices of corresponding state variables, are in systematic form. The corresponding matrices are given by:  
\end{ex}
\begin{small}
 \begin{IEEEeqnarray*}{lCl"s}
\mathbf{A}=
\left[
\begin{array}{c|c|c}
\mathbf{0} &\mathbf{A}_{1,2}&\mathbf{A}_{1,3}\\
\hline
\mathbf{A}_{2,1}& \mathbf{0}& \mathbf{A}_{2,3}\\
\hline
\mathbf{A}_{3,1}& \mathbf{A}_{3,2}&\mathbf{0}
\end{array}
\right]
=\left[
\begin{array}{c c c c|c c c c|c c c c}
0& 0 &0 &0 &0& 0 &0& 1& 0 &1& 0 &0 \\
0& 0& 0 &0& 0&0 &0& 0& 0& 0& 0& 1\\
0& 0& 0 &0& 0& 0 &1 &0 &0 &0 &0& 0\\
0 &0& 0& 0& 1& 0 &0& 0& 0 &0& 1& 0\\
\hline
0& 0 &0 &0 &0& 0& 0 &0& 1 &0 &0 &0\\
0& 0 &1& 0& 0 &0 &0 &0 &0& 0 &1& 0\\
0 &1& 0& 0 &0 &0& 0 &0 &0&1 &0& 0\\
0 &0 &0& 1& 0& 0 &0 &0& 0 &0 &0 &0\\
\hline
0& 1 &0 &0& 0& 0 &0 &1 &0& 0 &0 &0\\
0& 0 &0 &0& 0& 1 &0& 0& 0& 0& 0& 0\\
1 &0& 0 &0&0& 0& 0& 0 &0 &0 &0 &0\\
0& 0& 1& 0& 1 &0 &0 &0 &0& 0& 0& 0\\
\end{array}
\right]\:,
\end{IEEEeqnarray*} 
%In this paper, the above representation of matrix $\mathbf{A}$ is called \emph{systematic form}. The reason for zero diagonal blocks is that the non-zero circulants of the single-edge quasi-cyclic codes are made of  shifted version of identity matrices. As a result, a state variable belonging to a specific layer cannot be a function of other state variables within the same layer. Otherwise, a VN must be connected to more than one CN from that layer which is a direct contradiction to the structure of the quasi-cyclic codes. It should be noted that the square zero matrices on the main diagonal, depending on the number of state variables within the same layer, $n_{L_j}$ for $1\leq j \leq J$, can have different sizes, $n_{L_j} \times n_{L_j}$. Besides, the other non-zero sub-blocks are not necessarily all square matrices.  

%\begin{rem}
%To apply our proposed method in this paper, generally, it is not required to obtained the systematic form of the system matrices and this is only done for the presentation and analysis purposes in the following sections.
%\end{rem}
%
%The matrices $\mathbf{B}$ and $\mathbf{B}_{ex}$ are also partitioned to $J$ row blocks. In our example, these matrices have the following structures
 \begin{IEEEeqnarray*}{lCl"s}
\mathbf{B}=
\left[
\begin{array}{c}
\mathbf{B}_1\\
\hline
\mathbf{B}_2\\
\hline
\mathbf{B}_3
\end{array}
\right]
=\left[
\begin{array}{c c c c c}
1& 0 &0 &0 &0 \\
0& 0 &0 &1 &0\\
0& 0 &0 &0 &1\\
0& 1 &0 &0 &0\\
\hline
0& 0 &1 &0 &0\\
0& 1 &0 &0 &0\\
1& 0 &0 &0 &0\\
0& 0 &0 &0 &1\\
\hline
1& 0 &0 &0 &0\\
0& 0 &1 &0 &0\\
0& 0 &0 &1 &0\\
0& 1 &0 &0 &0\\
\end{array}
\right],
%\end{IEEEeqnarray*}
% \begin{IEEEeqnarray*}{lCl"s}
\mathbf{B}_{ex}=
\left[
\begin{array}{c}
\mathbf{B}_{ex,1}\\
\hline
\mathbf{B}_{ex,2}\\
\hline
\mathbf{B}_{ex,3}
\end{array}
\right]=\left[
\begin{array}{c c c}
0& 0 &0  \\
0& 1 &0 \\
0& 0 &1 \\
0& 0 &0 \\
\hline
1& 0 &0 \\
0& 0 &0 \\
0& 0 &0 \\
0& 0 &1 \\
\hline
0& 0 &0 \\
1& 0 &0 \\
0& 1 &0 \\
0& 0 &0 \\
\end{array}
\right].
\end{IEEEeqnarray*} 
\end{small}
%\subsection{The Linear Model Matrices of the Horizontal Layered Decoder}\label{LayMatSec}
\begin{figure}
\centering
\includegraphics[width=1.8in]{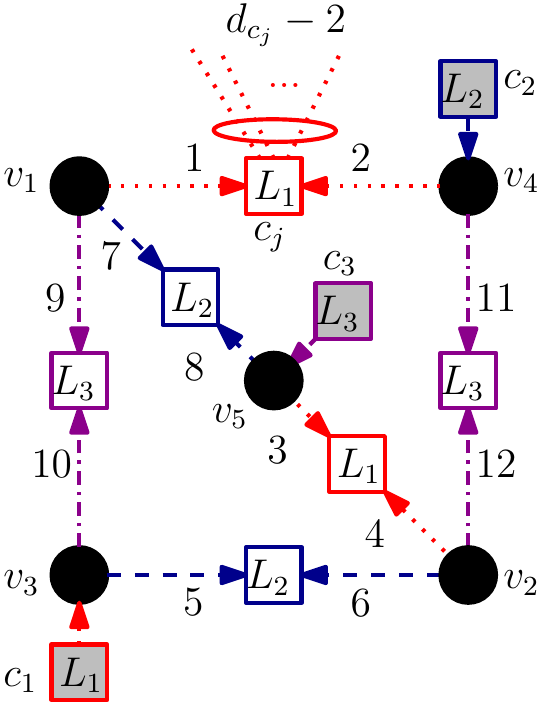}
\caption{A $(5,3)$ LETS of the Tanner $(155, 64)$ code whose edges are labeled in a systematic form. The edges of different layers are distinguished by different colors and line types.}
\label{(5,3)lay_col}
\end{figure}
In a layered schedule, the state variables of a LETS are updated in $J$ rounds, each round corresponding to one row layer, within one iteration.
Corresponding to the updating of each layer, there thus exists a set of matrices used in the model. The $m_s \times m_s$ matrix 
that represents the relationship among the state variables corresponding to the updating of layer $j$ is denoted by $\mathcal{A}_j, 1 \leq j \leq J$,
and is referred to as the ``transition matrix of layer $j$.''
Matrix $\mathcal{A}_j$ has the same block structure as matrix ${\bf A}$. In fact, the $j$th row block of the two matrices are identical. 
Matrix  $\mathcal{A}_j$, however, has zero blocks everywhere else other than the diagonal blocks that are each an identity matrix.  
This is to indicate that as the messages in the $j$th layer are updated, all the other messages in the LETS remain unchanged.
Similarly, notation $\mathcal{B}_j$ is used to denote the $m_s \times a$ matrix that is responsible for the contribution of channel LLRs 
in the state variables that are updated in layer $j$. This matrix has the same row block structure as in ${\bf B}$, with the 
difference that except for row block $j$ (consisting of $n_{L_j}$ rows) that is identical to that of ${\bf B}$, all the other row blocks are zero.
Similarly, for the contribution of unsatisfied CN messages to the state variables of layer $j$, the $m_s \times b$ matrix $\mathcal{B}_{ex,j}$ 
can be defined as a matrix with $J$ row blocks whose $j$th row block is equal to that of $B_{ex}$ while the other row blocks are 
all zero. This matrix, however, is not directly utilized in the model. The reason is that, two groups of unsatisfied CN messages contribute to the state variables in layer $j$
at iteration $\ell$. The first group are those that are updated at the end of $(\ell-1)$th iteration, $\L_{ex}^{(\ell-1)}$. The second group are the ones 
that are updated earlier in iteration $\ell$, $\L_{ex}^{(\ell)}$. Correspondingly, the matrices $\overset{\triangleleft}{\mathfrak{B}}_{ex,j}$ and $\overset{\triangleright}{\mathfrak{B}}_{ex,j}$ are defined to account for the two contributions, respectively, and we have
% illustrate the contribution of $\L_{ex}^{(\ell-1)}$ and $\L_{ex}^{(\ell)}$ in updating the state variables within the $j$th layer, respectively. The following relation holds for these matrices
\begin{equation}
\overset{\triangleleft}{\mathfrak{B}}_{ex,j}+\overset{\triangleright}{\mathfrak{B}}_{ex,j}=\mathcal{B}_{ex,j}\:.
\end{equation} 
%In this section, following the systematic representation of the model matrices of standard MP in the previous section, the required matrices of the linear model corresponding to the layered MP in an LETS are introduced. 

%At each layer of the MP algorithm, only, certain number of CNs are activated. This means while the state variables related to the active layers are updated, the others remained unchanged. Since the state variables in flooding and layered schedule algorithms are considered to be the same, the transition matrix of the flooding decoder can be used to obtain the layered decoder transition matrix. In fact, the system matrix at the $j$th layer of the layered decoder, $\mathcal{A}_j$, is equal to an $m_s \times m_s$ identity matrix whose $j$th row block is replaced by the $j$th row block of the matrix $\mathbf{A}$. For example, the system matrix at the 2nd layer of the MP algorithm in Fig. \ref{(5,3)lay_col} is
\begin{ex}
\label{ex3}
For the same LETS structure discussed in Example~\ref{ex2}, the model matrices corresponding to the second layer of row layered schedule are the followings:
\begin{small}
\begin{IEEEeqnarray*}{lCl"s}
 \mathbcal{A}_2=
\left[
\begin{array}{c|c|c}
\mathbf{I} &\mathbf{0}&\mathbf{0}\\
\hline
\mathbf{A}_{2,1}& \mathbf{0}& \mathbf{A}_{2,3}\\
\hline
\mathbf{0}& \mathbf{0}&\mathbf{I}
\end{array}
\right]\:,
%\\
%\\ \\
%=\left[
%\begin{array}{c c c c|c c c c|c c c c}
%1& 0 &0 &0 &0& 0 &0& 0& 0 &0& 0 &0 \\
%0& 1& 0 &0& 0& 0 &0& 0& 0& 0& 0& 0\\
%0& 0& 1 &0& 0& 0 &0 &0 &0 &0 &0& 0\\
%0 &0& 0& 1& 0& 0 &0& 0& 0 &0& 0& 0\\
%\hline
%0& 0 &0 &0 &0& 0& 0 &0& 1 &0 &0 &0\\
%0& 0 &1& 0& 0 &0 &0 &0 &0& 0 &1& 0\\
%0 &1& 0& 0 &0 &0& 0 &0 &0&1 &0& 0\\
%0 &0 &0& 1& 0& 0 &0 &0& 0 &0 &0 &0\\
%\hline
%0& 0 &0 &0& 0& 0 &0 &0 &1& 0 &0 &0\\
%0& 0 &0 &0& 0& 0 &0& 0& 0& 1& 0& 0\\
%0 &0& 0 &0&0& 0& 0& 0 &0 &0 &1 &0\\
%0& 0& 0& 0& 0 &0 &0 &0 &0& 0& 0& 1\\
%\end{array}
%\right].
%\end{IEEEeqnarray*}
%The symbol $\mathbf{I}$ is used to represent identity matrix. But the size of the identity matrices at different layers, depending on the number of state variable in different layers, can be different. Presence of the $\mathbf{I}$ implies that the corresponding state variables are not updated and remain unchanged.
%The matrix $\mathcal{B}_j$ is an $m_s \times a$ matrix that indicates the contribution of the channel LLRs in calculation of the state variables of the $j$th layer. It, basically, consists of $J$ row blocks in which the $j$th layer is equal to the $j$th row block of matrix $\mathbf{B}$ and the others are equal to zero. For instance, the matrix $\mathcal{B}_2$ of Fig. \ref{(5,3)lay_col} has the following structure
%\begin{IEEEeqnarray*}{lCl"s}
\mathfrak{B}_2=
\left[
\begin{array}{c}
\mathbf{0}\\
\hline
\mathbf{B}_{2}\\
\hline
\mathbf{0}
\end{array}
\right]\:,
\overset{\triangleleft}{\mathfrak{B}}_{ex,2}=
\left[
\begin{array}{c c c}
0& 0 &0  \\
0& 0 &0 \\
0& 0 &0 \\
0& 0 &0 \\
\hline
0& 0 &0 \\
0& 0 &0 \\
0& 0 &0 \\
0& 0 &1 \\
\hline
0& 0 &0 \\
0& 0 &0 \\
0& 0 &0 \\
0& 0 &0 \\
\end{array}
\right], \
%\end{IEEEeqnarray*}
%\begin{IEEEeqnarray*}{lCl"s}
\overset{\triangleright}{\mathfrak{B}}_{ex,2}=
\left[
\begin{array}{c c c}
0& 0 &0  \\
0& 0 &0 \\
0& 0 &0 \\
0& 0 &0 \\
\hline
1& 0 &0 \\
0& 0 &0 \\
0& 0 &0 \\
0& 0 &0 \\
\hline
0& 0 &0 \\
0& 0 &0 \\
0& 0 &0 \\
0& 0 &0 \\
\end{array}
\right].
\end{IEEEeqnarray*}  
\end{small}
In the second layer, the unsatisfied CNs $c_1$ and $c_3$ contribute to the values of state variables $x_5$ and $x_8$, respectively. While the message from $c_1$ 
is updated at the $L_1$ layer (before $L_2$), the message from $c_3$ is updated only in $L_3$ (after $L_2$). So, the contribution from $c_1$ to $x_5$ 
is reflected through $\overset{\triangleright}{\mathfrak{B}}_{ex,2}$ which uses a newly updated (at the current iteration) version of CN message, 
while the contribution from $c_3$ to $x_8$ is via $\overset{\triangleleft}{\mathfrak{B}}_{ex,2}$ and uses the CN message updated at the end of the previous iteration.
%is the state variable 5 which belongs to the $L_2$ layer uses an updated version of the unsatisfied CN input within the current iteration. However, since the unsatisfied CN $c_3$ is in $L_3$ layer, it operates after the $L_2$ layer and the message number 8 uses the unsatisfied input updated at the end of the previous iteration. As a result, the entry of the row 5 and column 1 ($c_1$) of $\overset{\triangleright}{\mathfrak{B}}_{ex2}$ as well as the entry of the row 8 and column 3 ($c_3$) of $\overset{\triangleleft}{\mathfrak{B}}_{ex2}$ are set to 1.
\end{ex}

%As it was explained in previous sections, in flooding schedule decoder, the impact of external connections of the mis-satisfied CNs on the internal messages flowing in different edges of an LETS subgraph is modelled by an iteration dependent multiplicative gain, $\bar{g}'_\ell$. Generally, this gain model must have a diagonal matrix form. However, since in \cite{But_SS,Schleg}, it is assumed all the mis-satisfied CNs within the subgraph have equal impact (equal gain) on the messages going through them, by averaging over all the mis-satisfied gains, the final diagonal gain matrix turns into the form of $\bar{g}'_\ell\mathbf{I}_{m_s\times m_s}$ that can be simplified to a scalar gain $\bar{g}'_\ell$ in derivation of the state-space equations. 

The last set of  matrices that we need in the linear state-space model of a row layered decoder are the gain matrices. In the flooding schedule,
the impact of external connections of the missatisfied CNs on the internal messages of a LETS is modelled by 
an iteration dependent scalar gain, $\bar{g}'_\ell$~\cite{But_SS}. This, in fact, corresponds to a diagonal gain matrix with
equal diagonal elements, i.e., $\bar{g}'_\ell\mathbf{I}_{m_s\times m_s}$, and assumes that all the missatisfied CNs have equal impact on the internal 
messages of the LETS. For a layered decoder, however, there can be a significant difference among the distribution of external messages of
different missatisfied CNs, depending on their layer and their external connections.
%In layered decoder, also, we use the linear estimation to model the effect of external connections of the mis-satisifed CNs. However, unlike the standard decoder, this is not a good idea to obtain an average gain over all the mis-satisfied CNs. In other words, the distribution of external messages entering to the mis-satisfied CNs, due to the serial scheduling, might be, substantially, different. This issue will be explicitly discussed in section \ref{gain_lay_section}. 
To model the missatisfied CN gains within a given layer $j$ at iteration $\ell$, we thus use the gain matrix $\mathfrak{G}_j^{(\ell)}$. 
To define $\mathfrak{G}_j^{(\ell)}$, we first consider the following $m_s\times m_s$ diagonal matrix:\begin{IEEEeqnarray*}{lCl"s}
\mathbf{G}^{(\ell)}=\left[
\begin{array}{c c c c}
\bar{g}_1^{'(\ell)}&0&\dotsb &0  \\
0&\bar{g}_2^{'(\ell)} & \ddots &\vdots\\
%0& g_2^{(l)}  &0 &\dotsb&\vdots\\
\vdots & \ddots &\ddots &0 \\
0&\dotsb &0&\bar{g}_{m_s}^{'(\ell)}
\end{array}
\right],
\end{IEEEeqnarray*}
whose diagonal entries are the linear gains corresponding to $m_s$ state variables.
% that are obtained based on the probability distribution of external messages of each individual mis-satisfied CN. 
(Note that the state variables corresponding to the same missatisfied CN have equal gains.) 
We then partition $\mathbf{G}^{(\ell)}$ into $J \times J$ block matrices. This partitioning corresponds to different layers of the LETS  and is similar 
to the partitioning of the transition matrix $\mathbf{A}$. As a result, we have
\begin{IEEEeqnarray*}{lCl"s}
\mathbf{G}^{(\ell)}=\left[
\begin{array}{c |c| c| c|c}
\mathbf{G}_1^{(\ell)}&\mathbf{0}&\mathbf{0}&\dotsb &\mathbf{0}  \\
\hline
\mathbf{0}&\mathbf{G}_2^{(\ell)}&\mathbf{0}& \dotsb &\mathbf{0}\\
\hline
\mathbf{0} &\mathbf{0} &\ddots &\ddots & \vdots\\
\hline
\vdots &\vdots &\ddots & \ddots &\mathbf{0}\\
\hline
\mathbf{0} &\mathbf{0}&\dotsb&\mathbf{0} & \mathbf{G}_J^{(\ell)}
\end{array}
\right].
\end{IEEEeqnarray*} 
Now, the gain matrix $\mathfrak{G}_j^{(\ell)}$ corresponding to layer $j$ is defined as $\mathbf{G}^{(\ell)}$ in which all the diagonal blocks, except the $j$th one, are replaced with the identity matrix. For example, for $j=2$,
 \begin{IEEEeqnarray*}{lCl"s}
\mathfrak{G}_2^{(\ell)}=\left[
\begin{array}{c |c| c| c|c}
\mathbf{I}&\mathbf{0}&\mathbf{0}&\dotsb &\mathbf{0}  \\
\hline
\mathbf{0}&\mathbf{G}_2^{(\ell)}&\mathbf{0}& \dotsb &\mathbf{0}\\
\hline
\mathbf{0} &\mathbf{0} &\mathbf{I} &\ddots & \vdots\\
\hline
\vdots &\vdots &\ddots & \ddots &\mathbf{0}\\
\hline
\mathbf{0} &\mathbf{0}&\dotsb&\mathbf{0} & \mathbf{I}
\end{array}
\right].
\end{IEEEeqnarray*}
   
\subsection{Linear State-Space Model of LETSs for Row Layered SPA} \label{subsec4.2}
%In this section, based on the introduced matrices in the previous part, the linear state-space model of the LETSs in layered decoders is presented.

%\text{($l\geq 1$ \& $j = 1$}): \\
%\label{lay_s1}
%\tilde{\mathbf{x}}^{(l,1)}=\\
%\   \mathfrak{G}_1^{(l)}\big (\mathfrak{A}_1\tilde{\mathbf{x}}^{(l-1,J)}+\mathfrak{B}_1\L+\overset{\triangleleft}{\mathfrak{B}}_{ex1}\L_{ex}^{(l-1)}+\overset{\triangleright}{\mathfrak{B}}_{ex1}L_{ex}^{(l)} \big )\IEEEyesnumber \\

Using the matrices presented in the previous subsection, we have the following linear state-space model of a LETS for row layered SPA:

\begin{small}
\begin{IEEEeqnarray*}{lCl"s}
\tilde{\mathbf{x}}^{(0,j)}=\bold{0} \ \ \ \ \ \ \ \  \ \ \ \ \ \ \ \ \ \ \ \ \  \ \ \ \ \ \text{\ \ \ \  \ \ \ \ \ \ for $\ell=0$, $ 1 \leq j \leq J$} \IEEEyesnumber\\
\label{lay_sj}
\tilde{\mathbf{x}}^{(\ell,j)}= \ \ \ \ \ \ \ \ \ \ \ \ \ \ \ \ \ \ \ \ \ \text{\ \ \ \ \ \ \ \ \ \ \ \ \ \ \ \ \ for $\ell\geq 1$, $ 1 \leq j \leq J$}\\  \  \mathfrak{G}_j^{(\ell)}\big (\mathfrak{A}_j\tilde{\mathbf{x}}^{(\ell-\delta_{j1}, j-1+J\delta_{j1})} +\mathfrak{B}_j\L+\overset{\triangleleft}{\mathfrak{B}}_{ex,j}\L_{ex}^{(\ell-1)}+\overset{\triangleright}{\mathfrak{B}}_{ex,j}\L_{ex}^{(\ell)} \big ).\IEEEyesnumber 
\end{IEEEeqnarray*} 
\end{small}

In the above model, the state vector at layer $j$ of iteration $\ell$ is denoted by $\tilde{\mathbf{x}}^{(\ell,j)}$. The vectors 
$\L$, $\L_{ex}^{(\ell-1)} $ and $\L_{ex}^{(\ell)} $ are the inputs to the model and represent channel LLRs and messages from unsatisfied CNs at iterations $\ell-1$ and $\ell$, respectively,
%\footnote{In Butler's model of the standard decoder, $\L_{ex}^{(\ell)} $ is used as the unsatisfied CNs input vector which, in terms of iteration index, is equivalent to $\L_{ex}^{(\ell-1)} $ in our model. The reason is that in \cite{But_SS}, iterations of the MP algorithm starts with CN updates while in this paper, as it was explained at the beginning, the initial phase of decoding at each iteration is calculation of VN to CN messages. Moreover, in our state space definition, initial state values are set to zero which is different from \cite{But_SS}. Indeed, the iteration in our model starts from $\ell=1$ not $\ell=0$.}. 
where $\L_{ex}^{(0)}=\bold{0}$. %The vector $\tilde{\mathbf{x}}^{(\ell,j)}$ contains the state variables at the layer $j$ of the $l$th iteration. 
Also, $\delta_{j1}$, is the Kronecker delta function which is equal to $1$, when $j=1$, and is zero, otherwise. Equation (\ref{lay_sj}) implies that at the first layer of every iteration, the state vector is updated based on the state vector from the last layer of the previous iteration, while in the other layers, the updated states are a function of the state vector of the previous layer within the same iteration. 

%It is important to mention that the elements of $\tilde{\mathbf{x}}^{(l,j)}$ are not exactly the messages flowing at different edges of the LETS subgraph but their scaled version with the iteration-dependant mis-satisfied CNs gain matrices at different layers, $\mathfrak{G}_j^{(l)}$. This definition is a bit different from the state variables in standard decoder. While it does not make any difference in the estimations results, for the purpose of better representation of the non-recursive formula at the end of this section, the scaled messages are considered as state variables in the rest of the paper.

Next, we use induction to derive a non-recursive equation for the state vector at the end of iteration $\ell$, i.e., $\tilde{\mathbf{x}}^{(\ell,J)}$.
For this, we first define some new matrices. The first matrix is defined as
\begin{equation}\label{A_jdowntok_l}
\tilde{\mathbf{A}}_{J \rightarrow k}^{(\ell)}=(\mathfrak{G}_J^{(\ell)}\mathcal{A}_J)  (\mathfrak{G}_{J-1}^{(\ell)}\mathcal{A}_{J-1}) \dots (\mathfrak{G}_k^{(\ell)}\mathcal{A}_k),
\end{equation}
which is, an ordered multiplication of the scaled version of transition matrices of different layers, $\mathfrak{G}_j^{(\ell)}\mathcal{A}_j$, 
from the layer with maximum index $J$ down to the $k$th layer, $k \geq 1$. %The index $(\ell)$ shows the iteration dependency of this matrix due to time variant gain matrices. 
By using $\tilde{\mathbf{A}}_{J \rightarrow k}^{(\ell)}$, we define three other matrices as follows:  

\begin{equation}
\tilde{\mathbf{B}}^{(\ell)}=\tilde{\mathbf{A}}_{J \rightarrow 2}^{(\ell)}\mathfrak{G}_1^{(\ell)}\mathfrak{B}_1+
\tilde{\mathbf{A}}_{J \rightarrow 3}^{(\ell)}\mathfrak{G}_2^{(\ell)}\mathfrak{B}_2+\dots+\mathfrak{G}_J^{(\ell)}\mathfrak{B}_J,
\end{equation}
\begin{equation}
\overset{\triangleright}{\mathbf{ B}}_{ex}^{(\ell)}=\tilde{\mathbf{A}}_{J \rightarrow 2}^{(\ell)}\mathfrak{G}_1^{(\ell)}\overset{\triangleright}{\mathfrak{B}}_{ex,1}+
\tilde{\mathbf{A}}_{J \rightarrow 3}^{(\ell)}\mathfrak{G}_2^{(\ell)}\overset{\triangleright}{\mathfrak{B}}_{ex,2}+\dots+\mathfrak{G}_J^{(\ell)}\overset{\triangleright}{\mathfrak{B}}_{ex,J},
\end{equation}
\begin{equation}\label{Bex_backtriangel_l}
\overset{\triangleleft}{\mathbf{ B}}_{ex}^{(\ell)}=\tilde{\mathbf{A}}_{J \rightarrow 2}^{(\ell)}\mathfrak{G}_1^{(\ell)}\overset{\triangleleft}{\mathfrak{B}}_{ex,1}+
\tilde{\mathbf{A}}_{J \rightarrow 3}^{(\ell)}\mathfrak{G}_2^{(\ell)}\overset{\triangleleft}{\mathfrak{B}}_{ex,2}+\dots+\mathfrak{G}_J^{(\ell)}\overset{\triangleleft}{\mathfrak{B}}_{ex,J}.
\end{equation}
Finally, the non-recursive formula of the state vector at the end of iteration $\ell$ is derived as   
\begin{IEEEeqnarray*}{l}
\label{ss_lay_main}
\tilde{\mathbf{x}}^{(\ell,J)}=\sum_{i=1}^\ell \big(\prod_{j=i+1}^{\underrightarrow{\ell}}\tilde{\mathbf{A}}_{J \rightarrow 1}^{(j)}\big)\tilde{\mathbf{B}}^{(i)}\L
%\\ 
\ +\sum_{i'=1}^\ell \big(\prod_{j'=i'+1}^{\underrightarrow{\ell}}\tilde{\mathbf{A}}_{J \rightarrow 1}^{(j')}\big)\big(\overset{\triangleleft}{\mathbf{ B}}_{ex}^{(i')}\L_{ex}^{(i'-1)}+\overset{\triangleright}{\mathbf{ B}}_{ex}^{(i')}\L_{ex}^{(i')}\big)\:,
\IEEEyesnumber
\end{IEEEeqnarray*}
where the right arrow on top of the product sign denotes the matrix product applied from the left. For example,
\begin{equation*}
\prod_{j=i+1}^{\underrightarrow{\ell}}\tilde{\mathbf{A}}_{J \rightarrow 1}^{(j)}=\tilde{\mathbf{A}}_{J \rightarrow 1}^{(\ell)}\tilde{\mathbf{A}}_{J \rightarrow 1}^{(\ell-1)}\dots \tilde{\mathbf{A}}_{J \rightarrow 1}^{(i+1)}\:.
\end{equation*} 
In the rest of the paper, since we only consider the state vector $\tilde{\mathbf{x}}^{(\ell,J)}$ at the end of each iteration, for simplicity, we may drop the index $J$, and represent the state vector by $\tilde{\mathbf{x}}^{(\ell)}$. %It can be easily verified that substituting $\ell=1,2$ and $J=3$ in (\ref{ss_lay_main}), would result in equations (\ref{s13}) and (\ref{s23}), respectively.

In Subsection \ref{ProbErrLay}, Equation (\ref{ss_lay_main}) will be used to estimate the probability of a LETS failure in the layered decoder.  

\subsection{The Application of DE to the Layered Decoder} \label{subsec4.3}
In the state-space analysis of layered decoders, we use DE to calculate the distribution of messages entering the LETS from the rest of the Tanner graph.  
%There are a number of important points discussed below that must be taken into account.
%
%\subsubsection{QC-LDPC vs. Random LDPC Codes}
To apply DE to a layered decoder of a QC-LDPC code, for each iteration, one needs to derive two distributions corresponding to each edge $e$ of the base graph.
The two distributions correspond to the messages passed from CNs (VNs) to VNs (CNs) of the Tanner graph that are connected by the cluster of edges 
associated with $e$. If edge $e$ connects VN $i$ to CN $j$ in the base graph, we say that the cluster of edges corresponding to $e$ connect Type-$i$ VNs to Type-$j$ CNs
in the Tanner graph. The computation tree for the calculation of such distributions is not only a function of the base graph but also depends on the order in
 which the messages of different row layers are updated. To identify these distributions, we use the notations $\psi^{[i\rightarrow j]}_{\ell}$ and $\psi^{[i\leftarrow j]}_{\ell}$ to denote the probability distribution of the messages from Type-$i$ VNs to Type-$j$ CNs and vice versa in iteration  $\ell$, respectively.

\begin{ex}
\label{exer}
%Construction of QC-LDPC codes by lifting the base graph imposes a structure on the Tanner graph of the code. 
Consider the base graph of Fig. \ref{BaseTannerGraph} corresponding to the following base matrix:
\begin{equation}\label{Hb_example}
\mathbf{H}_b=\left[
\begin{array}{c c c c}
1&1&1&0\\
1&0&1&1\\
0&1&0&1
\end{array}
\right]\:.  
\end{equation}
%with the base graph of Fig. \ref{BaseTannerGraph} can be lifted $3$ times as Fig. \ref{LiftedTanner}. 
%The VNs and CNs are categorized based on the columns and rows they belong which is referred to as \emph{VN type} and \emph{CN type}. In our example, there are $4$ types of VNs and $3$ types of CNs. The CNs and VNs of each type are connected to specific type of nodes determined by the base graph meaning that the structure of the code is not quite random. This property of the QC-LDPC codes or, in general, protograph LDPC codes are referred to as \emph{deterministic neighbourhood} \cite{Thorpe2003}. For example, all the VNs of type-$1$ are only connected to CNs type-$1$ or $2$. One of the major results of deterministic neighbourhood is that, unlike random LDPC codes, two protograph LDPC codes with the same degree distributions but various base matrices might have different thresholds \cite{Liva2007ExitProtog}.
\begin{figure}
\centering
\includegraphics[width=2.8in]{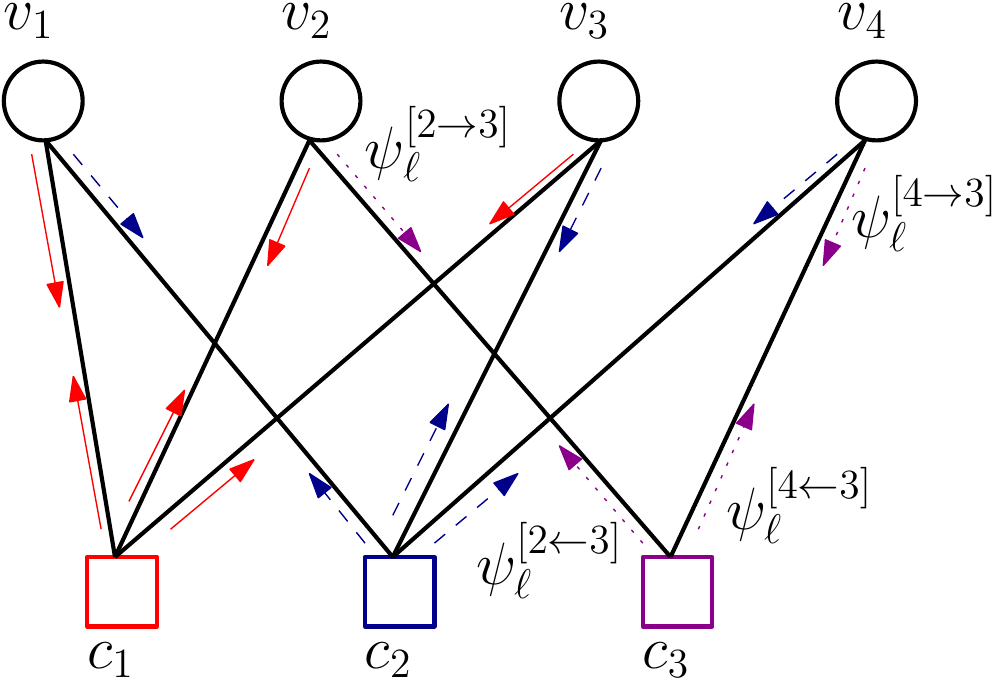}
\caption{Base graph of Example~\ref{exer}: The solid (red), dashed (blue) and dotted (purple) arrows, correspond to messages passed within layers $L_1$ to $L_3$, respectively. The symbols next to the dotted arrows 
show the probability distributions of messages updated in the $3$rd layer of the $\ell$th iteration.}
\label{BaseTannerGraph}
\end{figure}
%\if0
%\begin{figure}
%\centering
%\subfloat[][Base graph of Example~\ref{exer}: The solid (red), dashed (blue) and dotted (purple) arrows, correspond to messages passed within layers $L_1$ to $L_3$, respectively. The symbols next to the dotted arrows show the probability distributions of messages updated in the $3$rd layer of the $\ell$th iteration.]{
%\includegraphics[width=2.8in]{BaseTannerGraph_new.pdf}
%\label{BaseTannerGraph}
%}\\
%\subfloat[][Lifted graph: The base graph is circularly lifted. The lifting degree is $p=3$.]{
%\includegraphics[width=3.0in]{LiftedTaner.pdf}
%\label{LiftedTanner}
%}
%\caption{The Tanner graph related to  the base matrix of $\mathbf{H}_b$ in \eqref{Hb_example} and its lifted version. }
%\end{figure}
%\fi
%\subsubsection{Layered vs. Flooding Schedule Decoding of QC-LDPC Codes}

The messages updated in different layers are identified on the base graph with different colors and line types. 
Suppose that the three layers are updated in accordance to the increasing row index. The computation trees 
of some of the CN to VN messages for the first iteration are demonstrated in Figs. \ref{1stLayer} to \ref{3rdLayer}. 
As can be seen, the computation trees are different depending on the type of the VN at the root, and the type of the connecting CN.
In particular, the trees have different depths within the same iteration. The trees will also change by changing the order in which the layers are updated.
This is unlike the flooding schedule, for which the depth of all computation trees at iteration $\ell$ is $2\ell$ regardless of the type of VNs or CNs, or any permutation of row layers. This implies that, for a layered decoder, the message distributions change with changing the order in which the row layers are updated.

\begin{figure}
\centering
\subfloat[][Messages sent from Type-$1$ CNs to Type-$1$ VNs in the $1$st layer of the $1$st iteration.]{
\includegraphics[width=0.7in]{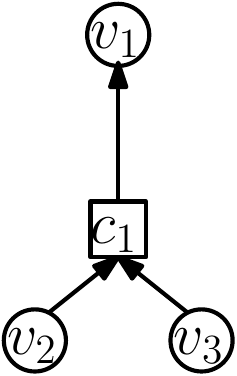}
%\caption{The computation tree of the messages sent from type-1 CNs to type-1 VNs at the 1st layer of the 1st iteration.}
\label{1stLayer}
}\ \ \ \ \ \  \ \ \ \ \ \
\subfloat[][Messages sent from Type-$2$ CNs to Type-$3$ VNs in the $2$nd layer of the $1$st iteration.]{
\includegraphics[width=0.8in]{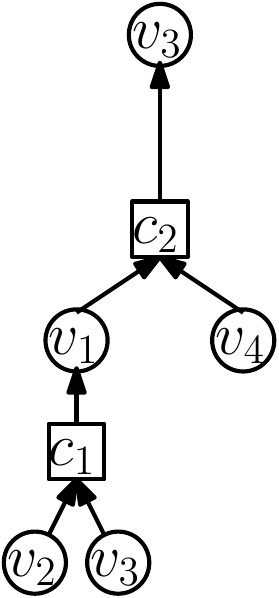}
%\caption{The computation tree of the messages sent from type-2 CNs to type-3 VNs at the 2nd layer of the 1st iteration.}
\label{2ndLayer}}\ \ \ \ \ \  \ \ \ \ \ \ %\\
\subfloat[][Messages sent from Type-$2$ CNs to Type-$4$ VNs in the $2$nd layer of the $1$st iteration.]{
\includegraphics[width=0.9in]{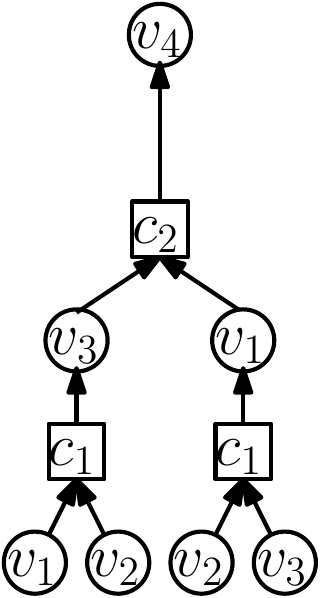}
%\caption{The computation tree of the messages sent from type-2 CNs to type-3 VNs at the 2nd layer of the 1st iteration.}
\label{2ndLayer_B}} \ \ \ \ \ \  \ \ \ \ \ \
\subfloat[][Messages sent from Type-$3$ CNs to Type-$2$ VNs in the $3$rd layer of the $1$st iteration.]{
\includegraphics[width=0.9in]{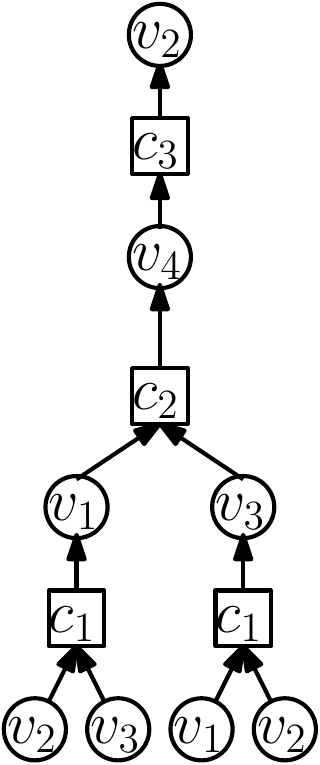}
%\caption{The computation tree of the messages sent from type-3 CNs to type-2 VNs at the 3rd layer of the 1st iteration.}
\label{3rdLayer}}
\caption{The computation trees of certain messages at different layers of the $1$st iteration.}
\end{figure} 

\end{ex}
Consider a LETS whose missatisfied CNs belong to $J$ layers, $L_1,\ldots,L_J$. Consider a permutation $\pi$ over the set of integer numbers $\{1,\ldots,J\}$.
Assume that the messages within the row layers of the parity-check matrix (within one iteration) are updated such that CNs with type $\pi(1)$ within the LETS 
are updated first, followed by CNs with type $\pi(2)$ and so on. This means that missatisfied CNs of the LETS that belong to $L_i$ are of Type-$\pi(i)$, for $i=1,\ldots,J$.    
To obtain the distribution of incoming messages to the LETS from the rest of the Tanner graph using DE, one needs to know not only the topology of the TS, but also the following information:
\begin{enumerate}
\item types of all missatisfied and unsatisfied CNs of the TS as well as the layer to which each CN of the TS belongs (the latter corresponds to knowing the permutation $\pi$), 
\item the type and layer of all the variable nodes that are externally connected to all the missatisfied CNs, and the type and layer of all internal VNs connected to unsatisfied CNs. 
\end{enumerate}
We refer to the above information as {\em TS layer profile} (TSLP).
%
%By using the DE technique, based on the layers order, a probability distribution is obtained for every edge of the base graph  at each iteration. The next step is to select the required distributions for the linear model. In other words, the unsatisfied and mis-satisfied CNs of a TS are located in different layers of the code. Also, they are connected to certain type of VNs. By excluding the internal VNs of the TSs, the external VNs connected to different mis-satisfied CNs can be used to identify the required probability distributions. Also, based on the unsatisfied CNs layers and their neighbouring VNs, the required distributions of the unsatisfied CNs input can be selected.  Given a TS, these information are referred to as \emph{TS layer profile} (TSLP). So, we formally define TSLP as in Definition \ref{TSLP_def}.  
%\begin{mydef}\label{TSLP_def}
%TS layer profile (TSLP): All the information about the type of variable nodes connected to different check nodes of the TS and the check node types connected to different variable nodes of the TS.
%\end{mydef}
We note that the TSLP can be different for isomorphic TSs, thus, resulting in different harmfulness for such TSs. 

%\begin{rem}
%Different TSLPs among isomorphic TSs might, also, impact their harmfulness in flooding schedule decoders. The reason is the possible variation among the statistical property of external messages entering an LETS subgraph. However, regarding the layered decoders, in addition to the substantial differences of external messages distributions as a function of TSLP of LETSs, the order of layers inside the LETSs, having different TSLPs, might be changed that may affect the failure rate, significantly.    
%\end{rem}

\begin{ex}\label{example_TSLP}
Consider the $(5,3)$ LETS of Fig. \ref{(5,3)lay_col}, and assume that the layers of the code are updated in the increasing order of row indices. Variable nodes $v_1$ and $v_5$ 
are of Type-3 and Type-5, respectively, while VNs $v_2$, $v_3$ and $v_4$, are of Type-1. As an example of a missatisfied CN, consider $c_j$. Check node $c_j$ is of Type-$1$ and is in $L_1$. 
Also, the $d_{c_j}-2=3$ external VNs of $c_j$ are of Type-$2$, -$4$ and -$5$, respectively. As a result, to derive the gain value corresponding to $c_j$ in the linear state-space model, 
one needs the probability distributions of messages from Type-$2$, -$4$ and -$5$ VNs sent to Type-$1$ CNs from the DE results. 
%The same approach should be taken for all the mis-satisfied CNs of the subgraph. 
As an example of an unsatisfied CN, consider $c_2$.  Check node $c_2$ is of Type-$2$ and is in $L_2$, and is connected to the internal VN $v_4$, which is of Type-$1$.
The probability distribution of the messages sent from $c_2$ to $v_4$, thus follows the DE results for messages from Type-$2$ CNs to Type-$1$ VNs.
\end{ex}  

The following lemma follows directly from the fact that a QC-LDPC code is a cyclic lifting of the corresponding base code.
 
\begin{lem}\label{TSLP_lem}
Let $p$ be the lifting degree of a QC-LDPC code ${\cal C}$, and let $\mathcal{S}$ be a LETS of ${\cal C}$.
Then, the TSLP of $\mathcal{S}$ is the same as the TSLP of any of the isomorphic LETSs whose VNs are obtained by 
circularly shifting (modulo $p$) the VNs of $\mathcal{S}$.
\end{lem}  
%\begin{proof}
%Let the VNs of $\mathcal{S}$ together with all the external VNs connected to mis-satisfied and unsatisfied CNs are denoted by $\bar{\mathcal{S}}$. The external VNs related to each CN are clustered first and then are placed in $\bar{\mathcal{S}}$ in an ordered fashion. Due to the quasi-cyclicity property \cite{Nguyen}, there exist $p-1$ set of VNs that are obtained by proper circular shift of the individual nodes within the corresponding circulants. Moreover, the subgraph induced by those sets are isomorphic to $\bar{\mathcal{S}}$. As the TSLP is determined based on the external connections of a TS, the TSLPs of the $p-1$ resulted LETSs are equivalent to the TSLP of $\mathcal{S}$.     
%\end{proof}  
%
%\begin{rem}
%If the VNs of a TS, $\mathcal{S}$, are of different types (different columns of the the base graph), shifting the VNs of $\mathcal{S}$ as discussed in Lemma \ref{TSLP_lem} lead to $p-1$ distinct TSs. Otherwise, the shifting procedure might not lead to separate TSs and some of them might completely overlap each other. 
%\end{rem}  
\subsection{Missatisfied CN Gain Model in Layered Decoders}\label{gain_lay_section} 
\begin{figure}
\centering
\includegraphics[width=1.5in]{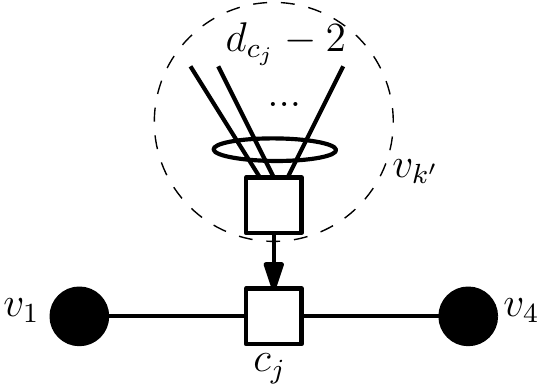}
\caption{The external messages of a missatisfied CN are represented by a virtual VN.}%VNs $v_1$ and $v_4$ related to Fig. \ref{(5,3)lay_col} with the mis-satisfied CN $c_j$ connecting them. }
\label{mis_sat_53}
\end{figure}
In Subsection \ref{subsec4.1}, we introduced the matrix $\mathbf{G}^{(\ell)}$ whose diagonal elements are the gains associated with missatisfied CNs at different layers of the $\ell$th iteration.  
Due to the difference in the distributions of external messages entering different missatisfied CNs, the gain for each missatisfied CN needs to be calculated separately.
%As it is mentioned before, due to the layered schedule, the distribution of the messages at different layers are significantly different. Moreover, at each layer, the distribution of the LLRs going through different CNs or even the same CNs from the neighbouring VNs might be different. Therefore, in order to model the effect of mis-satisfied CNs external connections, a reasonable approach would be to compute linear gains for all the mis-satisfied CNs included in the TS, separately.
For simplicity, the external connections of each missatisfied CN are represented by a virtual VN, as shown in Fig.~\ref{mis_sat_53}.
%Due to the variety of distributions among the external connections of a mis-satisfied CN, to approximate a linear gain, the external connections of each of the mis-satisfied CNs are represented by a virtual VN, first, and then the linear gain is approximated. 
Without loss of generality, we consider VNs $v_1$ and $v_4$ of the $(5,3)$ LETS of Fig. \ref{(5,3)lay_col} with the missatisfied CN $c_j$ connecting them (as shown in Fig. \ref{mis_sat_53}). 
The virtual VN is denoted by $v_{k'}$ in Fig. \ref{mis_sat_53}.
%The box-plus algorithm can be implemented in a pairwise fashion in a sense that all the external messages entering the mis-satisfied CN $c_j$ can be combined by repetitive application of equation \eqref{boxplus1}. It is assumed the resulted message comes from a virtual VN denoted by $v_{k'}$. 
The message $L_{\ell}^{[4 \leftarrow j]}$ can then be calculated as 
\begin{equation}
L_{\ell}^{[4 \leftarrow j]}=f(L_{\ell}^{[1 \rightarrow j]}, L_{\ell}^{[k' \rightarrow j]})\:,
%=L_{\ell}^{[1 \rightarrow j]}\boxplus L_{\ell}^{[k' \rightarrow j]},
\end{equation}
where $f(\cdot,\cdot)$ is the box-plus operation given in (\ref{boxplus_pair}).
%\begin{equation}
%f(x_1,x_2)= \ln\Bigg(\dfrac{1+e^{x_1+x_2}}{e^{x_1}+e^{x_2}}\Bigg).
%\end{equation}
For small values of $L_{\ell}^{[1 \rightarrow j]}$, i.e., $L_{\ell}^{[1 \rightarrow j]} \approx 0$,
% (which is an instance of internal LETS messages), 
the CN linear estimation based on Taylor expansion can be obtained as
\begin{equation}
L_{\ell}^{[4 \leftarrow j]} \approx f_{x_1}(0,L_{\ell}^{[k' \rightarrow j]})L_{\ell}^{[1 \rightarrow j]}\:,
\end{equation}
where $f_{x_1}$ represents the partial derivative of $f(x_1,x_2)$ with respect to $x_1$.  We thus have
\begin{equation}
\begin{split}
f_{x_1}(0,L_{\ell}^{[k' \rightarrow j]}) = \ln\Bigg(\dfrac{e^{L_{\ell}^{[k' \rightarrow j]}}-1}{e^{L_{\ell}^{[k' \rightarrow j]}}+1} \Bigg)  &\\ =\tanh \Bigg( \dfrac{{L_{\ell}^{[k' \rightarrow j]}}}{2}\Bigg)\:.
\end{split}
\end{equation}
The average gain corresponding to CN $c_j$, $\bar{g}^{(\ell)}_{c_j}$, is then obtained as
\begin{equation}
\bar{g}^{(\ell)}_{c_j}=\int_{-\infty}^{\infty}\tanh(\frac{\lambda}{2})\hat{\psi}^{[k'\rightarrow j]}_\ell(\lambda)d\lambda,
\label{g_lay}
\end{equation}
in which $\hat{\psi}^{[k'\rightarrow j]}_\ell$ is the probability distribution of the message from virtual VN, $v_{k'}$, to the missatisfied CN $c_j$ at iteration $\ell$. 
Given the TSLP of a LETS, this distribution can be calculated for different missatisfied CNs using DE. 
%With respect to Example \ref{example_TSLP}, $\hat{\psi}^{[k'\rightarrow j]}_\ell$ is, technically, the CN to VN message distribution of a CN with degree 4 (i.e. $d_{c_j}-1$) where the inputs are the distribution of messages from type-2, 4 and 5 VNs toward type-1 CNs. 
It is noted that there are $\frac{m_s}{2}$ missatisfied CNs in a LETS subgraph, and that each of the gains calculated by \eqref{g_lay} should be used  for the two state variables that have the corresponding CN in common. %The integral in \eqref{g_lay} can be computed numerically as the distributions are discretized.   

We further modify Equation \eqref{g_lay} to take into account the effect of polarity inversion, i.e., whenever an error occurs in the messages from the virtual VN, $v_{k'}$, for $k'=\{1,\dots,\frac{m_s}{2}\}$, the polarity of the state variables passing through the corresponding missatisfied CN is altered. The probability of polarity inversion is thus calculated by
\begin{equation}
P_{inv,\ell}^{[k'\rightarrow j]}=\int_{-\infty}^{0}\hat{\psi}^{[k'\rightarrow j]}_\ell(\lambda)d\lambda\:.
\label{Pinv_lay}
\end{equation}
%In layered decoder, we propose to model the effect of inversion by randomly changing the polarity of state variables. This is unlike the standard decoder where the state variables are set to zero.
To incorporate the polarity inversion in the model, we modify the average gains as follows:
%In layered decoder, like the standard decoder, the effect of polarity inversion is modelled by randomly setting the gain equal to zero. Nevertheless, as it was explained before, appending the polarity inversion effect into the model is a heuristic approach and can be considered as just adding more weight to the importance of inversions. By using \eqref{Pinv_lay}, the average mis-satisfied CNs gain can be written as
\begin{equation}
\bar{g}'^{(\ell)}_{c_j}=(1-P_{inv,\ell}^{[k'\rightarrow j]})\bar{g}^{(\ell)}_{c_j}\:,
\end{equation}  
where $\bar{g}^{(\ell)}_{c_j}$ is given by (\ref{g_lay}).
%As iterations go on, the multiplicative gains of different mis-satisfied CNs approach to 1.

\subsection{Spectral Properties of LETS System Matrices in Layered Decoders } \label{SpectralProp_SubSec}
A careful study of Equation \eqref{ss_lay_main} reveals that the iteration dependent matrix
%inspection of non-recursive layered decoder state-space equation, \eqref{ss_lay_main}, it can be observed that the iteration dependent matrix 
\begin{equation}
\tilde{\mathbf{A}}_{J \rightarrow 1}^{(\ell)}=(\mathfrak{G}_J^{(\ell)}\mathcal{A}_J)  (\mathfrak{G}_{J-1}^{(\ell)}\mathcal{A}_{J-1}) \dots  (\mathfrak{G}_1^{(\ell)}\mathcal{A}_1),
\end{equation}
plays a crucial role in the evolution of the state vector, and consequently in the failure rate, of a LETS. 
In the event of a failure, as the iterations progress, the average gains of missatisfied CNs tend to one. We thus have
\begin{equation}\label{Atild_iterDepend_AtildIterIndepend}
\displaystyle{\lim_{\ell \to \infty}}\tilde{\mathbf{A}}_{J \rightarrow 1}^{(\ell)}=\mathcal{A}_J\mathcal{A}_{J-1} \dots  \mathcal{A}_1 \stackrel{\Delta}{=} \tilde{\mathbf{A}}_{J \rightarrow 1}.\\
\end{equation} 
The matrix $\tilde{\mathbf{A}}_{J \rightarrow 1}$, obtained by the multiplication of transition matrices of different layers of the LETS, is iteration independant. In the state-space analysis of the layered decoder, $\tilde{\mathbf{A}}_{J \rightarrow 1}$ plays a similar role as the transition matrix $\mathbf{A}$ does in the model for a flooding decoder. In the following, we refer to $\tilde{\mathbf{A}}_{J \rightarrow 1}$ as the transition matrix of the layered decoder, and investigate its spectral properties. 
In particular, we demonstrate that, unlike the case for flooding where for majority of LETSs, the transition matrix $\mathbf{A}$ is irreducible, 
the transition matrix $\tilde{\mathbf{A}}_{J \rightarrow 1}$ for layered decoding is reducible. We also show that the spectral radius of $\tilde{\mathbf{A}}_{J \rightarrow 1}$ is always larger than that of $\mathbf{A}$.
%Our goal is to estimate the asymptotic behaviour of high powers of matrix $\tilde{\mathbf{A}}_{J \rightarrow 1}$ by using its possible dominant eigenvalue and eigenvector to obtain an error indicator function based on equation \eqref{ss_lay_main}. However, this requires a detailed study of spectral property of $\tilde{\mathbf{A}}_{J \rightarrow 1}$. In the following, the spectral property of $\tilde{\mathbf{A}}_{J \rightarrow 1}$ is investigated and the main result is presented in Theorem \ref{main_theo_rowlayered_spectra}.

Using the systematic form of system model matrices as described in Subsection~\ref{subsec4.1}, the first column block of the transition matrix  $\tilde{\mathbf{A}}_{J \rightarrow 1}$, consisting of $n_{L_1}$ columns, is zero. We thus have the following result.

\begin{lem}
\label{Atild_reducible_lem}
The transition matrix  $\tilde{\mathbf{A}}_{J \rightarrow 1}$ of a LETS in a layered decoder is reducible.
\end{lem} 
%\begin{proof}
%By considering the systematic representation of system matrices per each layer, denoting the number of state variables belonging to the first layer by $n_{L_1}$, the first $n_{L_1}$ columns of $\tilde{\mathbf{A}}_{J \rightarrow 1}$ are equal to zero columns. Generally, the number of zero columns might be more than $n_{L_1}$. Based on the definition of irreducible matrices and equation \eqref{Perm_irreducible_def}, $\tilde{\mathbf{A}}_{J \rightarrow 1}$ can be repartitioned as \eqref{Perm_irreducible_def} meaning that it is a reducible matrix. 
%\end{proof}

%Following the systematic representation of the transition matrices per each layer in section \ref{LayMatSec}, $\tilde{\mathbf{A}}_{J \rightarrow 1}$ has a special structure in which the first column consists of zero blocks. 

\begin{ex}
The transition matrix, $\tilde{\mathbf{A}}_{3 \rightarrow 1}$, of the $(5,3)$ LETS in Fig. \ref{(5,3)lay_col} (also considered in Examples $2$ and $3$) is equal to  
\end{ex}

\begin{small}
\begin{equation}\label{Atild_3arrow1}
\begin{split}
&\tilde{\mathbf{A}}_{3 \rightarrow 1}=\mathbcal{A}_3\mathbcal{A}_2\mathbcal{A}_1=
\left[
\begin{array}{c|c|c}
\mathbf{0} &\mathbf{A}_{1,2}&\mathbf{A}_{1,3}\\
\hline
\mathbf{0}&\mathbf{A}_{2,1}\mathbf{A}_{1,2}& \mathbf{A}_{2,1}\mathbf{A}_{1,3}+\mathbf{A}_{2,3}\\
\hline
\mathbf{0}& \mathbf{A}_{3,1}\mathbf{A}_{1,2}+\mathbf{A}_{3,2}\mathbf{A}_{2,1}\mathbf{A}_{1,2}&\mathbf{A}_{3,1}\mathbf{A}_{1,3}+\mathbf{A}_{3,2}\mathbf{A}_{2,1}\mathbf{A}_{1,3}+\\
& &\mathbf{A}_{3,2}\mathbf{A}_{2,3} 
\end{array}
\right]\\
& \ \ \ \ \ \ \ \ \ \ \ \ \ \ \ \ \ \ =\left[
\begin{array}{c c c c|c c c c|c c c c}
0&0&0&0&0&0&0&1&0&1&0&0\\
0&0&0&0&0&0&0&0&0&0&0&1\\
0&0&0&0&0&0&1&0&0&0&0&0\\
0&0&0&0&1&0&0&0&0&0&1&0\\
\hline
0&0&0&0&0&0&0&0&1&0&0&0\\
0&0&0&0&0&0&1&0&0&0&1&0\\
0&0&0&0&0&0&0&0&0&1&0&1\\
0&0&0&0&1&0&0&0&0&0&1&0\\
\hline
0&0&0&0&1&0&0&0&0&0&1&1\\
0&0&0&0&0&0&1&0&0&0&1&0\\
0&0&0&0&0&0&0&1&0&1&0&0\\
0&0&0&0&0&0&1&0&1&0&0&0
\end{array}
\right]. 
\end{split}
\end{equation}
\end{small}

%For the matrix $\tilde{\mathbf{A}}_{3 \rightarrow 1}$ in the above example, by different partitioning and comparing the result with \eqref{Perm_irreducible_def}, we can write,
%\begin{small}
%\begin{equation}
%\begin{split}
%&\tilde{\mathbf{A}}_{3 \rightarrow 1}=\\&
%\left[
%\begin{array}{c| c c}
%\mathbf{0} &\mathbf{A}_{12}&\mathbf{A}_{13}\\ 
%\hline
%\mathbf{0}&\mathbf{A}_{21}\mathbf{A}_{12}& \mathbf{A}_{21}\mathbf{A}_{13}+\mathbf{A}_{23}\\
%\mathbf{0}& \mathbf{A}_{31}\mathbf{A}_{12}+\mathbf{A}_{32}\mathbf{A}_{21}\mathbf{A}_{12}&\mathbf{A}_{31}\mathbf{A}_{13}+\mathbf{A}_{32}\mathbf{A}_{21}\mathbf{A}_{13}+\\
%& &\mathbf{A}_{32}\mathbf{A}_{23} 
%\end{array}
%\right]\\
%&=\left [\begin{array}{c|c}
%\mathbf{M}_a &\mathbf{M}_{c}\\
%\hline
%\mathbf{0}& \mathbf{M}_b
%\end{array} \right ], 
%\end{split}
%\end{equation} 
%\end{small} 
%that verifies the reducibility of the system matrix. In our scenario, $\mathbf{M}_a$ is an all-zero submatrix. 
\begin{theo}{\cite[p.~50]{varga}}\label{ferob_normal_theo}
For every real $n\times n$ matrix $\mathbf{M}$, there exists an $n\times n$ permutation matrix $\mathbf{P}$ such that %the $t \times t$ block matrix $\mathbf{P}\mathbf{M}\mathbf{P}^T$ can be obtained as
\begin{IEEEeqnarray*}{lCl"s}\label{FroNormForm}
\mathbf{P}\mathbf{M}\mathbf{P}^T=\left[
\begin{array}{c c c c}
\mathbf{M}_{1,1}&\mathbf{M}_{1,2}&\dotsb &\mathbf{M}_{1,t}  \\
0&\mathbf{M}_{2,2} & \dots &\mathbf{M}_{2,t}\\
%0& g_2^{(l)}  &0 &\dotsb&\vdots\\
\vdots & \ddots &\ddots &\vdots \\
0&\dotsb &0&\mathbf{M}_{t,t}
\end{array}
\right] ,
\IEEEyesnumber
\end{IEEEeqnarray*}
is a $t \times t$ array of matrices, where each of the diagonal matrices $\mathbf{M}_{i,i}$, $1\leq i \leq t$, is either irreducible or a $1\times 1$ null matrix (zero).  
\end{theo}
The form of a matrix given by (\ref{FroNormForm}) is called \textit{Frobenius normal form} or FNF, in brief. In the above theorem, if $\mathbf{M}$ is irreducible, then $t=1$.
%\footnote{It is, also, called normal form or reducible normal form of a matrix.}. 
%Obviously, regarding the above theorem and equation \eqref{Perm_irreducible_def}, if $\mathbf{M}$ is irreducible, then $t=1$.

%It will be explained in later sections that the system matrix $\tilde{\mathbf{A}}_{J \rightarrow 1}$ is not necessarily in the normal form and a proper symmetrical permutation may be required to further transform $\tilde{\mathbf{A}}_{J \rightarrow 1}$ to the form of \eqref{FroNormForm}.   
\begin{theo}\cite [p.~51]{varga}
\label{rho_reducible_nonNeg_theo}  
Let $\mathbf{M}$ be an $n\times n$ non-negative matrix. Then,
\begin{itemize}
        \item[(a)] $\mathbf{M}$ has a real non-negative eigenvalue equal to its spectral radius $\rho(\mathbf{M})$. In addition, the corresponding real eigenvalue is positive unless $\mathbf{M}$ is reducible and its FNF is strictly upper triangular.\footnote{A matrix $\mathbf{M}$ is called strictly upper triangular if $\mathbf{M}(i,j)=0$ for $i\geq j$, where $\mathbf{M}(i,j)$ denotes the element of $\mathbf{M}$ in the $i$th row and the $j$th column.}
        \item[(b)] The associated eigenvector of $\rho(\mathbf{M})$ is non-negative.%, $\mathbf{x}\geq \mathbf{0}.$
    \end{itemize} 
\end{theo}
%\begin{proof}
%Due to the significance of this theorem in the following sections, the proof of part (a) is described and the rest can be found in \cite {varga}.\\

The proof of Theorem~\ref{rho_reducible_nonNeg_theo} follows from Theorem~\ref{Perron_Frobenius_theo} if $\mathbf{M}$ is irreducible.
% matrix, the results immediately follow from the Perron Frobenius theorem. 
If $\mathbf{M}$ is reducible, then by considering the FNF $\mathbf{P}\mathbf{M}\mathbf{P}^T$ of $\mathbf{M}$, and the fact that the eigenvalues of $\mathbf{M}$
are the union of the eigenvalues of the diagonal matrices of FNF, 
% together with some zeros if some of the diagonal matrices of FNF are equal to zero (the multiplicity of the zero eigenvalue, in this case is equal to the number of zero diagonal matrices of FNF), 
one can see that $\mathbf{M}$ has either a real positive eigenvalue equal to its spectral radius (if there is at least one irreducible diagonal matrix in FNF), 
or the spectral radis is zero (if all the diagonal matrices of FNF are $1 \times 1$ zero matrices).

\begin{theo}\cite [Theorem 22]{noutsos_slide}\label{Mx_bx_geq_zero}
Let $\mathbf{M}$ be an $n\times n$ real non-negative matrix. Suppose that for a non-negative vector $\mathbf{x}\geq \mathbf{0}$ and $\mathbf{x}\neq \mathbf{0}$, we have  $\mathbf{Mx}-\beta\mathbf{x}\geq \mathbf{0}$, where $\beta>0$ is a constant. Then,
\begin{equation*}
\rho(\mathbf{M}) \geq \beta,
\end{equation*} 
with the inequality being strict if $\mathbf{Mx}-\beta\mathbf{x}> \mathbf{0}$. 
\end{theo}

The following result is simple to prove.

\begin{lem}
Let $\tilde{\mathbf{A}}_{J \rightarrow 1}$ and $\mathbf{A}$ be the transition matrices of a LETS in layered and flooding decoders, respectively. 
%By decomposing the systematic form of $\mathbf{A}$ as 
We then have
\begin{equation}\label{Atild_equivalent}
 \tilde{\mathbf{A}}_{J \rightarrow 1}=\mathbf{A}+\mathbf{A}'(\mathbf{A-I}), 
\end{equation} 
where
\begin{equation}
\mathbf{A}'=\mathbf{A}_l^{J-1}+\mathbf{A}_l^{J-2}+\dots+\mathbf{A}_l,
\end{equation}
and $\mathbf{A}_l$ is a lower triangular matrix whose non-zero elements (under main diagonal) are equal to those of matrix $\mathbf{A}$.
%\begin{equation}
%\mathbf{A}=\mathbf{A}_l+\mathbf{A}_u,\label{AeqAlAu}
%\end{equation}
%where $\mathbf{A}_l$ and $\mathbf{A}_u$ are the lower and upper triangular matrices whose non-zero entries are equal to the lower part (under main diagonal) and upper part (above main diagonal) of matrix $\mathbf{A}$, respectively, the following equality is valid
\end{lem}
%\begin{proof}
%It is straight forward to prove the above equalities by induction.
%\end{proof}

%\begin{rem}
%Considering the LETS structures with irreducible system matrices in the flooding decoder, $\mathbf{A}$, their maximum positive real eigenvalues, $r$, are greater than $1$, $r>1$ \cite[Theorem 2]{But_SS}. Regarding the simple cycles as the only LETSs with reducible system matrices, the maximum eigenvalue is $1$, $r=1$ \cite[Appendix B]{But_SS}. 
%\end{rem}

\begin{pro}\label{rhoTild_geq_rho_corollary}
Let $\tilde{\mathbf{A}}_{J \rightarrow 1}$ and $\mathbf{A}$ be the transition matrices of a LETS in layered and flooding decoders, respectively. 
If $\mathbf{A}$ is irreducible, then
\begin{equation*}
\rho(\tilde{\mathbf{A}}_{J \rightarrow 1})\geq \rho(\mathbf{A}).
\end{equation*}  
\end{pro}
\begin{proof}
Since $\mathbf{A}$ is irreducible, it has a simple positive eigenvalue, $r$, which is equal to $\rho(\mathbf{A})$. To $r$, 
there corresponds a positive right eigenvector, $\mathbf{u}_1$. By right multiplication of \eqref{Atild_equivalent} with $\mathbf{u}_1$, we have
\begin{equation}\label{Atild_u_1Projection}
 \tilde{\mathbf{A}}_{J \rightarrow 1}\mathbf{u}_1=r\mathbf{u}_1+\mathbf{A}'(r-1)\mathbf{u}_1. 
\end{equation}
Since $r>1$, $\mathbf{u}_1>\mathbf{0}$ and $\mathbf{A}'\geq \mathbf{0}$, we have $\mathbf{A}'(r-1)\mathbf{u}_1 \geq \mathbf{0}$ and 
$r \mathbf{u}_1>\mathbf{0}$, and therefore
 \begin{equation}
 \tilde{\mathbf{A}}_{J \rightarrow 1}\mathbf{u}_1-r\mathbf{u}_1\geq \mathbf{0}. 
\end{equation}
Using this in Theorem \ref{Mx_bx_geq_zero} completes the proof.
\end{proof} 

Based on Lemma \ref{Atild_reducible_lem}, we know that $\tilde{\mathbf{A}}_{J \rightarrow 1}$ is a non-negative reducible matrix. 
As a result, the FNF of $\tilde{\mathbf{A}}_{J \rightarrow 1}$ has either irreducible or $1\times 1$ null matrices as its diagonal blocks. 
In the following, we demonstrate that  the FNF of $\tilde{\mathbf{A}}_{J \rightarrow 1}$ has only one irreducible diagonal block with all the other blocks 
being equal to $1\times 1$ zero matrices. To show this, we use the directed edge graph (digraph) representation of the transition matrices. This representation, which was also used in~\cite{But_SS}, represents each state variable with a node in the graph and the dependencies among them by directed edges, i.e., 
a directed edge $x_i x_j$ in the graph means that state variable $x_j$, corresponding to the {\em head} node of the edge, at iteration $\ell$ is a function of state variable $x_i$, corresponding to the {\em tail} node of the edge, at iteration $\ell -1$. In the digraph, the tail (head) nodes of all the incoming (outgoing) edges of a node $x$, i.e., all the edges with head (tail) $x$, are called the {\em parents} ({\em children}) of $x$. 

It appears that the adjacency matrix of the digraph is equal to the transpose of the transition matrix.
In the following, we use notations $D_{f}(\mathcal{S})=(V_{f}, E_{f})$ and $D_{l}(\mathcal{S})=(V_{l}, E_{l})$ to denote the digraphs 
corresponding to a LETS $\mathcal{S}$ in flooding and layered decoding, respectively, where the set of nodes and edges in each graph are 
denoted by $V$ and $E$, respectively, with indices $f$ and $l$ indicating ``flooding'' and ``layered,'' respectively.  

A digraph is {\em strongly connected} if there exists a directed walk between any pair of its nodes.
Consider a LETS $\mathcal{S}$ with an irreducible transition matrix $\mathbf{A}$ in a flooding decoder.
The irreducibility of $\mathbf{A}$ implies that the digraph $D_{f}(\mathcal{S})=(V_{f}, E_{f})$  is strongly connected~\cite{horn}.

\begin{ex}
The digraph $D_{f}(\mathcal{S})$ of the flooding decoder for the $(5,3)$ LETS of Fig. \ref{(5,3)lay_col} is shown in Fig. \ref{dig_Flood_53}. 
As can be seen this digraph is strongly connected.
\end{ex}

The digraph $D_{l}(\mathcal{S})$ of a LETS $\mathcal{S}$ for a layered decoder can be obtained from its adjacency 
matrix $\tilde{\mathbf{A}}^T_{J \rightarrow 1}$. The digraph $D_{l}(\mathcal{S})$ can also be constructed from 
the flooding digraph $D_{f}(\mathcal{S})$ through $J-1$ steps, each corresponding to one layer of decoding. 
In step $j, j = 1, \ldots, J-1$, one considers all the nodes corresponding to state variables that are updated in layer $L_{j+1}$ 
and their incoming edges. If the tail of any of such incoming edges has been updated in the previous layers within the same iteration, then 
the incoming edge is removed and new incoming edges are created from the parents of the updated tail to the head of the removed edge. 
 
% as an adjacency matrix, the layered schedule algorithm inside an LETS, also, can be illustrated with a directed edge graph, $D_{ls}(\mathcal{S})=(V_{d,{ls}}, E_{d,{ls}})$. The resulted digraph represents the dependency of the state variables (vertices) at iteration $\ell$ to the state variables of iteration $\ell-1$. The digraph of an LETS corresponding to row layered decoder, $D_{ls}$, from the standard decoder digraph, $D_{fs}$, can be obtained by Algorithm \ref{algorithm_lay2floDig}. This algorithm, basically, follows the definition of the layered decoders in which the updated reliabilities of the processed layers are used in the subsequent layers. This is equivalent to the line 6 of the algorithm where the parents of an updated node are connected to its child nodes that belong to the subsequent layers. It is noted that the digraph of the row layered decoder can be directly obtained from $\tilde{\mathbf{A}}_{J \rightarrow 1}$. However, Algorithm \ref{algorithm_lay2floDig}, by revealing the connection between the flooding schedule and layered schedule digraphs would be used in the following parts for the proofs. 
\begin{ex}
The steps for the construction of $D_{l}(\mathcal{S})$ from $D_{f}(\mathcal{S})$ for the $(5,3)$ LETS of Fig. \ref{(5,3)lay_col} are shown in Fig. \ref{DIG53_example}, where different colors and line types are used to identify the state variables and updated edges in different layers.
It can be verified that the adjacency matrix of the resulted digraph of Fig.~\ref{dig_Lay3_53} is equal to the transpose of the matrix $\tilde{\mathbf{A}}_{3 \rightarrow 1}$, given in Equation \eqref{Atild_3arrow1}. In particular, the zero columns of matrix $\tilde{\mathbf{A}}_{3 \rightarrow 1}$ correspond to nodes of Fig.~\ref{dig_Lay3_53} with no outgoing edges.
\end{ex}

\begin{figure*}
\centering
\subfloat[][\label{dig_Flood_53}Flooding decoder digraph $D_{f}(\mathcal{S})$]{
\includegraphics[width=2.5in]{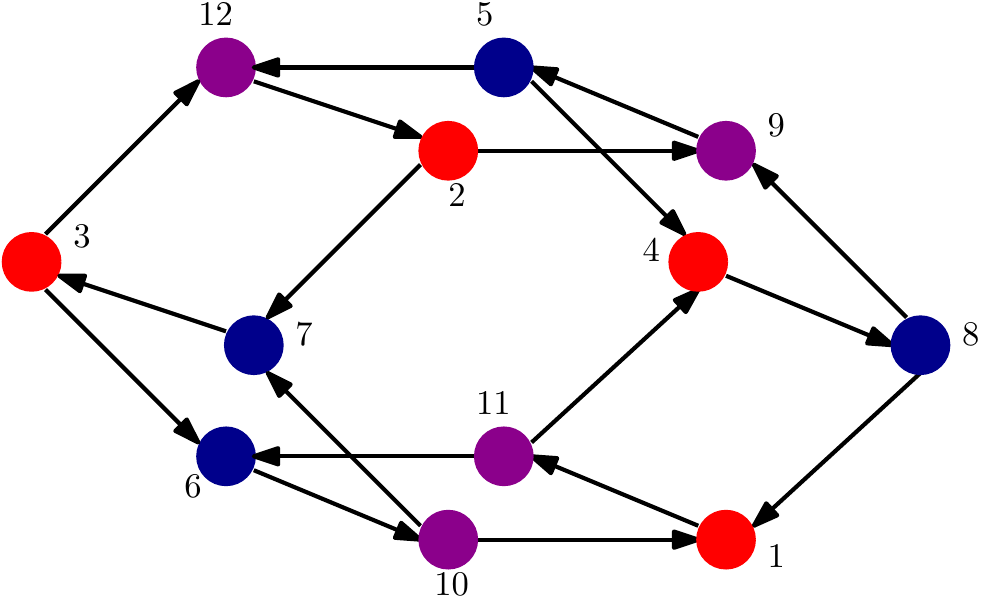}}
\ \ \ \ \ \ \ \ \subfloat[][\label{dig_Lay1_53}Activating the edges updated in $L_1$]{
\includegraphics[width=2.5in]{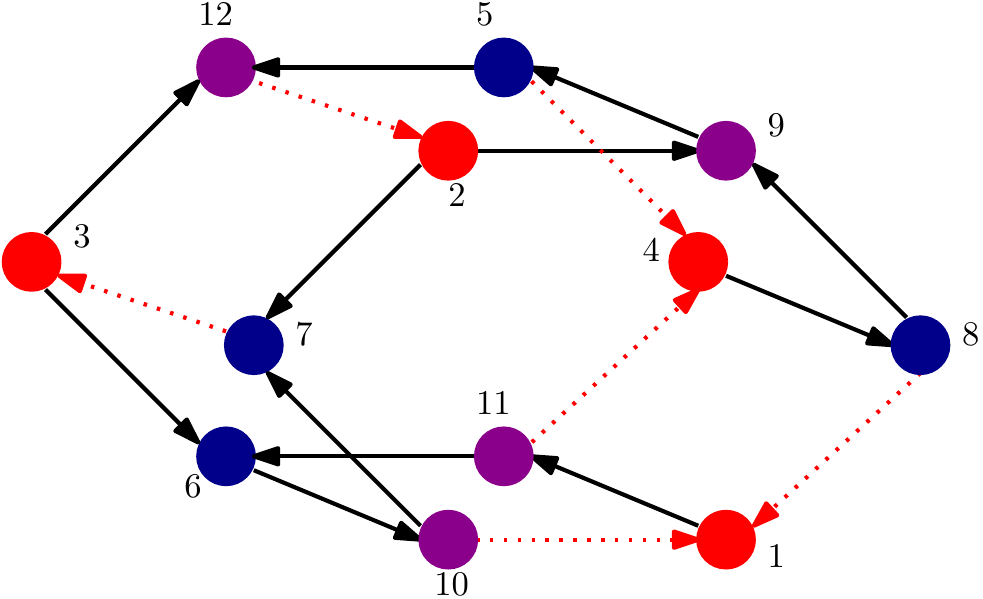}}\\
\subfloat[][\label{dig_Lay2_53}Step 1 of constructing the layered decoder digraph]{
\includegraphics[ width=2.5in, angle =0]{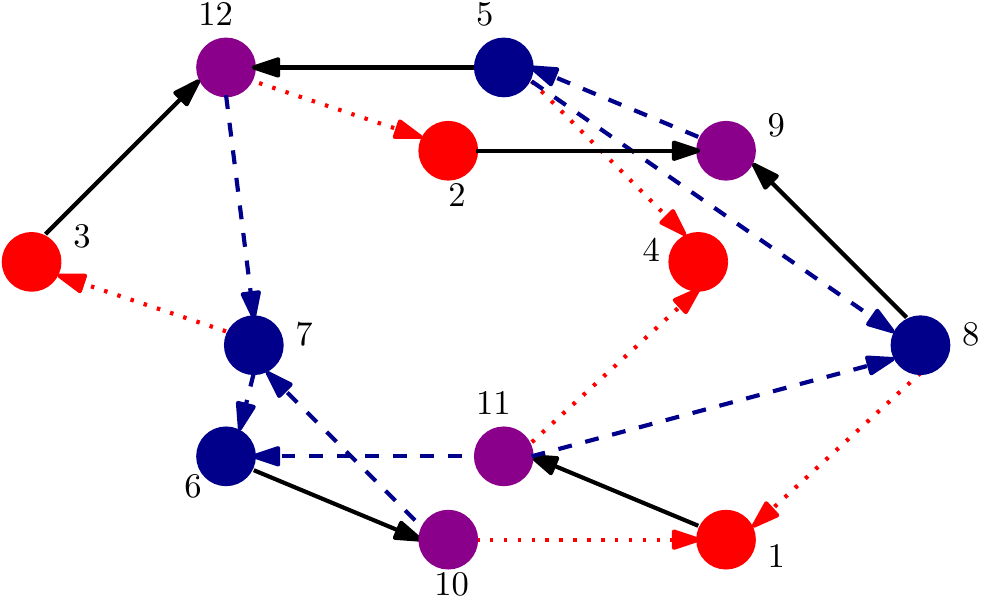}}
\ \ \ \ \ \ \ \  \subfloat[][\label{dig_Lay3_53}Step 2 of constructing the layered decoder digraph]{
\includegraphics[ width=2.5in, angle =0]{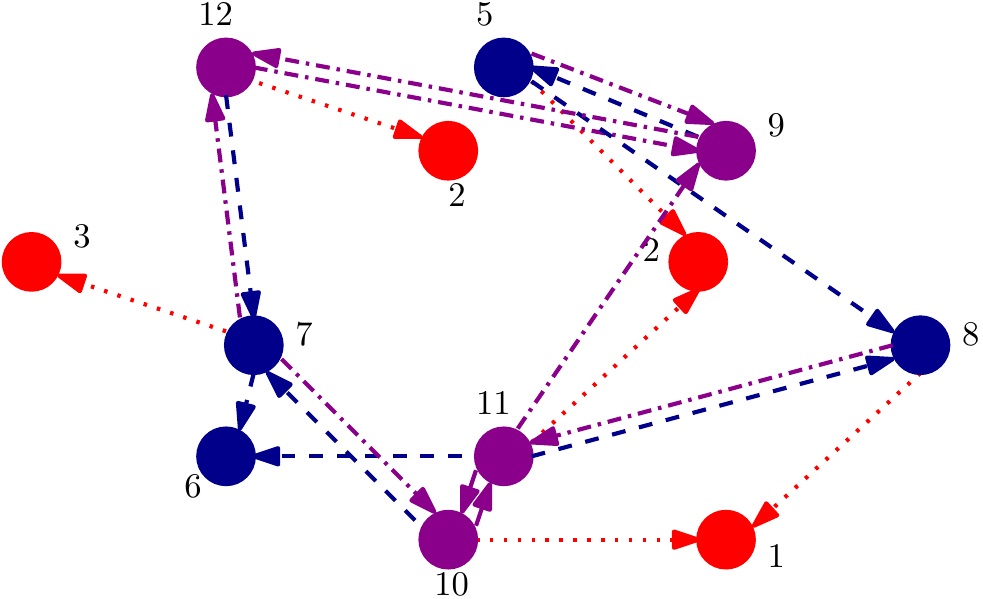}}
\caption{The steps of constructing the digraph $D_{l}(\mathcal{S})$ of the $(5,3)$ LETS for a layered decoder from the flooding digraph $D_{f}(\mathcal{S})$. (The colors (edge types) red (dotted), blue (dashed) and purple (dash-dotted) represent $L_1$, $L_2$ and $L_3$, respectively.) }
\label{DIG53_example}
\end{figure*}
The following result follows from the construction of $D_{l}(\mathcal{S})$  from $D_{f}(\mathcal{S})$.

\begin{lem}
\label{out_edge_removal_lem}
Consider a state variable $v \in V_{f}$ of $D_{f}(\mathcal{S})$, for a LETS $\mathcal{S}$, updated in layer $L_k$, $k=\{1,2,\dots,J\}$.
If all the outgoing edges of $v$ are connected to nodes from subsequent layers $L_j$, $k<j\leq J$, 
the corresponding node $v \in V_{l}$ of $D_{l}(\mathcal{S})$ will not have any outgoing edges. 
The total number of nodes $n_z$ in $D_{l}(\mathcal{S})$ that have this property are greater 
than or equal to the number of nodes updated in the first layer, i.e., $n_z \geq n_{L_1}$.   
\end{lem}

%\begin{proof}
%Following the execution of line 6 in Algorithm \ref{algorithm_lay2floDig}, given the parent node of an in-edge belonging to $L_k$ and the child node belonging to layer $L_j$, if $j>k$, that specific in-edge is removed. If all the child nodes of a given parent node justify this statement, all the corresponding edges are removed meaning that the parent node will not have any out-edge in layered decoder digraph $D_{ls}(\mathcal{S})$. This always happens to the nodes of the first layer as all their child nodes belong to the next layers. 
%\end{proof}

%As can be seen in Fig. \ref{dig_Lay3_53}, out-edges removal happens to the first-layer nodes 1, 2, 3 and 4. However, it might also occur in other nodes such as vertex 6 in our example. These nodes are related to columns 1, 2, 3, 4 and 6 of \eqref{Atild_3arrow1} that are all-zeros. 

The following lemma applies to LETSs with irreducible transition matrices ${\bf A}$ for flooding schedule, and 
follows from the construction of $D_{l}(\mathcal{S})$ from $D_{f}(\mathcal{S})$ and the fact that for such LETSs the digraph
$D_{f}(\mathcal{S})$ is strongly connected.

\begin{lem}
\label{Strong_connect_lay_dig_lem}
Consider a LETS $\mathcal{S}$ with an irreducible transition matrix ${\bf A}$ for flooding schedule.
If after the construction of $D_{l}(\mathcal{S})$ from $D_{f}(\mathcal{S})$, all the nodes with no outgoing edges, and their incoming edges are
removed from $D_{l}(\mathcal{S})$, then the remaining digraph (if any)
%, denoted by $D'_{l}(\mathcal{S})$, 
is strongly connected.   
\end{lem}  

Lemma \ref{Strong_connect_lay_dig_lem} implies that the FNF of the transition matrix $\tilde{\mathbf{A}}_{J \rightarrow 1}$ of a layered decoder 
for a LETS with irreducible matrix ${\bf A}$ has at most one irreducible diagonal block. The following Lemma proves 
that the number of irreducible blocks is exactly one. 
 
 \begin{lem}\label{1_only_1_irred_lem}
The FNF of a layered decoder transition matrix $\tilde{\mathbf{A}}_{J \rightarrow 1}$ corresponding to a LETS, $\mathcal{S}$, 
with irreducible flooding transition matrix, $\mathbf{A}$, has one and only one irreducible diagonal block.  
\end{lem}

\begin{proof}
From Proposition~\ref{rhoTild_geq_rho_corollary}, we have $\rho(\tilde{\mathbf{A}}_{J \rightarrow 1})\geq \rho(\mathbf{A})$.
Moreover, since ${\bf A}$ is irreducible, $\rho(\mathbf{A}) > 1$~\cite{But_SS}. This means $\rho(\tilde{\mathbf{A}}_{J \rightarrow 1}) > 1$, and thus the FNF of $\tilde{\mathbf{A}}_{J \rightarrow 1}$ must have at least one irreducible diagonal block. This together with Lemma~\ref{Strong_connect_lay_dig_lem} completes the proof.
%Based on Lemma \ref{Atild_reducible_lem}, $\tilde{\mathbf{A}}_{J \rightarrow 1}$ is reducible. According to Theorem \ref{ferob_normal_theo}, the diagonal blocks of its Frobenius normal form are either $1\times1$ null or irreducible matrices. Corollary \ref{cor_nonzeroSpectral} implies that the non-zero positive spectral radius, only, exists if at least one of the diagonal blocks is irreducible. This is the case for matrix $\tilde{\mathbf{A}}_{J \rightarrow 1}$ as $\rho(\tilde{\mathbf{A}}_{J \rightarrow 1})\geq \rho(\mathbf{A})>1$, following the Corollary \ref{rhoTild_geq_rho_corollary}. On the other hand, Lemma \ref{Strong_connect_lay_dig_lem} implies that there exists at most one strongly connected digraph, $D'_{ls}(\mathcal{S})$, as a subset of layered decoder digraph $D_{ls}(\mathcal{S})$ meaning that the number of irreducible diagonal blocks in normal form of $\tilde{\mathbf{A}}_{J \rightarrow 1}$ cannot be more than one. Therefore, it is concluded that the Frobenius normal form of $\tilde{\mathbf{A}}_{J \rightarrow 1}$ has only one irreducible diagonal sub-block.        
\end{proof}

%The implication of Lemma \ref{1_only_1_irred_lem} for our example in Fig. \ref{dig_Lay3_53}, which corresponds to system matrix of \eqref{Atild_3arrow1}, is that if the nodes 1, 2, 3, 4 and 6 are removed from the digraph, the remaining subgraph is strongly connected. It is noted that the digraph of this figure is obtained for \eqref{Atild_3arrow1} which is not in the normal form. However, by using a proper permutation matrix $\mathbf{P}$ that swaps the 5th and 6th columns and accordingly swaps the rows symmetrically, i.e., $\mathbf{P}\tilde{\mathbf{A}}_{3 \rightarrow 1} \mathbf{P}^T$, the resulted isomorphic graph would be the same as Fig. \ref{dig_Lay3_53} in which the nodes denoted by different labels.    

Based on Lemma~\ref{1_only_1_irred_lem}, 
%Let us denote the total number of state variables, whose corresponding nodes in $D_{fs}(\mathcal{S})$ are, only, parents of the nodes from subsequent layers, by $n_z$. Naturally, based on Lemma \ref{out_edge_removal_lem}, $n_z$ is greater than or equal to the number of states in the first layer, $n_z\geq n_{L_1}$. 
there exists a permutation matrix $\mathbf{P}$ such that the FNF of $\tilde{\mathbf{A}}_{J \rightarrow 1}$ can be written as follows
\begin{equation}
\mathbf{P}\tilde{\mathbf{A}}_{J \rightarrow 1}\mathbf{P}^T=\left [\begin{array}{c|c}
\mathbf{0}_{n_z\times n_z} &\tilde{\mathbf{A}}'_{n_z\times (m_s-n_z)}\\
\hline
\mathbf{0}_{(m_s-n_z)\times n_z}& \tilde{\mathbf{A}}_{(m_s-n_z)\times (m_s-n_z)}
\end{array} \right ],
\label{p_Atild_p_1}
\end{equation}
where $\tilde{\mathbf{A}}$ is a non-negative irreducible matrix and $n_z$ is defined in Lemma \ref{out_edge_removal_lem}.
%It should be noticed that whenever $n_z= n_{L_1}$, $\tilde{\mathbf{A}}_{J \rightarrow 1}$ has the Frobenius normal form by itself and no permutation is required.
  
\begin{theo}\label{main_theo_rowlayered_spectra}
Let $\tilde{\mathbf{A}}_{J \rightarrow 1}$ be the non-negative transition matrix of a LETS, which is not a simple cycle, in a layered decoder. 
Also, let the matrix $\tilde{\mathbf{A}}$ be the only irreducible diagonal block of the FNF of $\tilde{\mathbf{A}}_{J \rightarrow 1}$. Then,
\begin{itemize}
\item[(a)]$\tilde{\mathbf{A}}_{J \rightarrow 1}$ has a simple positive eigenvalue, $\tilde{r}>1$, equal to its spectral radius $\rho(\tilde{\mathbf{A}}_{J \rightarrow 1})$. 
\item[(b)]To $\tilde{r}$, there corresponds a non-negative left eigenvector $\tilde{\mathbf{w}}_1^T$ that is equal to the positive left eigenvector of $\tilde{\mathbf{A}}$, denoted by $\tilde{\boldsymbol{\omega}}_1^T$, appended by $n_z$ zeros. 
\item[(c)]To $\tilde{r}$, there corresponds a non-negative right eigenvector $\tilde{\mathbf{u}}_1$ that is equal to the positive right eigenvector of $\tilde{\mathbf{A}}$, denoted by $\tilde{\boldsymbol{\nu}}_1$, appended by the vector $\frac{\tilde{\mathbf{A}}'}{\tilde{r}}\tilde{\boldsymbol{\nu}}_1$, in which the matrix $\tilde{\mathbf{A}}'$ is defined in \eqref{p_Atild_p_1}.
\end{itemize} 
\end{theo}
\begin{proof}
(a) The eigenvalues of $\tilde{\mathbf{A}}_{J \rightarrow 1}$ are the roots of $\det(\tilde{\mathbf{A}}_{J \rightarrow 1}-\mu\mathbf{I})$, which based on \eqref{p_Atild_p_1} simplifies to $\det(\tilde{\mathbf{A}}-\mu\mathbf{I})\mu^{n_z}$. The eigenvalues of irreducible matrix $\tilde{\mathbf{A}}$ together with $n_z$ zeros are thus the eigenvalues of $\tilde{\mathbf{A}}_{J \rightarrow 1}$. Since, for a LETS that is not a simple cycle, $\tilde{\mathbf{A}}$ is a non-negative irreducible matrix, its dominant eigenvalue, $\tilde{r}$, which is positive and simple, is the dominant eigenvalue of $\tilde{\mathbf{A}}_{J \rightarrow 1}$ as well. Based on the proof of Lemma~\ref{1_only_1_irred_lem}, this eigenvalue is larger than one, i.e., $\tilde{r}>1$.

(b) Based on \eqref{p_Atild_p_1}, to find the left eigenvector, $\tilde{\mathbf{w}}_1^T$, corresponding to $\tilde{r}$, we need to solve
\begin{equation}
\tilde{\mathbf{w}}_1^T\left [\begin{array}{c|c}
\mathbf{0}_{n_z\times n_z} &\tilde{\mathbf{A}}'_{n_z\times (m_s-n_z)}\\
\hline
\mathbf{0}_{(m_s-n_z)\times n_z}& \tilde{\mathbf{A}}_{(m_s-n_z)\times (m_s-n_z)}
\end{array} \right ]=\tilde{r}\tilde{\mathbf{w}}_1^T.
\end{equation} 
The solution to the above equation is a vector of the form
\begin{equation}
\tilde{\mathbf{w}}_1^T=\left [\begin{array}{c|c} \mathbf{0}_{1 \times n_z} & \tilde{\boldsymbol{\omega}}_1^T 
\end{array} \right ],
\end{equation}   
where $\tilde{\boldsymbol{\omega}}_1^T$ is the positive left eigenvector of $\tilde{\mathbf{A}}$ corresponding to $\tilde{r}$.

(c) Similar to the proof of part (b), it can be easily shown that the right eigenvector of the FNF of $\tilde{\mathbf{A}}_{J \rightarrow 1}$ is
\begin{equation}
\tilde{\mathbf{u}}_1=\left [\begin{array}{c}  \frac{\tilde{\mathbf{A}}'}{\tilde{r}}\tilde{\boldsymbol{\nu}}_1\\
\hline
\tilde{\boldsymbol{\nu}}_1 
\end{array} \right ].
\end{equation} 
\end{proof}

\begin{ex}
The matrices $\tilde{\mathbf{A}}$ and $\tilde{\mathbf{A}}'$ in the FNF of the transition matrix of the 
$(5,3)$ LETS of the Tanner code for the layered decoder are 
%The normal form of the layered decoder system matrix related to $(5,3)$ LETS of the Tanner code is as follows
%
%\begin{small}
%\begin{IEEEeqnarray*}{l}
%\mathbf{P}\tilde{\mathbf{A}}_{3 \rightarrow 1}\mathbf{P}^T=\\
%\left[
%\begin{array}{c c c c c |c c c c c c c}
%0&0&0&0&0&0&0&1&0&1&0&0\\
%0&0&0&0&0&0&0&0&0&0&0&1\\
%0&0&0&0&0&0&1&0&0&0&0&0\\
%0&0&0&0&0&1&0&0&0&0&1&0\\
%0&0&0&0&0&0&1&0&0&0&1&0\\
%\hline
%0&0&0&0&0&0&0&0&1&0&0&0\\
%0&0&0&0&0&0&0&0&0&1&0&1\\
%0&0&0&0&0&1&0&0&0&0&1&0\\
%0&0&0&0&0&1&0&0&0&0&1&1\\
%0&0&0&0&0&0&1&0&0&0&1&0\\
%0&0&0&0&0&0&0&1&0&1&0&0\\
%0&0&0&0&0&0&1&0&1&0&0&0
%\end{array}
%\right].
%\end{IEEEeqnarray*}
%\end{small}
%So, the matrices $\tilde{\mathbf{A}}$ and $\tilde{\mathbf{A}}'$ can be obtained as
\begin{small}
\begin{IEEEeqnarray*}{l}
\tilde{\mathbf{A}}=
\left[
\begin{array}{c c c c c c c }
0&0&0&1&0&0&0\\
0&0&0&0&1&0&1\\
1&0&0&0&0&1&0\\
1&0&0&0&0&1&1\\
0&1&0&0&0&1&0\\
0&0&1&0&1&0&0\\
0&1&0&1&0&0&0
\end{array}
\right],
%\end{IEEEeqnarray*} 
%\end{small}
%and
%\begin{small}
%\begin{IEEEeqnarray*}{l}
\tilde{\mathbf{A}}'=
\left[
\begin{array}{c c c c c c c }
0&0&1&0&1&0&0\\
0&0&0&0&0&0&1\\
0&1&0&0&0&0&0\\
1&0&0&0&0&1&0\\
0&1&0&0&0&1&0\\
\end{array}
\right].
\end{IEEEeqnarray*} 
\end{small}
The dominant eigenvalue of $\tilde{\mathbf{A}}$ is $\tilde{r}=2.0136$, which is also equal to $\rho(\tilde{\mathbf{A}}_{3 \rightarrow 1})$. Moreover, the corresponding dominant left and right eigenvectors of the FNF of $\tilde{\mathbf{A}}_{3 \rightarrow 1}$ are
\begin{small}
\begin{IEEEeqnarray*}{l}
\tilde{\mathbf{w}}_1=
\left[
\begin{array}{c  }
		 0\\
         0\\
         0\\
         0\\
         0\\
         \hline
    0.2838\\
    0.4027\\
    0.2525\\
    0.3189\\
    0.4525\\
    0.5085\\
    0.3584
\end{array}
\right],\tilde{\mathbf{u}}_1=
\left[
\begin{array}{c  }
	0.3342\\
    0.2355\\
    0.2096\\
    0.2974\\
    0.3756\\
    \hline
    0.2647\\
    0.4220\\
    0.2974\\
    0.5329\\
    0.3756\\
    0.3342\\
    0.4743
\end{array}
\right].
\end{IEEEeqnarray*} 
\end{small}\\
The lower partition of the left and right eigenvectors are equal to $\tilde{\boldsymbol{\omega}}_1$ and $\tilde{\boldsymbol{\nu}}_1$, the left and right eigenvectors of $\tilde{\mathbf{A}}$, respectively. Moreover, the upper partition of $\tilde{\mathbf{w}}_1$ is extended by $n_z=5$ zeros while the upper part of $\tilde{\mathbf{u}}_1$ is equal to $\frac{\tilde{\mathbf{A}}'}{2.0136}\tilde{\boldsymbol{\nu}}_1$.        
\end{ex}
In the following, to complement the result of Theorem~\ref{main_theo_rowlayered_spectra}, we discuss the case of LETSs that are simple cycles.

\begin{theo}\label{CycleAtildSpectra}
Consider a simple cycle $\mathcal{S}$ of length $2a$, and let ${\bf A}$ and $\tilde{\mathbf{A}}_{J \rightarrow 1}$ denote the non-negative transition matrices of $\mathcal{S}$ for flooding and layered decoders, respectively. Then
\begin{itemize}
\item[(a)] The FNF of $\tilde{\mathbf{A}}_{J \rightarrow 1}$ has two irreducible diagonal blocks.
\item[(b)] $\tilde{r}=1$ is an eigenvalue of $\tilde{\mathbf{A}}_{J \rightarrow 1}$, whose right eigenvector $\tilde{\mathbf{u}}_1$ is the same as the right eigenvector of ${\bf A}$, and whose left eigenvector, $\mathbf{\tilde{w}}_1^T$, is an all-one vector of length $2a-n_z$ appended by $n_z$ zeros (up to a permutation), where $n_z$ is the number of zero columns of $\tilde{\mathbf{A}}_{J \rightarrow 1}$.
%\item[(c)] Denoting the number of zero columns of $\tilde{\mathbf{A}}_{J \rightarrow 1}$ by $n_z$, it can be shown that an $2a-n_z$ all-one vector appended by $n_z$ zeros is a left eigenvector, $\mathbf{\tilde{w}}_1^T$, of $\tilde{\mathbf{A}}_{J \rightarrow 1}$ corresponding to $\tilde{r}=1$. 
\end{itemize}
\end{theo}
\begin{proof}
(a) The flooding digraph $D_{f}(\mathcal{S})$ of the simple cycle $\mathcal{S}$ consists of two disconnected directed cycles in opposite directions~\cite{But_SS}. 
By the construction of $D_{l}(\mathcal{S})$ from $D_{f}(\mathcal{S})$, each of the directed cycles will be transformed into a new directed cycle,
that is strongly connected, and thus corresponds to an irreducible diagonal block in the FNF of $\tilde{\mathbf{A}}_{J \rightarrow 1}$.
%whose corresponding adjacency matrix has one irreducible diagonal sub-block. Therefore, the number of irreducible subgraphs of $\tilde{\mathbf{A}}_{J \rightarrow 1}$, in total, is two. 
In fact, the FNF of $\tilde{\mathbf{A}}_{J \rightarrow 1}$ has the following structure 

\begin{small}
\begin{IEEEeqnarray*}{lCl"s}  
\left [\begin{array}{c|c}
\mathbf{0} &\tilde{\mathbf{A}}'\\
\hline
\mathbf{0}& \begin{array}{c|c}
\tilde{\mathbf{A}}_{1_{(a-n_{z_1})\times(a-n_{z_1})}}&\mathbf{0}\\
\hline
\mathbf{0}& \tilde{\mathbf{A}}_{2_{(a-n_{z_2})\times(a-n_{z_2})}}
\end{array} 
\end{array} \right ]_{2a\times 2a},
\end{IEEEeqnarray*}
\end{small}
\hspace*{-0.4cm} in which $\tilde{\mathbf{A}}_{1}$ and $\tilde{\mathbf{A}}_{2}$ are the irreducible blocks corresponding to the two directed cycles
in opposite directions.
 
(b) It is known that for a simple cycle $\mathcal{S}$, the dominant eigenvalue of ${\bf A}$ is $r=1$ with multiplicity $2$~\cite{But_SS}.
It can be seen that Equation \eqref{Atild_u_1Projection} is also applicable to simple cycles. Replacing $r=1$ in \eqref{Atild_u_1Projection}
indicates that a right eigenvector $\mathbf{u}_1$ of ${\bf A}$ corresponding to $r=1$ is also a right eigenvector of $\tilde{\mathbf{A}}_{J \rightarrow 1}$
corresponding to $\tilde{r}=1$. In fact, the eigenvalue $\tilde{r}=1$ is the dominant eigenvalue of $\tilde{\mathbf{A}}_{J \rightarrow 1}$ with multiplicity $2$.
(Each of the diagonal blocks $\tilde{\mathbf{A}}_1$ and $\tilde{\mathbf{A}}_2$ has one dominant eigenvalue $\tilde{r}=1$.)
The all-one vectors with sizes $a-n_{z_1}$ and $a-n_{z_2}$ are the left eigenvectors of $\tilde{\mathbf{A}}_{1}$ and $\tilde{\mathbf{A}}_{2}$, respectively~\cite{But_SS}. 
It is then easy to see that an all-one vector with size $2a-n_{z_1}-n_{z_2}$ appended by $n_z=n_{z_1}+n_{z_2}$ zeros is a left eigenvector of $\tilde{\mathbf{A}}_{J \rightarrow 1}$ (up to a permutation).       
\end{proof}

\section{Calculation of the Failure Probability of a LETS and the Impact of Row Block Permutations}\label{Analysis_Opt_Section}
%\section{The Effect of Row Block Permutations on the Failure Probability of a LETS: Analysis and Optimization} \label{Analysis_Opt_Section} 
In Subsection~\ref{ProbErrLay}, we calculate the error probability of a LETS using the linear state-space model. The effect 
of row block permutations on the error probability of a LETS is then studied in Subsection~\ref{Eff_of_rowPerm_rtild_subsec}. We end this section by 
proposing a two-step search algorithm in Subsection~\ref{Optim_row_lay_sec} to find a row block permutation that minimizes the error floor.

\subsection{Computing the Failure Probability of LETSs in Layered Decoders} \label{ProbErrLay}
As the error indicator function, we consider the projection of the state vector, given in \eqref{ss_lay_main}, onto the 
non-negative left eigenvector $\tilde{\mathbf{w}}_1^T$ of $\tilde{\mathbf{A}}_{J \rightarrow 1}$ corresponding to the dominant eigenvalue $\tilde{r}$:
%From the non-recursive equation \eqref{ss_lay_main} of the state vector in a layered decoder and the fact that at higher iterations the gain values tend to 1, which results in iteration independent matrix of equation \eqref{Atild_iterDepend_AtildIterIndepend}, it is easy to show that as iterations go on, large powers of $\tilde{\mathbf{A}}_{J \rightarrow 1}$ would dominate \eqref{ss_lay_main} and determine its asymptotic behaviour. As a result, following Lemma \ref{AsympAtildBehave}, in order to find an error indicator function and simplify the non-recursive state-space formula related to layered decoder given in \eqref{ss_lay_main}, the projection of the state-vector on the non-negative left eigenvector $\tilde{\mathbf{w}}_1^T$ of $\tilde{\mathbf{A}}_{J \rightarrow 1}$ corresponding to the dominant eigenvalue $\tilde{r}$ is utilized. Therefore, the error indicator function of an LETS in layered decoder is defined as
\begin{IEEEeqnarray*}{l}
\tilde{\beta}^{(\ell)}\triangleq \mathbf{\tilde{w}}_1^T\tilde{\mathbf{x}}^{(\ell)}=\boldsymbol{\gamma}_{ch}^T\L
%\\ 
\ +\sum_{i'=1}^\ell \big({\overset{\triangleleft}{\boldsymbol{\gamma}}_{ex}^{(i')T}}\L_{ex}^{(i'-1)}+\overset{\triangleright}{\mathbf{ \boldsymbol{\gamma}}}_{ex}^{(i')T}\L_{ex}^{(i')}\big),
\IEEEyesnumber
\end{IEEEeqnarray*}
in which the vectors $\boldsymbol{\gamma}_{ch}^T$, $\overset{\triangleleft}{\boldsymbol{\gamma}}_{ex}^{(i')T}$ and $\overset{\triangleright}{\boldsymbol{\gamma}}_{ex}^{(i')T}$ have lengths equal to $a$, $b$ and $b$, respectively. These vectors are defined by  
\begin{IEEEeqnarray*}{lCl"s}
%\begin{equation}
\label{gamma_ch_T}
\boldsymbol{\gamma}_{ch}^T=\mathbf{\tilde{w}}_1^T\sum_{i=1}^\ell \big(\prod_{j=i+1}^{\underrightarrow{\ell}}\tilde{\mathbf{A}}_{J \rightarrow 1}^{(j)}\big)\tilde{\mathbf{B}}^{(i)}\:,\IEEEyesnumber \\
%\end{equation}
%\begin{equation}
\label{gamaExleft}
\overset{\triangleleft}{\boldsymbol{\gamma}}_{ex}^{(i')T}=\mathbf{\tilde{w}}_1^T\big(\prod_{j'=i'+1}^{\underrightarrow{\ell}}\tilde{\mathbf{A}}_{J \rightarrow 1}^{(j')}\big)\overset{\triangleleft}{\mathbf{ B}}_{ex}^{(i')} \ \ \ \ \  i'=1,\dots,\ell\:, \IEEEyesnumber \\
%\end{equation}
%\begin{equation}
\label{gamaExRight}
\overset{\triangleright}{\boldsymbol{\gamma}}_{ex}^{(i')T}=\mathbf{\tilde{w}}_1^T\big(\prod_{j'=i'+1}^{\underrightarrow{\ell}}\tilde{\mathbf{A}}_{J \rightarrow 1}^{(j')}\big)\overset{\triangleright}{\mathbf{ B}}_{ex}^{(i')}\ \ \ \ \  i'=1,\dots,\ell\:. \IEEEyesnumber
%\end{equation}
\end{IEEEeqnarray*}
The mean and variance of $\tilde{\beta}^{(\ell)}$ are calculated as 
\begin{IEEEeqnarray*}{lCl"s}\label{Mean_beta_layered}
\mathbb{E}[\tilde{\beta}^{(\ell)}]= (2/\sigma_{ch}^2) \sum_{k=1}^a(\boldsymbol{\gamma}_{ch}^T)_k %\\
+\sum_{i'=1}^{\ell-1} \big({\overset{\triangleleft}{\boldsymbol{\gamma}}_{ex}^{(i'+1)T}}+\overset{\triangleright}{\mathbf{ \boldsymbol{\gamma}}}_{ex}^{(i')T}\big)\mathbf{m}_{ex}^{(i')}+ \overset{\triangleright}{\mathbf{ \boldsymbol{\gamma}}}_{ex}^{(\ell)T}\mathbf{m}_{ex}^{(\ell)},
\IEEEyesnumber
\end{IEEEeqnarray*}
and
\begin{IEEEeqnarray*}{lCl"s}\label{Var_beta_layered}
\mathbb{VAR}[\tilde{\beta}^{(\ell)}]=(4/\sigma_{ch}^2)\sum_{k=1}^a(\boldsymbol{\gamma}_{ch}^T)_k^2 \\
+\sum_{i'=1}^{\ell-1} \big({\overset{\triangleleft}{\boldsymbol{\gamma}}_{ex}^{(i'+1)T}}+\overset{\triangleright}{\mathbf{ \boldsymbol{\gamma}}}_{ex}^{(i')T}\big)\mathbf{\Sigma}_{ex}^{(i')}\big({\overset{\triangleleft}{\boldsymbol{\gamma}}_{ex}^{(i'+1)}}+\overset{\triangleright}{\mathbf{ \boldsymbol{\gamma}}}_{ex}^{(i')}\big)+%\\
\overset{\triangleright}{\mathbf{ \boldsymbol{\gamma}}}_{ex}^{(\ell)T}\mathbf{\Sigma}_{ex}^{(\ell)}\overset{\triangleright}{\mathbf{ \boldsymbol{\gamma}}}_{ex}^{(\ell)},
\IEEEyesnumber
\end{IEEEeqnarray*}
respectively. The symbol $(.)_k$ is used to represent the $k$th element of the vector inside the parentheses. 
%In the above equations, $m_{ch}$ and $\sigma^2_{ch}$ are the mean and variance of channel LLRs, respectively. 
The $b\times1$ vector $\mathbf{m}_{ex}^{(i')}$ 
and the $b\times b$ matrix $\mathbf{\Sigma}_{ex}^{(i')}$ are the mean and the covariance matrix of the 
inputs from the unsatisfied CNs at iteration $i'$, respectively. 
Since we assume that external inputs to the LETS are independent, the matrix $\mathbf{\Sigma}_{ex}^{(i')}$ is diagonal. 
The mean and the variance of the inputs from unsatisfied CNs are calculated using DE. 
Also, the missatisfied CN gains are calculated based on DE, as explained in Subsection \ref{gain_lay_section}. These gains are 
utilized in iteration dependent gain matrices involved in Equations \eqref{A_jdowntok_l} to \eqref{Bex_backtriangel_l}. 
Finally, the vector $\boldsymbol{\gamma}_{ch}^T$ in \eqref{gamma_ch_T} as well as the set of vectors $\overset{\triangleleft}{\boldsymbol{\gamma}}_{ex}^{(i')T}$ and $\overset{\triangleright}{\boldsymbol{\gamma}}_{ex}^{(i')T}$ for $i'=1,\dots,\ell$ in \eqref{gamaExleft} and \eqref{gamaExRight}, respectively, are used in Equations \eqref{Mean_beta_layered} and \eqref{Var_beta_layered} to find the mean and the variance of the error indicator function, respectively. 
Finally, assuming that $\tilde{\beta}^{(\ell)}$ is Gaussian, the failure probability of a LETS, $\mathcal{S}$, is obtained by
\begin{equation}\label{P_fail_Layered}
{P_e(\mathcal{S})=\lim_{\ell\to\infty}\text{Pr}\{\tilde{\beta}^{(\ell)}<0\}=\lim_{l\to\infty}Q\bigg(\frac{\mathbb{E}[\tilde{\beta}^{(\ell)}]}{\sqrt{\mathbb{VAR}[\tilde{\beta}^{(\ell)}]}}\bigg).}
\end{equation}
In practice, the error probability, given above, converges rather fast within a few iterations.
%The above limit would result in the error probability. However, in practice, it is required to calculate \eqref{P_fail_Layered}, only, for a few iterations.

Suppose that the set $\{{\cal T}_i\}$ contains all the dominant LETSs of an LDPC code. Let ${\cal E}_i$ denote the event that the decoder is failed 
due to ${\cal T}_i$. To estimate the error floor of the LDPC code, we partition $\{{\cal T}_i\}$ in accordance with the TSLP of TSs.
This implies that within each class of LETSs, we first identify different non-isomorphic structures and then among those LETSs with the same structure, we 
identify different TSLPs. We then accordingly partition the class into different {\em groups}, where all the LETSs within a group have the same TSLP.
Suppose that $\mathcal{S}_i$ is the representative of the $i$th TS group with the size $\Upsilon_i$. (Note that having the same TSLP implies that all the LETSs within a group have the same system matrices.) We then have the following approximation for the error floor of the code:
\begin{equation}
P_{f} \approx P\bigg(\bigcup_i\mathcal{E}_i\bigg) \stackrel{<}{\approx} \sum_i \Upsilon_i P_e(\mathcal{S}_i),
%P(\cup_i {\cal E}_i) \stackrel{<}{\approx} \sum_i \Upsilon_i P_e(\mathcal{S}_i),
\end{equation}
where the last step follows from the union bound.
%Using the union bound, we then approximate the error floor of the code by
%\begin{equation}
%P_f\approx \sum_i \Upsilon_i P_e(\mathcal{S}_i),
%\end{equation}
%where $\mathcal{S}_i$ is the representative of the $i$th TS group with the size $\Upsilon_i$. (Note that having the same TSLP implies that all the LETSs within a group have the same system matrices.) 

We recall that for a LETS $\mathcal{S}$ with an irreducible flooding transition matrix ${\bf A}$, all the elements of the left eigenvector of ${\bf A}$ 
corresponding to the dominant eigenvalue $r$ are positive. For the layered transition matrix $\tilde{\mathbf{A}}_{J \rightarrow 1}$, however,
there are some zeros in the left eigenvector $\tilde{\mathbf{w}}_1^T$ corresponding to the dominant eigenvalue $\tilde{r}$. This implies that the corresponding state variables have no effect on the value of the error indicator function $\tilde{\beta}^{(\ell)}$. In fact, only the state variables corresponding to the 
irreducible part of the transition matrix $\tilde{\mathbf{A}}_{J \rightarrow 1}$ are the ones that determine the value of $\tilde{\beta}^{(\ell)}$, and are the main contributors to the growth of erroneous messages inside $\mathcal{S}$.

\subsection{Effect of Different Row Block Permutations on the Dominant Eigenvalue of $\tilde{\mathbf{A}}_{J \rightarrow 1}$}\label{Eff_of_rowPerm_rtild_subsec}
In the asymptotic regime of $\ell \rightarrow \infty$, the dominant eigenvalue $\tilde{r}$ of the transition matrix $\tilde{\mathbf{A}}_{J \rightarrow 1}$ plays an important role in the failure probability of a LETS in a layered decoder.  The value of $\tilde{r}$ can however change 
depending on the order in which different layers within the LETS are updated. To simplify the discussions, we assume that in the layered decoding, the row blocks of the parity check matrix are updated based on their increased indices, i.e., starting from the first row block all the way down to the row block number $m_b$. We thus associate the different orderings of row updates with different permutations of the row blocks of the parity-check matrix. In this part, we consider the effect  of different row block permutations on $\tilde{r}$. 
 
In general, for a LETS with $J$ layers, there are $J!$ different layer permutations. We denote the set of all possible permutations of $J$ layers by 
$\Pi_J$. A given permutation $\pi \in \Pi_J$ is identified by the sequence $\pi(1),\ldots,\pi(J)$, or more briefly by $\pi_1,\ldots,\pi_J$.
 
%The total number of layer permutations for a given LETS structure with maximum layer of $J$ is $J!=(J)\times (J-1)\times \dots \times 1$. The set of all the $J!$ permutations of the layers $\{1,2,\dots,J\}$ is denoted by $\Pi_J$. Also, $\varpi\in\Pi_J$ is defined as $\varpi=:\{\pi_1,\pi_2,\dots,\pi_J\}$ which is used to denote a specific permutation of the layers.
\begin{pro}\label{cyclic_perm_lem}
Consider the application of a permutation $\pi$ to the $J$ layers of a LETS $\mathcal{S}$, and denote the corresponding 
transition matrix of $\mathcal{S}$ in the layered decoder by $\tilde{\mathbf{A}}_{\pi_J \rightarrow \pi_1}$. 
The dominant eigenvalue of $\tilde{\mathbf{A}}_{\pi_J \rightarrow \pi_1}$ is then invariant to any cyclic shift of the 
permutation $\pi$. 
%All the cyclic shifts of a permutation $\varpi\in\Pi_J$ with $m_s\times m_s$ system matrix $\tilde{\mathbf{A}}_{\pi_J \rightarrow \pi_1}$ result in new system matrices whose dominant eigenvalues are equal to $\tilde{r}$, the dominant eigenvalue of $\tilde{\mathbf{A}}_{\pi_J \rightarrow \pi_1}$.      
\end{pro}
\begin{proof}
Suppose that the right eigenvector of $\tilde{r}$ for $\tilde{\mathbf{A}}_{\pi_J \rightarrow \pi_1}$ is $\tilde{\mathbf{u}}_1$. We then have
\begin{equation*}
\underbrace{\mathcal{A}_{\pi_J } \mathcal{A}_{\pi_{J-1} } \dots \mathcal{A}_{\pi_1 }}_{\tilde{\mathbf{A}}_{\pi_J \rightarrow \pi_1}}\tilde{\mathbf{u}}_1=\tilde{r}\tilde{\mathbf{u}}_1.
\end{equation*}
Multiplying both sides of this equation sequentially with $\mathcal{A}_{\pi_1 }, \mathcal{A}_{\pi_2 }, \ldots, \mathcal{A}_{\pi_{J-1} }$, results in
% \begin{equation*}
%\underbrace{\tilde{\mathbf{A}}_{\pi_1 }\tilde{\mathbf{A}}_{\pi_J }\dots\tilde{\mathbf{A}}_{\pi_2 }}_{\text{1 cyclic shift}}(\tilde{\mathbf{A}}_{\pi_1 }\tilde{\mathbf{u}}_1)=\tilde{r}(\tilde{\mathbf{A}}_{\pi_1 }\tilde{\mathbf{u}}_1).
%\end{equation*}
%Accordingly, by multiplying the consecutive equations with matrices $\tilde{\mathbf{A}}_{\pi_2 },\dots,\tilde{\mathbf{A}}_{\pi_{J-1} }$, we have
\begin{small}
\begin{IEEEeqnarray*}{lCl"s}
\underbrace{\mathcal{A}_{\pi_1 } \mathcal{A}_{\pi_J }\dots \mathcal{A}_{\pi_2 }}_{\text{1 cyclic shift}}(\mathcal{A}_{\pi_1 }\tilde{\mathbf{u}}_1)=\tilde{r}(\mathcal{A}_{\pi_1 }\tilde{\mathbf{u}}_1), \\
\underbrace{\mathcal{A}_{\pi_2 } \mathcal{A}_{\pi_1 } \mathcal{A}_{\pi_J } \dots \mathcal{A}_{\pi_3 }}_{\text{2 cyclic shifts}}(\mathcal{A}_{\pi_2 }\mathcal{A}_{\pi_1 } \tilde{\mathbf{u}}_1)=\tilde{r}(\mathcal{A}_{\pi_2 } \mathcal{A}_{\pi_1 } \tilde{\mathbf{u}}_1), \\
\ \ \ \ \ \ \ \ \ \ \ \ \  \vdots\\
\underbrace{\mathcal{A}_{\pi_{J-1} }\dots\mathcal{A}_{\pi_1 }\mathcal{A}_{\pi_J }}_{\text{$J-1$ cyclic shifts}}(\mathcal{A}_{\pi_{J-1}}\dots \mathcal{A}_{\pi_1 }\tilde{\mathbf{u}}_1)=\tilde{r}(\mathcal{A}_{\pi_{J-1}}\dots\mathcal{A}_{\pi_1 }\tilde{\mathbf{u}}_1),
\end{IEEEeqnarray*}
\end{small}
respectively. As can be seen, $\tilde{r}$ is the eigenvalue of all the transition matrices resulted from different cyclic shifts of the permutation $\pi$.
%related to cyclic shifts of the permutation $\varpi$ while the right eigenvectors are the scaled versions of $\tilde{\mathbf{u}}_1$. 
The above equations can be written for all the eigenvalues of $\tilde{\mathbf{A}}_{\pi_J \rightarrow \pi_1}$, and thus, $\tilde{r}$ remains the dominant eigenvalue of all such transition matrices. 
\end{proof}

In the following proposition, whose proof is given in the appendix, we prove that reversing the layer permutation does not change $\tilde{r}$.

\begin{pro}\label{flip_perm_proposition}
Consider the transition matrix $\tilde{\mathbf{A}}_{\pi_J \rightarrow \pi_1}$ of a LETS with $J$ layers corresponding to a layer permutation $\pi$.
Then, the eigenvalues of $\tilde{\mathbf{A}}_{\pi_J \rightarrow \pi_1}$ are invariant under the reversing of $\pi$, i.e., 
the transition matrix $\tilde{\mathbf{A}}_{\pi_1 \rightarrow \pi_J}$ corresponding to the reverse permutation $\pi_J,\ldots,\pi_1$ has the same eigenvalues as
$\tilde{\mathbf{A}}_{\pi_J \rightarrow \pi_1}$.
%Regarding the row layered decoder, let $\varpi=:\{\pi_1,\pi_2,\dots,\pi_J\}$ be a permutation of row layers in an LETS, $\mathcal{S}$, with $m_s\times m_s$ system matrix $\tilde{\mathbf{A}}_{\pi_J \rightarrow \pi_1}$ whose $i$th eigenvalue and its corresponding right eigenvector are equal to $\tilde{r}_i$ and $\tilde{\mathbf{u}}_i$, respectively. Then, all the eigenvalues of the system matrix corresponding to the reversed (flipped) permutation, denoted by $\tilde{\mathbf{A}}_{\pi_1 \rightarrow \pi_J}$, is equal to the eigenvalues of $\tilde{\mathbf{A}}_{\pi_J \rightarrow \pi_1}$, in particular, the dominant ones. 
\end{pro}
Based on Propositions~\ref{cyclic_perm_lem} and \ref{flip_perm_proposition}, we have the following result.

\begin{cor}\label{MaxDominantPerm}
For a LETS with $J$ layers ($J\geq 3$), the number of distinct dominant eigenvalues of the layered transition matrix 
corresponding to different permutations of the layers is upper bounded by $\frac{(J-1)!}{2}$. 
\end{cor}
%\begin{proof}
%The total number of permutations is $J!$. According to Lemmas \ref{cyclic_perm_lem} and \ref{flip_perm_lem}, the cyclic permutations as well as their flipped versions have the same dominant eigenvalues. Excluding those cases from the total number of permutations results in $\frac{J!}{2J}=\frac{(J-1)!}{2}$.
%\end{proof}
%\begin{rem}
%It is noted that the minimum possible $J$ of an LETS, $\mathcal{S}$, is $J=2$. The reason is that, based on LETS definition, the VN degrees of $\mathcal{S}$ must be at least $2$ and in QC-LDPC codes, a VN cannot be connected to two CNs from the same row layer. It is trivial that the dominant eigenvalue of different permutations of layers of an LETS with $J=2$ is unique as the cyclic shift of layers and its reversed versions generate the same system matrices.  
%\end{rem}
Corollary~\ref{MaxDominantPerm} implies that for a LETS with three layers ($J=3$), all the layered transition matrices corresponding to 
different layer permutations have the same dominant eigenvalue.
%\begin{cor}\label{MaxDominantPermCorJ3}
%Let the maximum number of layers in an LETS, $\mathcal{S}$, under row layered schedule be $3$. Then, the dominant eigenvalue of the system matrices corresponding to all the possible row schedules are equal.
%\end{cor}
%\begin{proof}
%This is the direct result of Theorem \ref{MaxDominantPerm} when $J=3$.
%\end{proof}
\begin{ex}
The Tanner $(155,64)$ LDPC code has $3$ row layers, and based on Corollary \ref{MaxDominantPerm}, the dominant eigenvalue of the layered transition matrices of all of its LETSs are invariant under all the possible ($6$) layer permutations. 
%The dominant eigenvalues of different LETS structures of this code such as $(5,3),(8,2),\dots$ are invariant under all the possible permutations of the layers which is $3!=6$. 
\end{ex}
%\begin{rem}
%The result of Theorem \ref{MaxDominantPerm} does not depend on a specific LETS structure and is general. However, depending on the structures of different LETSs of a code, the distinct dominant eigenvalues, except the LETSs in which only 3 row layers involved, can be much less than $\frac{(J-1)!}{2}$. For example, consider the codes $\mathcal{C}_1$ and $\mathcal{C}_2$ of Fig. \ref{par_used_lay} with $7$ and $6$  row layers respectively. The dominant LETS structure of $\mathcal{C}_1$ which is a $(5,5)$ ETS can have at most $360$ different eigenvalues related to different row permutations while, in practice, this structure, only, generates $8$ separated eigenvalues. On the other hand, the dominant $(7,1)$ LETS of $\mathcal{C}_2$ produces $57$ different dominant eigenvalues related to different permutations of the layers which is pretty close to, $60$, the upper bound obtained from Theorem \ref{MaxDominantPerm}.
%\end{rem}

\subsection{Optimizing the Row Layered Schedule}\label{Optim_row_lay_sec}
The contribution of each LETS $\mathcal{S}$ to the error floor of an LDPC code, decoded by a row layered iterative algorithm, depends on the 
distribution of messages entering $\mathcal{S}$ through its missatisfied and unsatisfied CNs as well as the internal structure of $\mathcal{S}$.
In the linear state-space model used in this work, the internal structure is reflected in the system matrices in general, and the transition matrix $\tilde{\mathbf{A}}_{J \rightarrow 1}$, in particular. More specifically, our results show that the dominant eigenvalue $\tilde{r}$ of the transition matrix plays an important role in the growth rate 
of erroneous messages inside the subgraph of $\mathcal{S}$, with larger values of $\tilde{r}$ generally corresponding to larger growth rate.
%The system matrix of an LETS, $\mathcal{S}$, in the layered schedule decoder, $\tilde{\mathbf{A}}_{J \rightarrow 1}$, and its corresponding dominant eigenvalue $\tilde{r}$, plays a crucial role in determining the contribution of $\mathcal{S}$ in the error floor of a given LDPC code. In other words, as iterations go on, these are the dominant positive eigenvalue and its corresponding eigenvectors of matrix $\tilde{\mathbf{A}}_{J \rightarrow 1}$ at the heart of the linear model that highly determine the growth rate of erroneous messages inside the subgraph and accordingly affect the harmfulness of an LETS structure. 
As we discussed in Subsection~\ref{Eff_of_rowPerm_rtild_subsec}, however, the transition matrix $\tilde{\mathbf{A}}_{J \rightarrow 1}$ and $\tilde{r}$ can change
with layer permutations. 
%are not unique and different permutations of horizontal layers might generate various system matrices with remarkably different dominant eigenvalues. 
Moreover, different layer permutations can result in notable change in the TSLP and thus the distributions 
%statistical properties of the surrounding graph of an LETS, in terms 
of the external messages entering the TS.
%subgraph via unsatisfied and mis-satisfied CNs. 
As a result, different row layered schedules can potentially produce error floors that are considerably different. 
This motivates the search for layer permutations which result in low error floors.
In the following, we show that our proposed model can be used not only for the error floor estimation of layered decoders, 
but also as an efficient tool to find row layered schedules with low error floors.
%to obtain an improved error floor performance.

Generally, for a given QC-LDPC code whose base matrix $\mathbf{H}_b$ has $m_b$ rows, there are ${m_b}!$ different row layered schedules corresponding to different row permutations of $\mathbf{H}_b$. The complexity of an exhaustive search among all such schedules based on the exact estimation of the error floor can be prohibitive for relatively large values of $m_b$. The main source of complexity in our model is to obtain the distribution of external messages by DE. 
(Recall that, for each iteration, $2 |E_b|$ distributions are needed to be calculated, where $|E_b|$ is the number of edges in the base graph,
which is equal to the number of nonzero elements of  $\mathbf{H}_b$.) 
To simplify the search among the ${m_b}!$ different row layered schedules, rather than the derivation of such distributions for each schedule, 
we select one schedule, say the one corresponding to the original order of the rows, and then at each iteration $\ell$ and for each layer $j$, 
we derive the average distributions of CN to VN messages and VN to CN messages, denoted by $\stackrel{\leftarrow}{\bar{\psi}_{\ell}^j}$
and $\stackrel{\rightarrow}{\bar{\psi}_{\ell}^j}$, respectively, where the average is taken over all the corresponding distributions within layer $L_j$.
These average distributions are then used to represent all the CN to VN and VN to CN distributions in the $j$th layer of decoding 
regardless of the schedule.
%While in our proposed model, it is required to calculate the probability distributions for a specific layers order, this is not reasonable to take the same approach for the purpose of optimization where there are $m_b!$ different layer orders. Thus, we sacrifice accuracy for the lower complexity. 
%
%\begin{rem}
%Based on Lemma \ref{flip_perm_lem}, for a given LETS, the dominant eigenvalues related to cyclic shifts as well as their flipped versions are the same. However, the dominant eigenvalue is not the only factor determining the harmfulness of an LETS in layered schedule. In other words, changing the order of layers might significantly alter the distribution of the messages entering to missatisfied and unsatisfied CNs and lead to completely different failure probabilities even for those schedules generating the same dominant eigenvalue. This is why the optimization of the layers schedule in the following is done exhaustively over all the possible permutations of the row layers. This approach, most of the time, is feasible as the number of row layers of the practical codes, in particular, the high rate ones are, typically, small.
%\end{rem}
%
%In the first phase of optimization, we use the density evolution result of a specific layer order, such as the row permutation of the original code, as the input of our model for all the possible permutations. In this regard, given the density evolution result of the candidate schedule, a single CN to VN probability distribution is calculated for each row layer by averaging over different type of CN to VN distributions per each layer. 
As an example, for the base graph of Fig. \ref{BaseTannerGraph}, the distribution $\stackrel{\leftarrow}{\bar{\psi}_{\ell}^3}$ is the average of 
$\psi_{\ell}^{[2\leftarrow 3]}$ and $\psi_{\ell}^{[4\leftarrow 3]}$, and $\stackrel{\rightarrow}{\bar{\psi}_{\ell}^3}$ is the average of 
$\psi_{\ell}^{[2\rightarrow 3]}$ and $\psi_{\ell}^{[4\rightarrow 3]}$.

% The resulted distributions are used as the unsatisfied CNs input statistics in different schedules. The same averaging procedure is done over the VN to CN messages distributions of different row layers. Regarding the $3$rd layer in Fig. \ref{BaseTannerGraph}, this can be done by computing the average of $\psi_{\ell}^{[2\rightarrow 3]}$ and $\psi_{\ell}^{[4\rightarrow 3]}$. Then, these average VN to CN probabilities are used to calculate a single virtual VN to mis-satisfied CN distributions per each layer and the results are used to compute the multiplicative gains of the layers. In this regard, it is assumed all the $d_{c_j}-2$ external messages of a missatisfied CN related to a specific layer, Fig. \ref{mis_sat_53}, have the same distributions as the average VN to CN distribution of that layer computed based on the explanations above.\footnote{Another possible approach is that to repeat all the above steps for a couple of different permutations of layers selected at random and, then, another average is taken over all the average distributions of each permutation. By the way, in this paper, we only use the first approach, which is based on the permutation of the original code, that is sufficient for our purpose.}

The averaging process just explained will result in a less accurate estimate of the failure rate of a LETS. In fact, by ignoring the effect of scheduling in 
the distribution of external messages and by the averaging, we only observe the effect of the schedule on the internal messages of the LETS.
% In the approach just explained, the focus is on the growth rate of the messages in different layers rather than their differences within the same layer and this is achieved by the described averaging over the distributions per layers.
To find the schedule with the lowest error floor, we then perform the search in two steps. In the first step, we use the above approximation/simplification and search among
all the $m_b!$ schedules to find a few candidates that have lower error floors. In the second step, we examine the candidate schedules by calculating the error floor estimates accurately (by considering the effect of scheduling in the distribution of external messages), and find the one with the lowest error floor.

%In the second phase, when an initial estimate of error for different schedules is obtained, for a handful of the best schedules, the actual probability distributions obtained from DE technique is used to verify the results and find one of the best row schedules.\\
%All the above steps will be discussed in section \ref{Simulation_sec_lay}.

\section{Simulation Results and Discussions} \label{sec4}\label{Simulation_sec_lay}
In this section, we investigate the accuracy of the proposed linear state-space model in estimating the error floor of row layered SPA through simulations. 
%in the performance of our proposed model is investigated via simulation results. 
We consider two QC-LDPC codes $\mathcal{C}_1$ and $\mathcal{C}_2$, whose exponent matrices are shown in Figs.~\ref{C1_640par} and \ref{C2_576par}, respectively.
%(The ``-1'' elements of the exponent matrices represent all-zero blocks within the parity-check matrices.) 
$\mathcal{C}_1$ is a $(640,192)$ variable-regular code with $d_v=5$ and irregular CN degrees~\cite{Ryan1} and $\mathcal{C}_2$ is a $(576,432)$ irregular 
code used in Wimax standard~\cite{wimax_standard}. The lifting degrees for the two codes are $64$ and $24$, respectively.

%\subsection{Scheduling Impact on the Error Floor Performance}
\begin{figure*}
  \centering  
 \subfloat[Exponent matrix of $\mathfrak{C}_1$: $(640, 192)$ QC-LDPC code with lifting degree $64$.]{\label{C1_640par}
  $ 
  \scriptsize{
\begin{array}{|c| c |c |c |c |c |c| c| c| c|}
\hline
8&-1&26&62&-1&-1&59&19&-1&60\\
\hline
51&9&50&-1&39&-1&4&-1&25&26\\
\hline
-1&32&10&7&56&52&41&55&61&41\\\hline
46&11&-1&43&-1&63&8&51&37&-1\\\hline
47&7&7&50&49&53&-1&12&-1&-1\\\hline
31&-1&-1&-1&31&38&-1&-1&23&48\\\hline
-1&25&21&56&59&30&27&23&27&18\\\hline
\end{array}}
$}\\
\subfloat[Exponent matrix of $\mathfrak{C}_2$: $(576, 432)$ Wimax QC-LDPC code with lifting degree 24.]{\label{C2_576par}
$\scriptsize{
\begin{array}{|c |c |c |c |c |c |c| c| c |c |c| c| c| c |c| c| c| c| c| c| c| c| c| c|}
\hline
-1&20&-1&7&-1&-1&3&6&4&-1&-1&21&7&13&19&23&5&23&0&0&-1&-1&-1&-1\\\hline
10&-1&3&17&8&-1&-1&-1&-1&17&10&2&9&10&8&14&9&6&-1&0&0&-1&-1&-1\\\hline
-1&-1&5&-1&-1&15&9&-1&17&16&-1&9&1&18&11&7&15&1&20&-1&0&0&-1&-1\\\hline
16&0&-1&-1&15&-1&-1&0&12&-1&20&3&23&2&21&9&3&4&-1&-1&-1&0&0&-1\\\hline
-1&13&15&20&-1&6&18&-1&-1&-1&-1&21&19&0&0&18&15&6&-1&-1&-1&-1&0&0\\\hline
19&-1&-1&-1&3&7&-1&8&-1&18&7&17&21&21&6&16&2&22&0&-1&-1&-1&-1&0\\\hline
\end{array}
}
$
}
\caption{QC-LDPC codes used for simulations. The entries of the matrices, that are not equal to $-1$, represent the right circular shift of the identity matrix to create the corresponding block of the parity-check matrix. The $-1$ entries represent zero blocks.}
\label{par_used_lay}
\end{figure*}

The Monte Carlo simulation results as well as estimation results for the layered decoding of $\mathcal{C}_1$ where the order of row layers are the same as that of 
Fig. \ref{C1_640par}, for different saturation levels are presented in Fig. \ref{qc640_lay_fer_difSAT}. 
The most harmful structure of this code in the error floor region is the $(5,5)$ LETS, shown in Fig.~\ref{55qc640}, with multiplicity $64$. 
All the $(5,5)$ LETSs have the same TSLP. As expected, by increasing the saturation level, the error floor is reduced. 
The figure also shows a good match between the linear model estimation results and simulation results. The slight over-estimation of error floor is attributed to
the linear approximation of missatisfied CN operations (see~\cite{Ali-TCOM}, for more information).

% are well-matched with the simulated frame error rate (FER) curves. 
We now investigate the effect of row block permutations on the error floor of $\mathcal{C}_1$. $\mathcal{C}_1$ has $7$ row layers, each corresponding to one of the row blocks of the parity-check matrix. We label these row blocks with numbers $1$ to $7$ based on their indices in Fig. \ref{C1_640par}, i.e., the first row block in Fig. \ref{C1_640par} is labeled by $1$, the second by $2$, and so on. We then represent different schedules with different permutations of numbers from $1$ to $7$.
For example, permutation $(1, 2, 3, 4, 5, 6, 7)$ corresponds to a schedule which updates the row layers in the same order as they appear in Fig. \ref{C1_640par}.

%the order in which these row layers are updated. For example, if the layers of $\mathcal{C}_1$  are updated in the order they appear in Fig. \ref{C1_640par}, are labelled with numbers 1 to 7, different schedules can be represented by different permutations of numbers from 1 to 7. For example, in Fig. \ref{qc640_lay_fer_difSAT}, the order of layers is the same as Fig. \ref{C1_640par} that is denoted by $\{1,2,3,4,5,6,7\}$. This schedule produces a layered decoder system matrix for $(5,5)$ LETS whose dominant eigenvalue is $\tilde{r}=13.969$. 

\begin{figure}
\centering
\includegraphics[width=3.6in]%{qc640_layered_SS_same_g_as_flood.eps}
{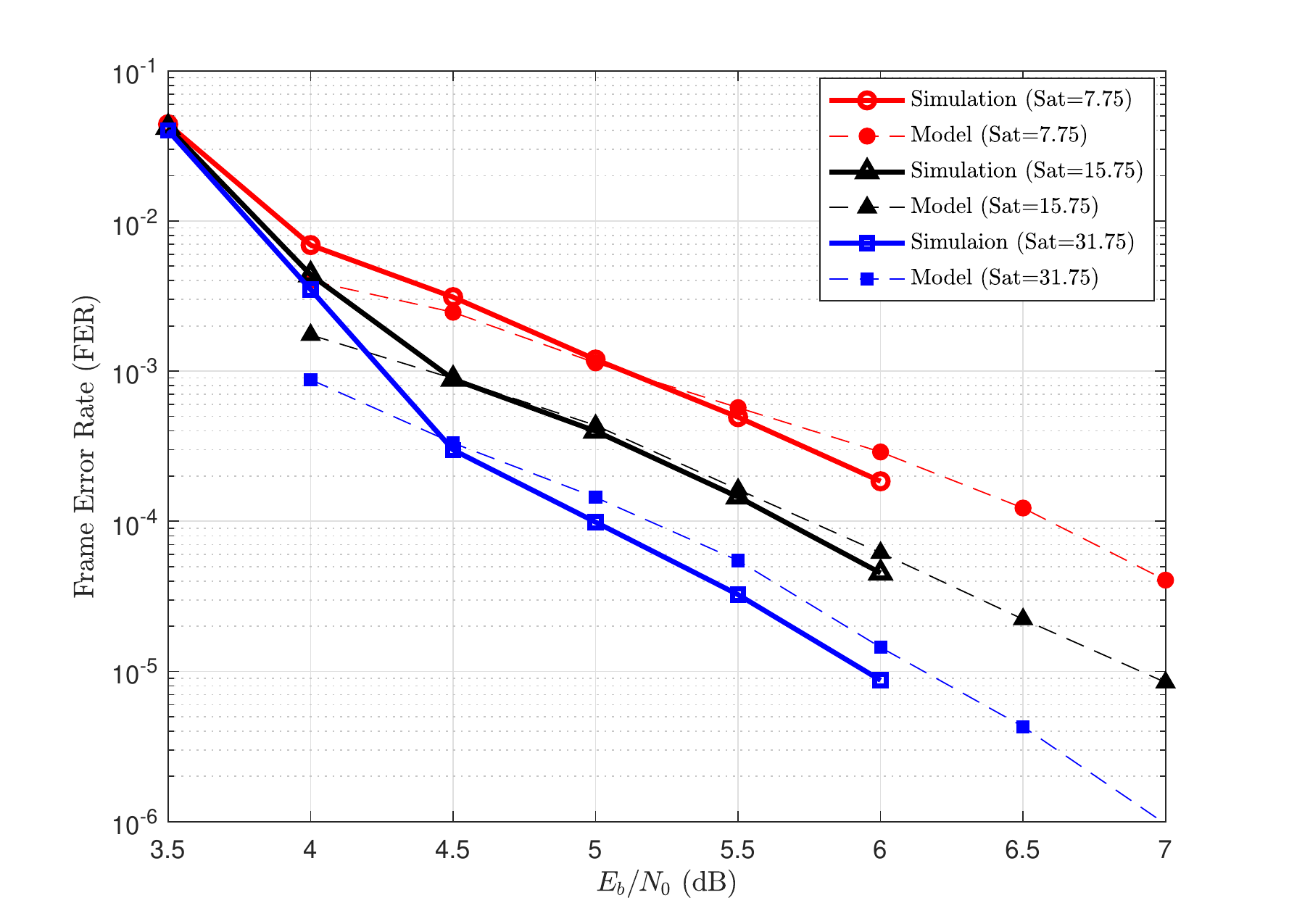}
\caption{Simulation and estimation results of $\mathcal{C}_1$ for different saturation levels. The maximum number of iterations $I_{max}=30$. }
\label{qc640_lay_fer_difSAT}
\end{figure}
\begin{figure}
\centering
\includegraphics[width=1.8in]{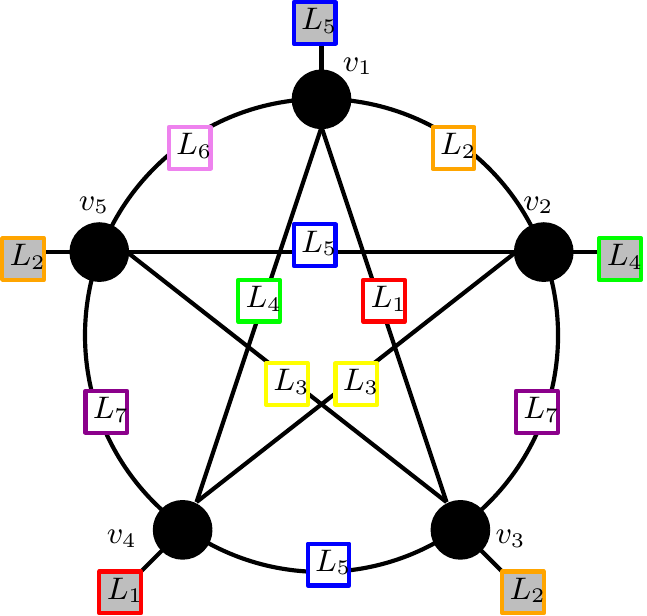}
\caption{The $(5,5)$ LETS structure of $\mathcal{C}_1$ in which the row layers for different CNs are shown.  }
\label{55qc640}
\end{figure}
In order to investigate the effect of different row schedules on the error floor performance of $\mathcal{C}_1$, 
the technique of Section \ref{Optim_row_lay_sec} is used. The total number of possible row permutations for $\mathcal{C}_1$ is $7!=5040$. 
These permutations in general correspond to different system matrices for the $(5,5)$ LETS. 
%system matrix with distinct dominant eigenvalues. 
By using a single application of DE, as discussed in Section \ref{Optim_row_lay_sec}, the failure rate of the $(5,5)$ LETS for different schedules is 
approximated. The results for SNR of $6$ dB and saturation level of $31.75$ are provided in Fig. \ref{qc640_lay_Opt_31_75} for all the schedules. 
%To generate these results, the SNR value is set to $6$ dB and the saturation level, for all the schedules, is fixed at 31.75. 
As can be seen, the estimation results are partitioned into $8$ groups, separated by vertical dotted lines. 
The transition matrices of the schedules within each group have the same dominant eigenvalue $\tilde{r}$. 
We have also sorted the groups according to the increasing value of $\tilde{r}$. 
The eight different values of $\tilde{r}$ are shown in Fig. \ref{qc640_lay_Opt_31_75}, and range from $12.402$ to $16.125$. Note that based on
Corollary~\ref{MaxDominantPerm}, the upper bound on the number of different $\tilde{r}$ values for different schedules for the case where the LETS has $7$ layers is $360$, in general.
Fig. \ref{qc640_lay_Opt_31_75} shows the trend that increasing $\tilde{r}$, on average, increases the failure probability of the LETS. 
Within each group, however, the variance of the error probabilities is still rather large. This implies that while $\tilde{r}$ plays an important role in the failure probability of a LETS, there are also other factors, including the layering structure of the TS reflected through the system matrices, that affect the harmfulness.
(Note that, in this analysis,  although the distributions of the external messages of the LETS for different layers, obtained through DE, 
remains constant for different schedules, but the assignment of different CNs and VNs of the LETS to different layers will change due to the change of schedule.
As a result, the distributions of messages associated with these nodes will also change in different schedules.)
% such as the statistical properties of the surrounding neighbourhood of an LETS which play a crucial role in determining the harmfulness. 
%These factors are taken into account by calculation of the linear gains as well as unsatisfied CNs inputs in the state-space model. It should be noticed that, based on our experiments, we have observed neither the best schedule nor the worst one is unique. 

To find a schedule with low error floor, in the next step, we select a few schedules whose transition matrices have the minimum dominant eigenvalue, $\tilde{r}=12.402$, and 
result in the lowest error rates in the first step. We then apply our estimation technique accurately to find the error floor of these candidate schedules, and select the one that has the lowest error floor.
%and estimate their error floor using the exact that have the lowest error rates are tested with the exact DE results in which the layers permutation is considered and 
As a result, we obtain the schedule corresponding to the permutation $(2,3,1,7,4,5,6)$, shown in Fig. \ref{qc640_lay_Opt_31_75} with a full diamond. 
For comparison, we have also selected one of the schedules with the worst error floor, $(4,6,5,7,3,1,2)$, as well as 
the original one, $(1,2,3,4,5,6,7)$. These schedules are specified in Fig. \ref{qc640_lay_Opt_31_75} by a full triangle and a full circle, respectively. 
The simulation and estimation results of these three schedules are presented in Fig. \ref{qc640_lay_diffSched_31_75}. 
The maximum number of iterations and the saturation level are $I_{max}=30$ and $31.75$, respectively. 
For comparison, we have also included the FER of the flooding schedule with maximum number of iterations $I_{max}=60$ in Fig. \ref{qc640_lay_diffSched_31_75}. 
As can be seen, all the estimation results match closely with the corresponding simulations.
Remarkably, there is a substantial difference between the FER of the best and worst layered schedules in the error floor region, 
with the performance of flooding schedule in the middle. This demonstrates the gain that one can obtain in performance by 
properly choosing the updating order of layers in a layered decoder, basically at no cost. It also shows that a layered decoder 
can, in general, have a better or a worse performance compared to its flooding counterpart. By proper permutation of row layers, 
the layered decoder not only has a faster convergence speed compared to a flooding decoder but also can have a better performance.

%The remarkable feature of this figure is the difference of the best and worst schedules error rates which is more than 2 order of magnitudes at high SNRs. This reveals the significance of scheduling that can impact the error rate even more than the saturation parameter. Furthermore, the error curve of the flooding schedule is between the two extreme cases of layered decoder implying that, generally, the row layered decoder can have better or worse error floor performance than its flooding schedule counterpart. Interestingly, the error floor of the layered decoder with the original code schedule and the flooding schedule decoder are overlapped. In all the cases, the linear model closely approximate\footnote{In some of the examples of this section, an approximation error might be observed. These errors are mainly due to the linearisation involved in the model. In other words, the multiplicative gains are obtained by Taylor expansion at the vicinity of zero while the true operational point of the original non-linear model depends on flowing messages inside the subgraph, external messages, scheduling, saturation level and so on. Therefore, selection of zero as the linearisation point might impose some approximation errors. Naturally, the amount of these errors are not fixed and by changing the parameters might be altered or even vanished.} the error probabilities as well as the relative floor of different schedules. 
\begin{figure}
\centering
\includegraphics[width=3.75in]{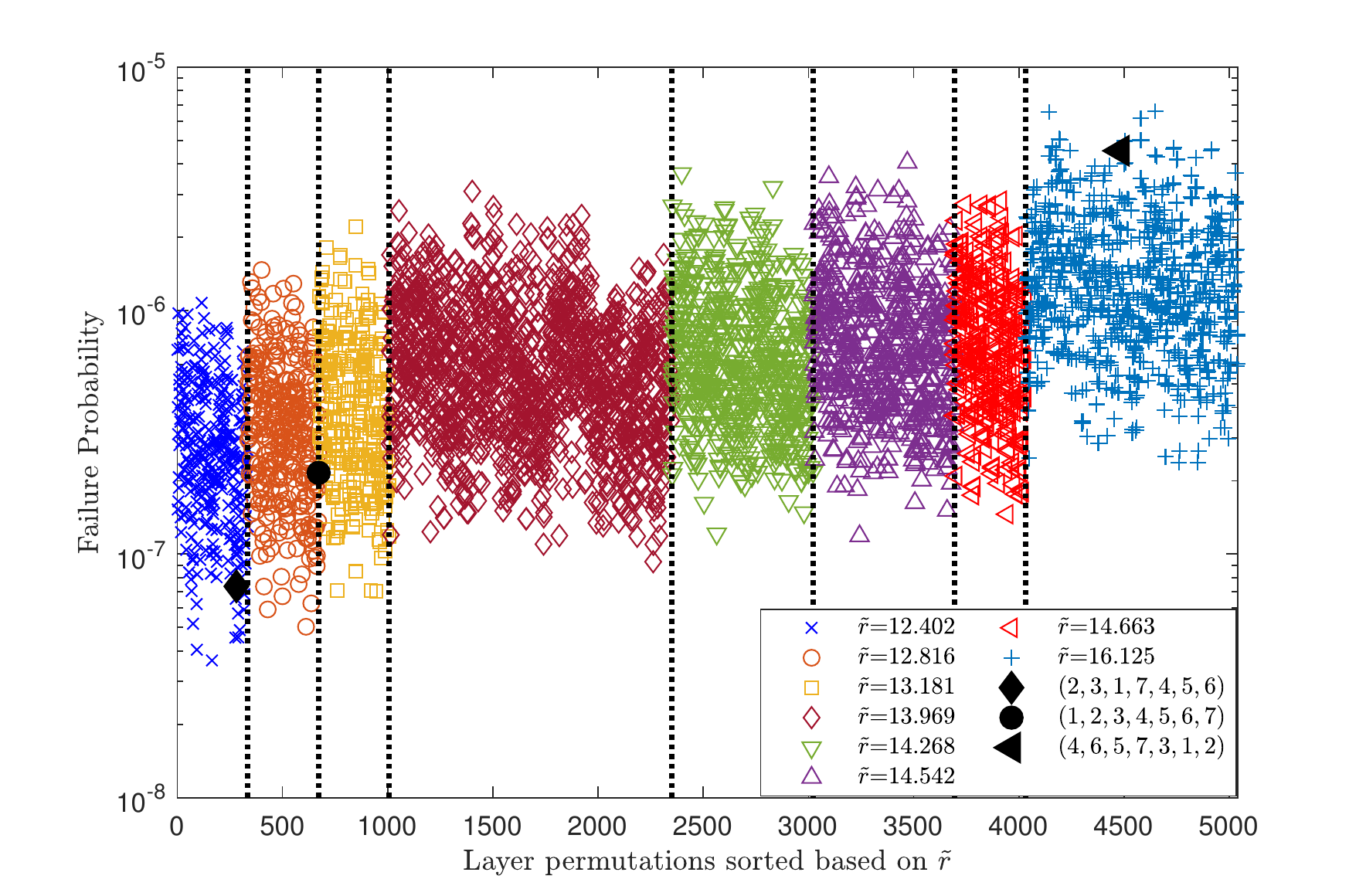}
\caption{The approximate estimate of the failure probability of the $(5, 5)$ LETS of $\mathcal{C}_1$ for various row layered schedules at $E_b/N_0=6$ dB and saturation level $31.75$. 
The schedules are sorted based on $\tilde{r}$. }
\label{qc640_lay_Opt_31_75}
\end{figure}

\begin{figure}
\centering
\includegraphics[width=3.6in]{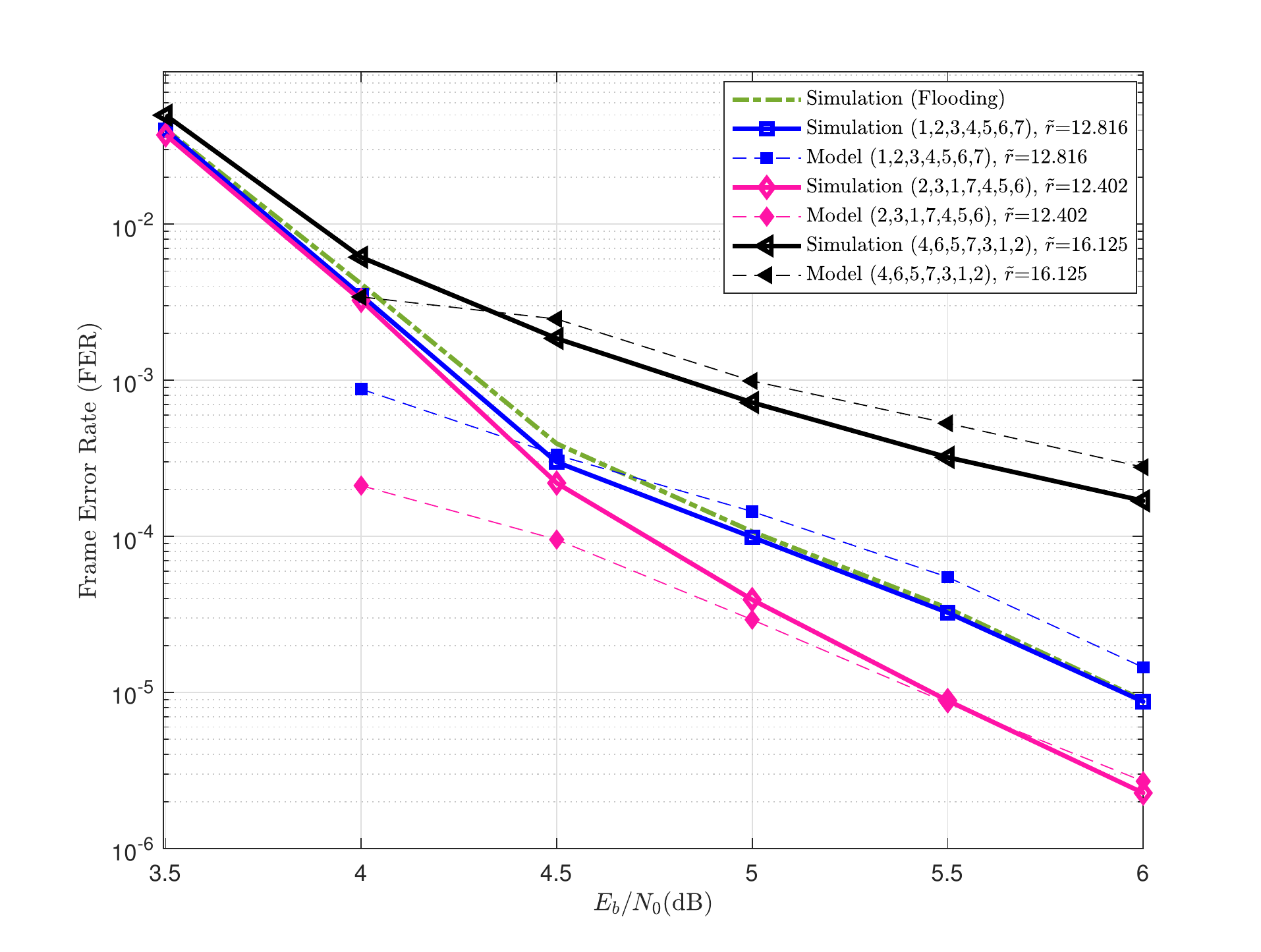}
\caption{The simulation and estimation results of $\mathcal{C}_1$ for different row schedules. (The saturation level is $31.75$, and the maximum number of iterations for layered schedules and the flooding schedule are set to $30$ and $60$, respectively.) }
\label{qc640_lay_diffSched_31_75}
\end{figure}

As the next example, we consider $\mathcal{C}_2$. The multiplicity of different $(a,b)$ LETSs of this code within the range $a \leq 8$ and $b \leq 2$ are listed in Table \ref{tab:wimax576}. These LETSs have been found using the exhaustive search algorithm of~\cite{hashemiireg}.
%\footnote{This code in \cite{hashemiireg} is denoted by $\mathcal{C}_7$. }. 
As can be seen, this code has a variety of LETSs that can potentially contribute to the error floor performance. 
Moreover, unlike the previous example, there are different non-isomorphic structures within each class of TSs. 
For example, there are $8$ different non-isomorphic $(7,1)$ LETS structures in this code. For $2$ out of $8$ structures, the code also contains TSs with two different TSLPs.
This means that there are ten $(7,1)$ LETS groups, each with size $24$, that can have different failure probabilities under a layered decoder. We
denote these groups by $(7,1)_1, \ldots, (7,1)_{10}$, respectively.
%Moreover, within each group, the failure probability of different LETSs can be different due to different system matrices, in general, and different $\tilde{r}$ values, in particular. For example, the LETS structure of $(7,1)_{10}$ has $6$ layers and in fact, if one considers all the possible layer permutations, this structure has $57$ different $\tilde{r}$ values. This is close to the upper bound of $60$ from Corollary~\ref{MaxDominantPerm}.
%In fact, two of the isomorphic groups, each, contain the LETSs with $2$ different TSLPs. 

\begin{table}[]
\centering
\caption{Multiplicities of $(a,b)$ LETSs of $\mathcal{C}_{2}$ within the range $a \leq 8$ and $b \leq 2$}
\begin{threeparttable}
\setlength{\tabcolsep}{1.5pt}
\renewcommand{\arraystretch}{1.0}
\label{tab:wimax576}
\begin{tabular}{||c|c|c|c|c|c|c|c|c|c|| }
\cline{1-10}
\cline{1-10}
$(a,b)$&$(4,2)$&$(5,2)$&$(6,1)$&$(6,2)$&$(7,1)$&$(7,2)$&$(8,0)$&$(8,1)$&$(8,2)$\\
\cline{1-10}
Multiplicity&$144$&$216$&$48$&$1068$&$240$&$3600$&$48$&$720$&$13464$ \\
\cline{1-10}
\end{tabular}
\end{threeparttable}
\end{table}

To evaluate the effect of row scheduling on the error floor of $\mathcal{C}_2$, the same general technique as the one employed in the previous example is used. 
In this regard, for each of the LETS groups (those with the same structure and TSLP) within $(6,1)$, $(7,1)$, $(8,0)$ and $(8,1)$ classes, 
we estimate the failure probability of various schedules. This is performed based on the approximate average DE method of Subsection~\ref{Optim_row_lay_sec} at $E_b/N_0=6$ dB 
and for a saturation level of $15.75$.  The number of LETS groups within each of the aforementioned classes are $2$, $10$, $2$ and $30$, respectively. Each group has the same size of $24$. To estimate the contribution of each LETS class to the error floor, we first estimate the failure rate of a member of each LETS group within the class, then multiply the result by $24$, and finally add up the results for different groups within the class. These results for different classes are presented in Fig. \ref{wimax576_lay_opt_15_75}, for all the $6! = 720$ possible row schedules.
% obtain the effect of scheduling on each class, the error curves related to different groups of a class are first multiplied by their multiplicity and then added together and the result is shown with a single curve per each class. Fig. \ref{wimax576_lay_opt_15_75} illustrates the final result. 
As can be observed, on average, the most harmful LETS class of this code, for saturation level $15.75$, 
is the $(7,1)$ class. %It is noted that these results are not sorted according to the dominant eigenvalues because the system matrix of each of the LETSs can have different eigenvalues which do not, necessarily, follow the same increasing or decreasing trend. For example, changing the schedule might cause an increase in the dominant eigenvalue of a group from $(7,1)$ class while decreases the dominant eigenvalue of a group from $(8,1)$ class. 
In order to observe the overall error probability for different row schedules, the contribution of different classes from Fig. \ref{wimax576_lay_opt_15_75} are added and presented in Fig. \ref{wimax576_lay_opt_15_75_totErr}. 
Interestingly, for $\mathcal{C}_2$, it appears that the original schedule $(1,2,3,4,5,6)$ has one of the lowest error rates.
This schedule is shown in Fig. \ref{wimax576_lay_opt_15_75_totErr} by a full triangle. 
The worst schedule in Fig. \ref{wimax576_lay_opt_15_75_totErr} is $(4,1,5,2,3,6)$, and is identified by a full square.
%point is that, in the second phase where the best schedule is identified, we found that the usual row order of the original code, $\{1,2,3,4,5,6\}$, is one of the best schedules. Accordingly, one of the worst schedules is identified as $\{4,1,5,2,3,6\}$.

\begin{figure}
\centering
\includegraphics[width=3.5in]{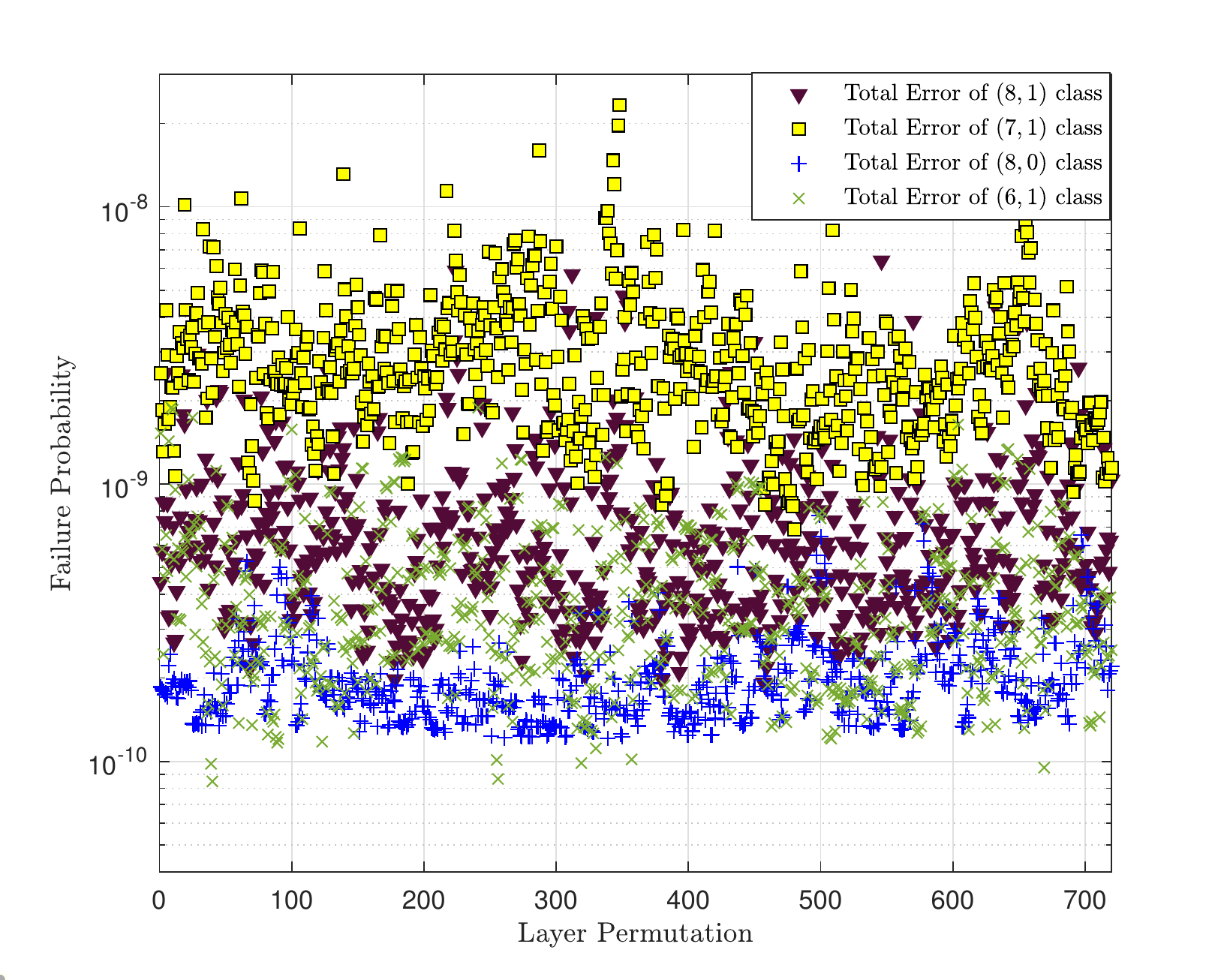}
\caption{The effect of various row layered schedules on the failure probability of different classes of LETSs in $\mathcal{C}_2$ ($E_b/N_0=6$ dB, saturation level $15.75$). }
\label{wimax576_lay_opt_15_75}
\end{figure}
\begin{figure}
\centering
\includegraphics[width=3.5in]{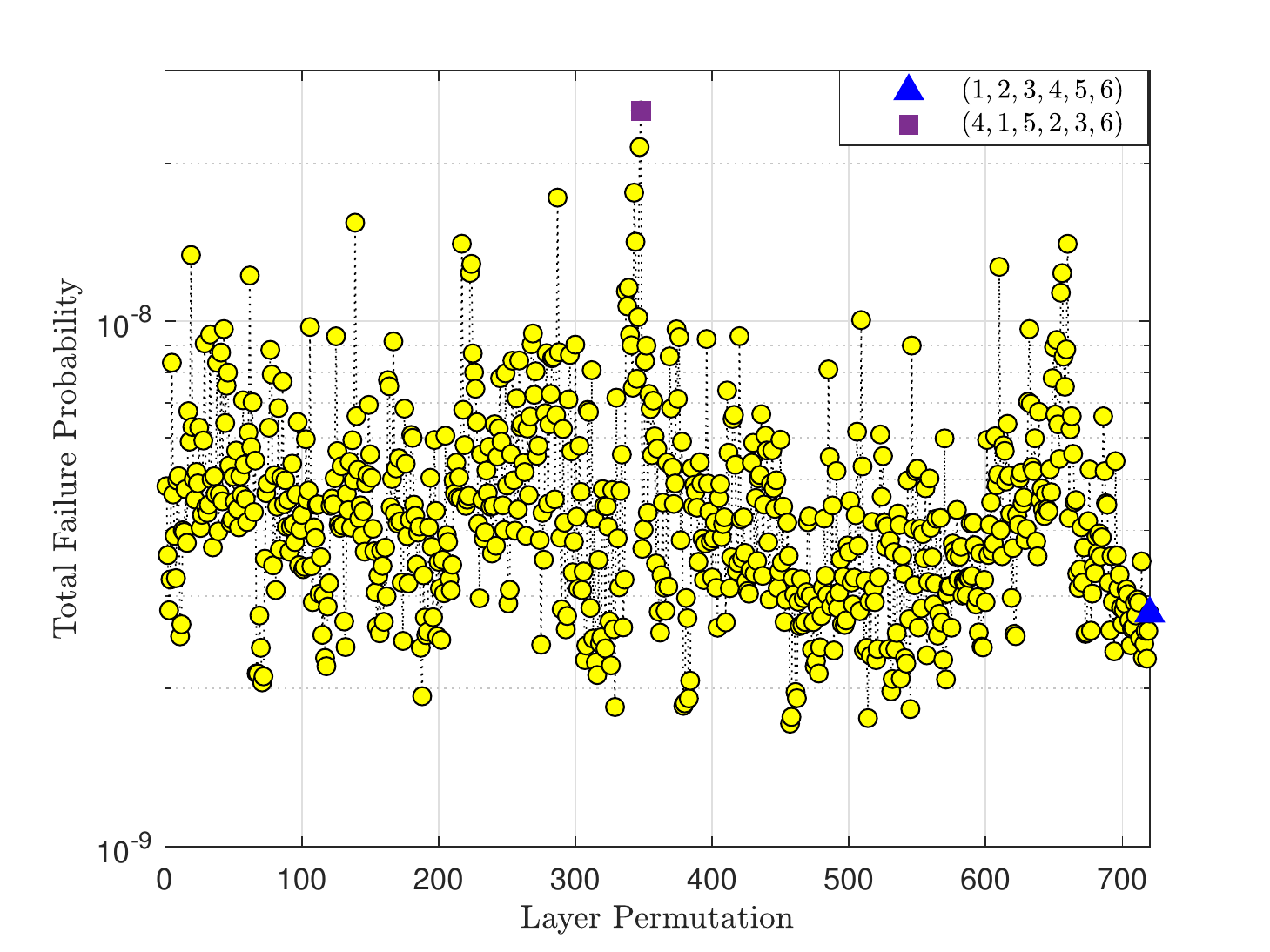}
\caption{The effect of various row layered schedules on the total FER of $\mathcal{C}_2$ ($E_b/N_0=6$ dB, saturation level $15.75$).}
%The summation of all the error estimates of Fig. \ref{wimax576_lay_opt_15_75} related to $\mathcal{C}_2$ at $E_b/N_0=6$ dB. The saturation level is 15.75. }
\label{wimax576_lay_opt_15_75_totErr}
\end{figure}

For the two schedules $(1,2,3,4,5,6)$ and $(4,1,5,2,3,6)$, we have estimated the failure probability 
of each of the $10$ LETS groups within the class $(7,1)$, using the exact DE. 
These results along with the total failure probability of the $(7,1)$ class are provided in Fig. \ref{wimax576_lay_All7_1_err}.
%As $(7,1)$ is the most harmful class of $\mathcal{C}_2$, in Fig. \ref{wimax576_lay_All7_1_err}, the estimation results of the chosen schedules for various $(7,1)$ LETSs as well as their total contribution in the error floor of $\mathcal{C}_2$ are presented. Similar figures can be obtained for other LETSs, as well. 
We note that the LETSs within first and third groups, i.e., $(7,1)_1$ and $(7,1)_3$, 
are isomorphic (they only differ by their TSLPs). So are the LETSs within $(7,1)_6$ and $(7,1)_7$.
%Denoting the different subgraphs by $(7,1)_1$ to $(7,1)_{10}$, the first and third structures are isomorphic\footnote{The second structure, $(7,1)_2$, have the same internal structure as $(7,1)_1$ and $(7,1)_3$. However, due to different location of unsatisfied CN, it is not considered as isomorphic to them. }. Similarly, the number 6 and 7 are, also, isomorphic LETSs. 
The examination of Fig. \ref{wimax576_lay_All7_1_err} for schedule $(4,1,5,2,3,6)$ shows
that the two groups $(7,1)_6$ and $(7,1)_7$, despite having the same structure, have different error probabilities, due to different TSLPs. Fig. \ref{wimax576_lay_All7_1_err} also demonstrates that, for both schedules, the LETS group $(7,1)_{10}$ is 
the most harmful one. The LETS structure of $(7,1)_{10}$ has $6$ layers and in fact, if one considers all the possible layer permutations, 
this structure has $57$ different $\tilde{r}$ values. This is close to the upper bound of $60$ from Corollary~\ref{MaxDominantPerm}.
The $57$ different $\tilde{r}$ values for the $(7,1)_{10}$ structure are between $6.043$ and $8.216$. For the two schedules
$(1,2,3,4,5,6)$ and $(4,1,5,2,3,6)$, these values are $6.408$ and $8.216$, respectively.
%However, the gap between the failure of $(7,1)_{10}$ and other structures in the worst schedule is much higher than the optimized one. 
Another observation from Fig. \ref{wimax576_lay_All7_1_err} is that the relative harmfulness of the TSs can change depending on the schedule. 
%This can be seen in Fig. \ref{wimax576_lay_opt_15_75} where the errors of different classes does not have a specific boundary and are overlapped. Moreover, the isomorphic TSs might have different failure rates. 
For example, while the $(7,1)_{8}$ is the second most harmful group of LETSs for schedule $(4,1,5,2,3,6)$, 
for $(1,2,3,4,5,6)$, the second most harmful group is $(7,1)_9$. 
%Further, the failure rate of the isomorphic $(7,1)_6$ and $(7,1)_7$ LETSs are different in the right figure while they have similar error probability in the left one.\\

\begin{figure}
\centering
\includegraphics[width=3.9in]{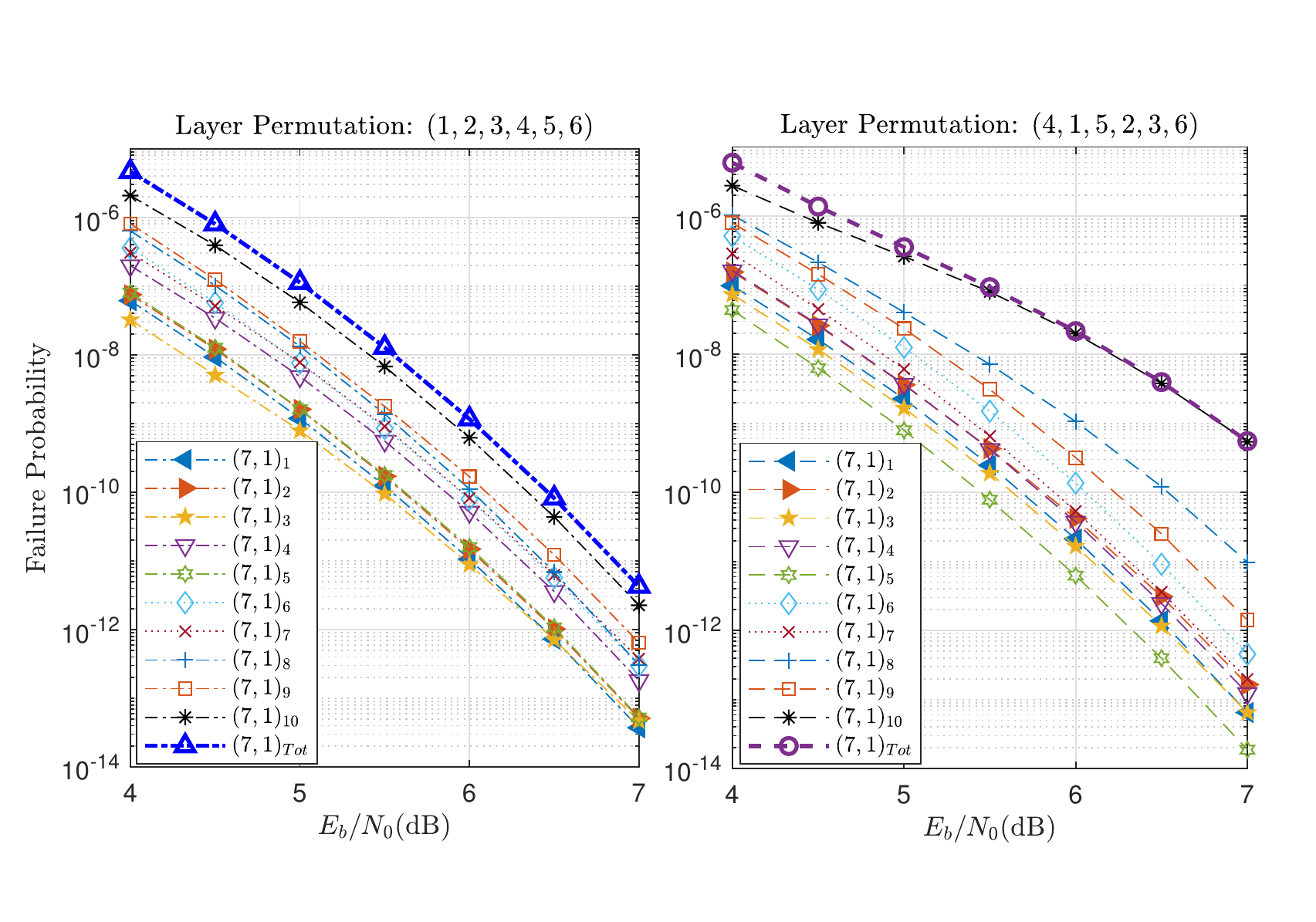}
\caption{The error estimation of the ten $(7,1)$ LETS groups of $\mathcal{C}_2$ for the two schedules 
$(1,2,3,4,5,6)$ and $(4,1,5,2,3,6)$ (saturation level $15.75$).}
%as well as their total contribution in the error floor of $\mathcal{C}_2$ for two separate row layered schedules. The saturation level is 15.75. }
\label{wimax576_lay_All7_1_err}
\end{figure}

In order to examine the accuracy of our estimations, the simulation results for the two schedules together with 
the estimation results based on the contributions of $(6,1)$, $(7,1)$, $(8,0)$ and $(8,1)$ LETS classes
are shown in Fig. \ref{wimax576_lay_dif_sched_1575}. As can be seen, for both schedules, there is a good match between simulations and estimations. 
%the linear model approximates the error floor of the both schedules reasonably well. It should be noticed that the error floor of the worst schedule is mainly due to $(7,1)$ while the total approximation of the regular schedule is affected by the sum of the errors of various classes. 
%%%%%%%%%%%
\if0
\begin{table}[]
\centering
\begin{threeparttable}
\caption{Multiplicities of $(a,b)$ LETSs of $\mathcal{C}_{2}$ within the range $a \leq 8$ and $b \leq 2$}
\setlength{\tabcolsep}{2pt}
\renewcommand{\arraystretch}{1.2}
\label{tab:wimax576}
\begin{tabular}{||c|c|| }
\cline{1-2}
\multicolumn{2}{||c||}{$\mathcal{C}_{2}$}\\%&\multicolumn{6}{c|||}{$\mathcal{C}_{5}$} \\
\cline{1-2}
$(a,b)$& Multiplicity \\%&\multicolumn{3}{c||}{Primes}&Total &Total& \cite{kyung2012finding}\\
%\cline{2-4}
%class& $s_4$&$s_5$&$s_6$&LETS&FEAS&FAS\\
\cline{1-2}
(4,2)&144\\
\cline{1-2}
(5,2)&216\\
\cline{1-2}
(6,1)&48\\
\cline{1-2}
(6,2)&1068\\
\cline{1-2}
(7,1)&240\\
\cline{1-2}
(7,2)&3600\\
\cline{1-2}
(8,0)&48\\
\cline{1-2}
(8,1)&720\\
\cline{1-2}
(8,2)&13464\\

%\cline{1-2}
%$dpl$&\multicolumn{6}{c|||}{442 sec. / 28 sec. \tnote{*}}\\
%\cline{1-2}
%$dot$&\multicolumn{6}{c|||}{- / 392 sec.}\\
\cline{1-2}
\end{tabular}
%\begin{tablenotes}
%    \item[\textdagger] ``nr" stands for \textit{not reported}.
%\item[*] Runtime to find all the instances of LETS structures reported in this table/runtime to find instances in the first %10 classes reported in this table.
%  \end{tablenotes}
\end{threeparttable}
\end{table}
%%%%%%%%%%%  
\fi   
%%%%%%%%%%% 

%%%%%%%%%%%    

\begin{figure}
\centering
\includegraphics[width=3.5in]{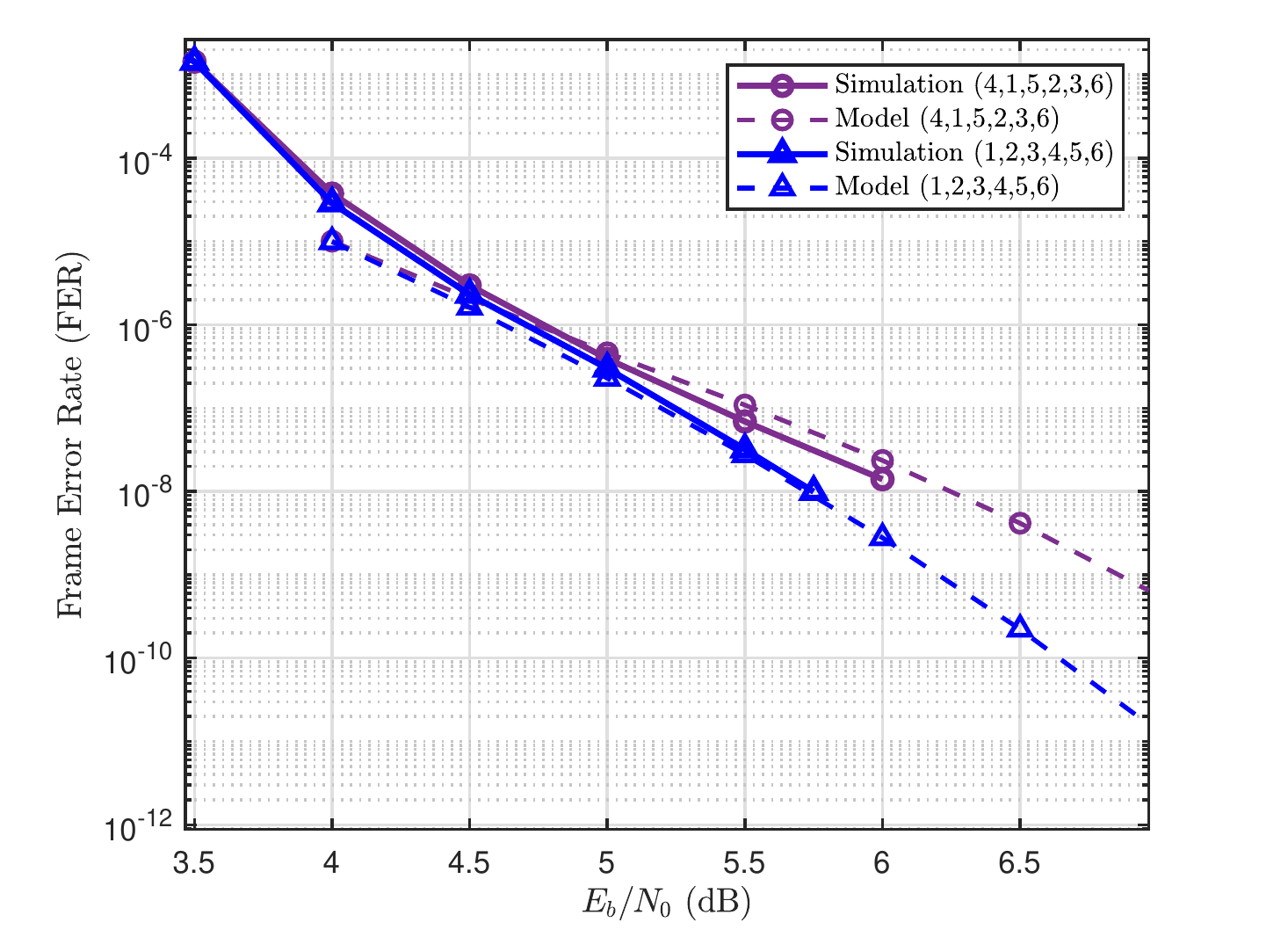}
\caption{Simulation and estimation results of $\mathcal{C}_2$ for the two schedules 
$(1,2,3,4,5,6)$ and $(4,1,5,2,3,6)$ (saturation level $15.75$, $I_{max}=30$).}
%showing the effect of different row layered schedules. The saturation level is 15.75 and $I_{max}=30$.}
\label{wimax576_lay_dif_sched_1575}
\end{figure}

Finally, in our experiments, we observe that the same row layered schedule that minimizes the error floor of SPA also performs well for min-sum algorithm (MSA).
This is explained in Fig.~\ref{qc640_MSA} for $\mathcal{C}_{1}$, where we have used row layered MSA (saturation level 
$31.75$, maximum number of iterations $30$) with the same three schedules as in Fig.~\ref{qc640_lay_diffSched_31_75}. As can be seen, the three schedules have the same relative performance as they had with SPA. In particular, 
the schedule $(2,3,1,7,4,5,6)$ that was optimal for SPA still performs the best with MSA.

\begin{figure}
\centering
\includegraphics[width=3.6in]{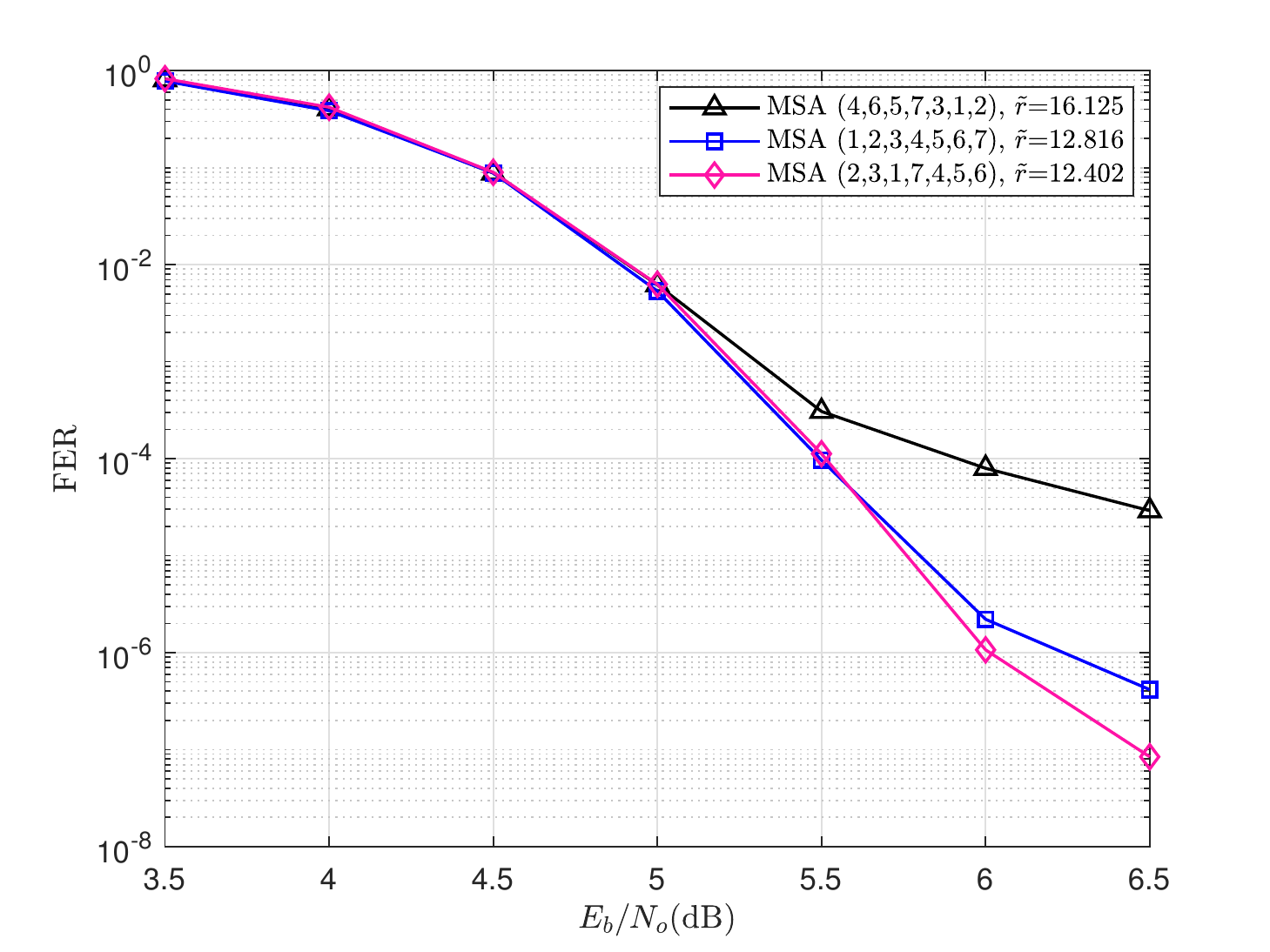}
\caption{Performance of  $\mathcal{C}_1$ under row layered MSA with row schedules similar to those of Fig.~\ref{qc640_lay_diffSched_31_75} (the saturation level and the maximum number of iterations are $31.75$ and  $30$, respectively.) }
\label{qc640_MSA}
\end{figure}

\section{Conclusion}
\label{con}
In this paper, we studied the error floor of QC-LDPC codes under row layered saturating SPA. For this, we developed a 
linear state-space model for LETSs of the code which incorporates the layered nature of scheduling. 
We then studied the system matrices of the model and made connections between these matrices and those corresponding to the linear state-space model of the flooding decoder. In particular, we demonstrated that the spectral radius of the transition matrix of the layered decoder is always larger than that of its flooding counterpart. 

We showed that the proposed model can estimate the failure probability of LETSs, as well as the error floor of the code rather accurately. In particular, we demonstrated that the failure rate of a LETS under layered decoding is not only a function of its topology, but also depends on the location of its constituent CNs in different layers. We called this information, TS layer profile, or TSLP, in brief. As a result, we established that the error floor of the same code under the same saturating SPA can significantly change by modifying the order in which the row layers are updated. We also studied the problem of finding the schedule with the lowest error floor and devised an efficient algorithm to find it. In particular, we demonstrated that the layered decoder, with a well designed schedule, can outperform its flooding counterpart. This adds yet another advantage to the application of layered decoding in practice. The well-known advantages, prior to this result, were the faster convergence and lower hardware complexity.

We note that the linear state-space model presented in this paper can also be applied to column layered decoders~\cite{Ali-Thesis,FB-TCOM}. 
The application however involves non-trivial modifications to derive the model parameters. 

For the codes studied in this work, the vast majority of problematic TSs were LETSs. We however note that the linear state-space model can 
also be applied to ETSs with leaf, if such TSs happen to have a non-negligible contribution to the error floor. 
 
An interesting line of inquiry would be to use the results of this work in the design of QC-LDPC codes with low error floor
under layered decoding.
%Future work may include the analysis of the spectral property of the LETS structures whose dominant eigenvalue of their row layered decoder system matrices under various permutations are invariant or at least much lower than the upper bound of Theorem \ref{MaxDominantPerm}.

\section{Appendix: Proof of Proposition~\ref{flip_perm_proposition}}
To prove Proposition~\ref{flip_perm_proposition}, we first need the following lemmas.
\begin{lem}\label{AteqPAP_Lemma}
Consider the linear state-space model of a LETS in which the two state variables corresponding to each missatisfied CN are labeled by consecutive numbers,
and let $\mathbf{A}$ be the corresponding flooding transition matrix. We then have
% matrix of an LETS in flooding schedule decoder whose state variables, belonging to the same mis-satisfied CN, are labelled by consecutive numbers.
\begin{itemize}
\item[(a)] If $\mathbf{A}(i,j)=1$, then $\mathbf{A}(j+2\times{\bmod(j,2)}-1,i+2\times{\bmod(i,2)}-1)=1$, where $\bmod(i,2)$ is used to denote the value of $i$ modulo $2$.
%The following mapping exists between the nonzero entries of matrix $\mathbf{A}$,
%\begin{IEEEeqnarray*}{lCl"s}
%\text{If\ }\mathbf{A}(i,j)=1\Rightarrow\\ \mathbf{A}(j+2\ {\bmod(j,2)}-1,i+2~ {\bmod(i,2)}-1)=1.\IEEEyesnumber\label{mappingA_ij}
%\end{IEEEeqnarray*}
\item[(b)] If the nonzero entries of an $m_s\times m_s$ symmetric and unitary permutation matrix $\mathbf{P}$ are defined by
$\mathbf{P}(i,i+2 \times {\bmod(i,2)}-1)=1$, for $i=1,\dots,m_s$, then
%\begin{IEEEeqnarray}{lCl"s}
%\mathbf{P}(i,i+2 \times {\bmod(i,2)}-1)=1 \text{\ \ \ for \ \ }i=1,\dots,m_s,\label{unitaryPerm}
%\end{IEEEeqnarray}
\begin{equation}\label{AteqPAP}
\mathbf{A}^T=\mathbf{P}\mathbf{A}\mathbf{P}.
\end{equation}
%Matrix $\mathbf{P}$ is symmetric, $\mathbf{P}=\mathbf{P}^T$, and unitary, $\mathbf{P}\mathbf{P}=\mathbf{I}$. The operator $\bmod(i,2)$ is used to denote the remainder of the Euclidean division of $i$ by $2$.
\end{itemize}
\end{lem}
\begin{proof}
(a) %Labelling the state variables related to a mis-satisfied CN by consecutive numbers provides extra information about the relation of state variables. 
To each missatisfied CN, there correspond two state variables with even and odd labels, respectively. 
Denoting the state variables of a given missatisfied CN by $x_{2k}$ and $x_{2k-1}$, 
%(noting that the odd index of a missatisifed CN is always less than the even index), 
one can see that each of them is a function of at least one other state variable corresponding to another missatisfied CN. 
Suppose that $x_{2k}$ is a function of state variable $x_{2k'}$ (or $x_{2k'-1}$) from another missatisfied CN. 
Then, $x_{2k'-1}$ (or $x_{2k'}$) must be a function of $x_{2k-1}$. This corresponds to the relationship between the entries of $\mathbf{A}$ as described in Part (a) of the lemma. %  is equivalent to the mapping of \eqref{mappingA_ij}. 
(b) The symmetric application of the permutation matrix $\mathbf{P}$ to $\mathbf{A}$, i.e., $\mathbf{P}\mathbf{A}\mathbf{P}^T$, 
permutes the even and odd rows and columns that correspond to each missatisfied CN. Based on the result of Part (a),
%Following the mapping of \eqref{mappingA_ij}, 
this permutation results in the transpose matrix $\mathbf{A}^T$. Also, Equation \eqref{AteqPAP} is derived based on the fact that $\mathbf{P}^T = \mathbf{P}$.  
\end{proof}
%Similarly, it can be shown if the matrix $\mathbf{A}$ be in the systematic form where in addition to the consecutive indices for each missatisfied CN, the state-variables are also sorted based on the row layers (such as Fig. \ref{(5,3)lay_col}), the following equations are, also, valid,
%\begin{equation}\label{AlP}
%\mathbf{A}_l\mathbf{P}=\mathbf{P}\mathbf{A}_u^T
%\end{equation}
%\begin{equation}\label{AuP}
%\mathbf{A}_u\mathbf{P}=\mathbf{P}\mathbf{A}_l^T.
%\end{equation}
\begin{lem}\label{muAlplusAuLemma}
%Let the $m_s\times m_s$ matrix $\tilde{\mathbf{A}}_{J \rightarrow 1}$ be the system matrix of an LETS, $\mathcal{S}$, corresponding to the row layered decoder. Also, 
Let $\mu_k$ and $\mathbf{u}_k$ be an eigenvalue and its corresponding right eigenvector of a layered transition matrix $\tilde{\mathbf{A}}_{J \rightarrow 1}$, which is in systematic form. Then,
\begin{equation}\label{muAlplusAu}
(\mu_k\mathbf{A}_l+\mathbf{A}_u)\mathbf{u}_k=\mu_k\mathbf{u}_k,
\end{equation}
where $\mathbf{A}_l$ and $\mathbf{A}_u$ are the lower and upper triangular parts of the corresponding flooding transition matrix $\mathbf{A}=\mathbf{A}_l+\mathbf{A}_u$.
%, the system matrix of the flooding schedule decoder corresponding to $\mathcal{S}$.
\end{lem}
\begin{proof}
For simplicity, we prove the result for a LETS with $J=3$ layers. The proof for larger values of $J$ is similar.
%For the matrix $\tilde{\mathbf{A}}_{J \rightarrow 1}$, we have
%\begin{equation*}
%\tilde{\mathbf{A}}_{J \rightarrow 1}\mathbf{u}_k=\mu_k\mathbf{u}_k
%\end{equation*}
%Suppose $J=3$, then $\tilde{\mathbf{A}}_{3 \rightarrow 1}$ is defined as \eqref{Atild_3arrow1}. By partitioning the right eigenvector to 3 subvector layers, we have
By the definition of an eigenvalue and the corresponding eigenvector, we have
\begin{IEEEeqnarray*}{lCl"s}
\tilde{\mathbf{A}}_{3 \rightarrow 1}
\left[
\begin{array}{c}
\mathbf{u}_{k,1}\\
\hline
\mathbf{u}_{k,2}\\
\hline
\mathbf{u}_{k,3}
\end{array}
\right]=\mu_k\left[
\begin{array}{c}
\mathbf{u}_{k,1}\\
\hline
\mathbf{u}_{k,2}\\
\hline
\mathbf{u}_{k,3}
\end{array}
\right],
\end{IEEEeqnarray*}
in which the eigenvector is partitioned according to the three layers. By replacing $\tilde{\mathbf{A}}_{3 \rightarrow 1}$ in the above equation with 
the sub-matrices from \eqref{Atild_3arrow1}, we can write
\begin{IEEEeqnarray*}{lCl"s}
\mathbf{A}_{1,2}\mathbf{u}_{k,2}+\mathbf{A}_{1,3}\mathbf{u}_{k,3}=\mu_k\mathbf{u}_{k,1},\\
\mathbf{A}_{2,1}(\underbrace{\mathbf{A}_{1,2}\mathbf{u}_{k,2}+\mathbf{A}_{1,3}\mathbf{u}_{k,3}}_{\mu_k\mathbf{u}_{k,1}})+\mathbf{A}_{2,3}\mathbf{u}_{k,3}=\mu_k\mathbf{u}_{k,2},\\
\mathbf{A}_{3,1}(\underbrace{\mathbf{A}_{1,2}\mathbf{u}_{k,2}+\mathbf{A}_{1,3}\mathbf{u}_{k,3}}_{\mu_k\mathbf{u}_{k,1}})+\\\mathbf{A}_{3,2}\big(\underbrace{\mathbf{A}_{2,1}(\mathbf{A}_{1,2}\mathbf{u}_{k,2}+\mathbf{A}_{1,3}\mathbf{u}_{k,3})+
\mathbf{A}_{2,3}\mathbf{u}_{k,3}}_{\mu_k\mathbf{u}_{k,2}}\big)=\mu_k\mathbf{u}_{k,3},
\end{IEEEeqnarray*}
or equivalently,
\begin{IEEEeqnarray*}{lCl"s}
\underbrace{\left[
\begin{array}{c|c|c}
\mathbf{0} &\mathbf{A}_{1,2}&\mathbf{A}_{1,3}\\
\hline
\mu_k\mathbf{A}_{2,1}& \mathbf{0}& \mathbf{A}_{2,3}\\
\hline
\mu_k\mathbf{A}_{3,1}& \mu_k \mathbf{A}_{3,2}&\mathbf{0}
\end{array}
\right]}_{(\mu_k\mathbf{A}_l+\mathbf{A}_u)}\left[
\begin{array}{c}
\mathbf{u}_{k,1}\\
\hline
\mathbf{u}_{k,2}\\
\hline
\mathbf{u}_{k,3}
\end{array}
\right]=\mu_k\left[
\begin{array}{c}
\mathbf{u}_{k,1}\\
\hline
\mathbf{u}_{k,2}\\
\hline
\mathbf{u}_{k,3}
\end{array}
\right],
\end{IEEEeqnarray*}
which is the same as \eqref{muAlplusAu} for $J=3$. %The same approach can be taken for $J>3$ and \eqref{muAlplusAu} is proved by induction.
\end{proof}
The proof of the following lemma is similar to that of Lemma~\ref{muAlplusAuLemma}.
\begin{lem}\label{AlplusmuAuLem}
Let $\tilde{\mathbf{A}}_{1 \rightarrow J}=\mathcal{A}_1\mathcal{A}_2 \cdots \mathcal{A}_J$ be the systematic layered transition matrix of a LETS for the row layered decoder in which the order of layers is reversed. Also, let $\mu'_{k}$ and $\mathbf{u'}_{k}$ be an eigenvalue and its corresponding right eigenvector of $\tilde{\mathbf{A}}_{1 \rightarrow J}$. Then
\begin{equation}\label{AlplusmuAu}
(\mathbf{A}_l+\mu'_{k}\mathbf{A}_u)\mathbf{u'}_{k}=\mu'_{k}\mathbf{u'}_{k},
\end{equation}
where $\mathbf{A}_l$ and $\mathbf{A}_u$ are the lower and upper triangular part of the corresponding flooding transition matrix $\mathbf{A}=\mathbf{A}_l+\mathbf{A}_u$.
%, the system matrix of the flooding schedule decoder corresponding to $\mathcal{S}$.
\end{lem}
%\begin{proof}
%The same approach in the proof of Lemma \ref{muAlplusAuLemma} can be used here. The difference is that $\tilde{\mathbf{A}}_{1 \rightarrow J}$ should be calculated instead of $\tilde{\mathbf{A}}_{J \rightarrow 1}$. When $J=3$, this matrix is defined as $\tilde{\mathbf{A}}_{1 \rightarrow 3}=\mathcal{A}_1\mathcal{A}_2\mathcal{A}_3$.
%\end{proof}
%
%\begin{rem}
%Given a permutation, $\varpi\in\Pi_J$, of the row layers in an LETS structure, the state-variables can be, always, relabelled (or equivalently the $\mathbf{A}$ matrix is symmetrically permuted) such that the system matrix of the flooding schedule decoder posses the systematic form. Therefore, equations \eqref{muAlplusAu} and \eqref{AlplusmuAu} would be related to $\varpi$ and its reversed version, respectively. 
%\end{rem}
%
%\begin{proof}
To prove Proposition~\ref{flip_perm_proposition}, without loss of generality, we assume that $\tilde{\mathbf{A}}_{\pi_J \rightarrow \pi_1}$ is in the systematic form. 
Then, according to Lemma \ref{muAlplusAuLemma}, for $\tilde{\mathbf{A}}_{\pi_J \rightarrow \pi_1}$, we have
\begin{equation*}
(\tilde{r}_i\mathbf{A}_l+\mathbf{A}_u)\tilde{\mathbf{u}}_i=\tilde{r}_i\tilde{\mathbf{u}}_i,
\end{equation*}
where $\tilde{r}_i$ and $\tilde{\mathbf{u}}_i$ are an eigenvalue and its corresponding eigenvector of $\tilde{\mathbf{A}}_{\pi_J \rightarrow \pi_1}$, respectively.
Suppose that $\mathbf{P}$ is the permutation matrix defined in Lemma~\ref{AteqPAP_Lemma}. Since $\mathbf{P}\mathbf{P}=\mathbf{I}$, we can write
\begin{equation}
(\tilde{r}_i\mathbf{A}_l\mathbf{P}+\mathbf{A}_u\mathbf{P})\mathbf{P}\tilde{\mathbf{u}}_i=\tilde{r}_i\tilde{\mathbf{u}}_i.
\label{eq234}
\end{equation}
Moreover, for a systematic flooding transition matrix $\mathbf{A}$, one can show that $\mathbf{A}_l\mathbf{P}=\mathbf{P}\mathbf{A}_u^T$
and $\mathbf{A}_u\mathbf{P}=\mathbf{P}\mathbf{A}_l^T$. Using these in (\ref{eq234}), we obtain
%Based on equations \eqref{AlP} and \eqref{AuP}, we have
\begin{equation}
(\mathbf{P}\mathbf{A}_l^T+\tilde{r}_i\mathbf{P}\mathbf{A}_u^T)\mathbf{P}\tilde{\mathbf{u}}_i=\tilde{r}_i\tilde{\mathbf{u}}_i.
\end{equation} 
Multiplying both sides by $\mathbf{P}$ results in
\begin{equation}
(\mathbf{A}_l+\tilde{r}_i\mathbf{A}_u)^T\mathbf{P}\tilde{\mathbf{u}}_i=\tilde{r}_i\mathbf{P}\tilde{\mathbf{u}}_i.
\label{eq333}
\end{equation}
From~(\ref{eq333}), the eigenvalues $\tilde{r}_i$ are the roots of the following equation:
%solutions of the above equation are a root of characteristic function and as the determinant of a matrix and its transpose are equal, we can write
\begin{equation}
\det[(\mathbf{A}_l+\tilde{r}_i\mathbf{A}_u)^T-\tilde{r}_i\mathbf{I}]=\det[(\mathbf{A}_l+\tilde{r}_i\mathbf{A}_u)-\tilde{r}_i\mathbf{I}] = 0.
\end{equation}
On the other hand, based on Lemma \ref{AlplusmuAuLem}, the roots of the determinant on the right hand side are the eigenvalues of $\tilde{\mathbf{A}}_{\pi_1 \rightarrow \pi_J}$. This completes the proof.
%, following the , is related to the system matrix of the flipped permutation, $\tilde{\mathbf{A}}_{\pi_1 \rightarrow \pi_J}$. 
%\end{proof}


\begin{thebibliography}{10}



\bibitem{danesh}
E. Cavus and B. Daneshrad, ``A performance improvement and error floor avoidance technique for belief propagation
decoding of LDPC codes,'' in {\em Proc. 16th IEEE Int. Symp. Personal, Indoor Mobile Radio Commun.}, Los Angeles, CA,
USA, Sep. 2005, pp. 2386--2390.

\bibitem{Ryan2}
Y. Han and W. E. Ryan, ``LDPC decoder strategies for achieving low error floors,'' in {\em Proc. Inform. Theory Appl. Workshop},
San Diego, CA, USA, Jan. 2008, pp. 277--286.

\bibitem{Ryan1}
Y. Zhang and W. E. Ryan, ``Toward low LDPC-code floors: a case study,'' {\em IEEE Trans. Commun.}, vol. 57, no. 6, pp. 1566--1573, Jun. 2009.

\bibitem{Kyung}
G. B. Kyung and C.-C. Wang, ``Finding the exhaustive list of small fully absorbing sets and designing the corresponding low error-floor decoder,'' {\em IEEE Trans. Commun.}, vol. 60, no. 6, pp. 1487--1498, Jun. 2012.

\bibitem{Zhang1}
S. Zhang and C. Schlegel, ``Controlling the error floor in LDPC decoding,'' {\em IEEE Trans. Commun.}, vol. 61, no. 9, pp. 3566--3575, Sep. 2013.

\bibitem{TB}
S. Tolouei and A. H. Banihashemi, ``Lowering the error floor of LDPC codes using multi-step quantization,'' {\em IEEE Commun. Lett.}, vol. 18, no. 1, pp. 86--89, Jan. 2014.

\bibitem{zhang_quasiuniform}
X. Zhang and P. H. Siegel, ``Quantized iterative message passing decoders with low error floor for LDPC codes,'' {\em IEEE Trans. Commun.}, vol. 62, no. 1, pp. 1--14, January 2014.

\bibitem{TSbreaking}
S. Kang, J. Moon, J. Ha and J. Shin, ``Breaking the trapping sets in LDPC codes: Check node removal and collaborative decoding," {\em IEEE Trans. Commun.}, vol. 64, no. 1, pp. 15--26, Jan. 2016.

\bibitem{Lee-2018}
H.-C. Lee, P.-C. Chou and Y.-L. Ueng, ``An effective low-complexity error-floor lowering technique for high-rate {QC-LDPC} codes,” {\em IEEE Commun. Lett.},
vol. 22, no. 10, pp. 1988--1991, Oct. 2018.

\bibitem{homayoon2020}
H. Hatami, D. G. M. Mitchell, D. J. Costello and T. E. Fuja, ``A threshold-based min-sum algorithm to lower the error floors of quantized LDPC decoders,'' {\em IEEE Trans. Commun.}, vol. 68, no. 4, pp. 2005--2015, Apr. 2020.

\bibitem{Zhao}
J. Zhao, F. Zarkeshvari and A. H. Banihashemi, ``On implementation of min-sum algorithm and its modifications for decoding low-density parity-check
({LDPC}) codes,” {\em IEEE Trans. Commun.}, vol. 53, no. 4, pp. 549–554, Apr. 2005.

\bibitem{mao2001heuristic}
Y.~Mao and A.~H. Banihashemi, ``A heuristic search for good low-density parity-check codes at short block lengths,'' in {\em Proc. IEEE Int. Conf. Comm., Helsinki, Finland}, pp.~41--44, Jun. 2001.

\bibitem{Tian-2004}
   T. Tian, C. Jones, J. D. Villasenor, and R. D. Wesel, ``Selective avoidance of cycles in irregular {LDPC} code construction,”
   {\em IEEE Trans. Commun.}, vol. 52, pp.~1242--1247, Aug. 2004.

\bibitem{xiao2004improved}
	H.~Xiao and A.~H. Banihashemi, ``Improved progressive-edge-growth ({PEG})
	construction of irregular {LDPC} codes,'' {\em IEEE Commun. Lett.}, vol.~8,
	no.~12, pp.~715--717, Dec. 2004.

\bibitem{Peg}
X.-Y. Hu, E.~Eleftheriou, and D.-M. Arnold, ``Regular and irregular progressive edge-growth tanner graphs,'' {\em IEEE Trans. Inf. Theory}, vol.~51, no.~1, pp.~386--398, Jan. 2005.
\bibitem{Ivkovic}
M. Ivkovic, S. K. Chilappagari, and B. Vasic, ``Eliminating trapping sets in low-density parity-check codes by using Tanner
graph covers,'' {\em IEEE Trans. Inf. Theory}, vol. 54, no. 8, pp. 3763--3768, Aug. 2008.

\bibitem{zheng2010constructing}
	X.~Zheng, F.~C.-M. Lau, and C.~K. Tse, ``Constructing short-length irregular
	{LDPC} codes with low error floor,'' {\em IEEE Trans. Commun.}, vol.~58,
	no.~10, pp.~2823--2834, Oct. 2010.

\bibitem{Asvadi}
R. Asvadi, A. H. Banihashemi, and M. Ahmadian-Attari, ``Lowering the error floor of LDPC codes using cyclic liftings,'' {\em
IEEE Trans. Inf. Theory}, vol. 57, no. 4, pp. 2213--2224, Apr. 2011.

\bibitem{Khaz}
S. Khazraie, R. Asvadi and A. H. Banihashemi, ``A PEG construction of finite-length LDPC codes with low error
floor,'' {\em IEEE Commun. Lett.}, vol. 16, pp. 1288--1291, Aug. 2012.

\bibitem{Nguyen}
D. V. Nguyen, S. K. Chilappagari, M. W. Marcellin, and B. Vasic, ``On the construction of structured LDPC codes free of
small trapping sets,'' {\em IEEE Trans. Inf. Theory}, vol. 58, no. 4, pp. 2280--2302, Apr. 2012.

\bibitem{Tao-2018}
X. Tao, Y. Li, Y. Liu, and Z. Hu, “On the construction of {LDPC} codes free of small trapping sets by controlling cycles,” {\em IEEE Commun. Lett.},
vol. 22, no. 1, pp. 9--12, Jan. 2018.	


\bibitem{Sima-CL1}
S. Naseri and A. H. Banihashemi, ``Construction of girth-8 {QC}-{LDPC} codes free of small trapping sets," {\em IEEE Commun. Lett.}, 
vol. 23, no. 11, pp. 1904--1908, Nov. 2019.

\bibitem{Sima-CL2}
S. Naseri and A. H. Banihashemi, ``Spatially coupled {LDPC} codes with small constraint length and low error floor," {\em IEEE Commun. Lett.}, vol. 24, no. 2, pp. 254--258, Feb. 2020.

\bibitem{Bashir-TCOM}
B. Karimi and A. H. Banihashemi, ``Construction of {QC} {LDPC} codes with low error floor by efficient systematic search and elimination of trapping sets," {\em IEEE Trans. Commun.}, vol. 68, no. 2, pp. 697--712, Feb. 2020.

\bibitem{Bashir-irreg}
B. Karimi and A. H. Banihashemi, ``Construction of irregular protograph-based {QC-LDPC} codes with low error floor," 
{\em IEEE Trans. Commun.}, vol. 69, no. 1, pp. 3--18, Jan. 2021.

\bibitem{Sima-TCOM}
S. Naseri and A. H. Banihashemi, ``Construction of time invariant spatially coupled {LDPC} codes free of small trapping sets,''
{\em IEEE Trans. Commun.}, available on IEEExplore Early Access.

\bibitem{Wang}
C. C. Wang, S. R. Kulkarni, and H. V. Poor, ``Finding all small error-prone substructures in {LDPC} codes,” {\em IEEE Trans. Inf. Theory}, vol. 55,
no. 5, pp. 1976--1999, May 2009.

\bibitem{mehdi2012}
M.~Karimi and A.~H. Banihashemi, ``Efficient algorithm for finding dominant  trapping sets of {LDPC} codes,'' {\em IEEE Trans. Inf. Theory}, vol.~58, no.~11, pp.~6942--6958, Nov. 2012.

\bibitem{mehdi2014}
M.~Karimi and A.~H. Banihashemi, ``On characterization of elementary trapping sets of variable-regular {LDPC} codes,'' {\em IEEE Trans. Inf. Theory}, vol.~60, no.~9, pp.~5188--5203, Sep 2014.


\bibitem{yoones2015}
Y.~Hashemi and A.~Banihashemi, ``On characterization and efficient exhaustive  search of elementary trapping sets of variable-regular {LDPC} codes,'' {\em IEEE Commun. Lett.}, vol.~19, pp.~323--326, Mar. 2015.

\bibitem{hashemireg}
Y.~Hashemi and A.~H. Banihashemi, ``New characterization and efficient exhaustive search algorithm for leafless elementary trapping sets of variable-regular LDPC codes,'' {\em IEEE Trans. Inf. Theory}, vol. 62, no. 12, pp. 6713--6736, Dec. 2016. 

\bibitem{hashemiireg}
Y.~Hashemi and A.~H. Banihashemi, ``Characterization of elementary trapping sets in irregular LDPC codes and the corresponding efficient exhaustive search algorithms,'' {\em IEEE Trans. Inf. Theory}, vol. 64, no. 5, pp. 3411--3430, May 2018.

\bibitem{richardson}
T.~Richardson, ``Error floors of {LDPC} codes,'' in {\em Proc. 41th annual   Allerton conf. on commun. control and computing}, Monticello, IL, USA, Oct. 2003, pp.~1426--1435.

\bibitem{Sun_phd}
J. Sun, ``Studies on graph--based coding systems,'' Ph.D. dissertation,
Dept. Elect. Eng., Ohio State Univ., Columbus, OH, USA,
2004.

\bibitem{Cole}
C. A. Cole, S. G. Wilson, E. K. Hall, and T. R. Giallorenzi, ``A general
method for finding low error rates of LDPC codes,'' {\em submitted to IEEE
Trans. Inf. Theory}, May 2006.

\bibitem{LaraIS}
L. Dolecek, Z. Zhang, M. Wainwright, V. Anatharam, and B. Nikolic. ``Evaluation of the low frame error rate performance of LDPC codes using importance sampling,'' in {\em Proc. IEEE Inf. Theory Workshop}, Lake Tahoe, CA, Sep. 2–6, 2007, pp. 202--207.

\bibitem{XB-2007}
H. Xiao and A. H. Banihashemi, ``Estimation of bit and frame error rates of finite-length low-density parity-check codes on binary symmetric channels," {\em IEEE Trans. Commun.}, vol. 55, no. 12, pp. 2234--2239, Dec. 2007.

\bibitem{daneshrad}
E. Cavus, C. L. Haymes and B. Daneshrad, ``Low BER performance estimation of LDPC codes via application of importance sampling to trapping sets,'' {\em IEEE Trans. Commun.}, vol. 57, no. 7, pp. 1886--1888, Jul. 2009.

\bibitem{Lara_SP}
L. Dolecek, P. Lee, Z. Zhang, V. Anatharam, B. Nikolic, and M. J. Wainwright,
``Predicting error floors of structured LDPC codes: deterministic
bounds and estimates,'' {\em IEEE J. Sel. Areas Commun.}, vol. 27, no. 6, pp.
908--917, Aug. 2009.

\bibitem{Ontology}
B. Vasić, S. K. Chilappagari, D. V. Nguyen and S. K. Planjery, ``Trapping set ontology,'' in {\em Proc. 47th Allerton Conf.,} Monticello, IL, 2009, pp. 1--7.

\bibitem{Hu_magneticIS}
X. Hu, Z. Li, B. Kumar, and R. Barndt, ``Error floor estimation of long LDPC codes on magnetic recording channels,'' {\em IEEE Trans. Magn.}, vol. 46, no. 6, pp. 1836--1839, Jun. 2010.

\bibitem{Schleg}
C. Schlegel and S. Zhang, ``On the dynamics of the error floor behavior
in (regular) LDPC codes,'' {\em IEEE Trans. Inf. Theory}, vol. 56, no. 7,
pp. 3248--3264, Jul. 2010.

\bibitem{Xiao}
H. Xiao, A. H. Banihashemi, and M. Karimi, ``Error rate estimation of
low-density parity-check codes decoded by quantized soft-decision iterative
algorithms,'' {\em IEEE Trans. Commun.}, vol. 61, no. 2, pp. 474--484,
Feb. 2013.

\bibitem{Sina} S. Tolouei and A. H. Banihashemi, ``Fast and accurate error floor estimation of quantized iterative decoders for variable-regular LDPC codes,'' {\em IEEE Comm. Lett.}, vol. 18, no. 8, pp. 1283--1286, Aug. 2014.

\bibitem{But_SS}
B. K. Butler and P. H. Siegel, ``error floor approximation for LDPC codes in the AWGN channel,'' {\em IEEE Trans. Inf. Theory}, vol. 60, no. 12, pp. 7416--7441, Dec. 2014.

\bibitem{Homayoon_SP}
H. Hatami, D. G. M. Mitchell, D. J. Costello and T. E. Fuja, ``Performance bounds and estimates for quantized {LDPC} decoders,'' {\em IEEE Trans. Commun.}, vol. 68, no. 2, pp. 683--696, Feb. 2020.

\bibitem{Ali-TCOM}
A. Farsiabi and A. H. Banihashemi, ``Error floor estimation of LDPC decoders - A code independent approach to measuring the harmfulness of trapping sets," {\em IEEE Trans. Commun.}, vol. 68, no. 5, pp. 2667--2679, May 2020.

\bibitem{Raveendran}
N. Raveendran, D. Declercq and B. Vasic, ``A sub-graph expansion-contraction method for error floor computation,'' {\em IEEE Trans. Commun.}, vol. 68, no. 7, pp. 3984--3995, July 2020.

\bibitem{Zhu}
M. Zhu, M. Jiang and C. Zhao, ``Error floor estimation of {QC-LDPC} coded modulation with importance sampling,'' {\em IEEE Comm. Lett.}, vol. 25, no. 1, pp. 28--32, Jan. 2021. 

\bibitem{Peyman}
P. Neshaastegaran, A. H. Banihashemi and R. Gohary, ``Error floor estimation of {LDPC} coded modulation systems using importance sampling,'' {\em IEEE Trans. Commun.},
available in IEEExplore Early Access.

\bibitem{AS_threshold}
A. Tomasoni, S. Bellini and M. Ferrari, ``Thresholds of absorbing sets in low-density parity-check codes,'' {\em IEEE Trans. Commun.}, vol. 65, no. 8, pp. 3238--3249, Aug. 2017.

\bibitem{Urbank}
T. J. Richardson, M. A. Shokrollahi and R. L. Urbanke, ``Design of capacity-approaching irregular low-density parity-check codes,'' {\em IEEE Trans. Inf. Theory,} vol. 47, no. 2, pp. 619-637, Feb 2001.

\bibitem{Milen}
O. Milenkovic, E. Soljanin, and P. Whiting, ``Asymptotic spectra of trapping sets in regular and irregular LDPC code
ensembles,'' {\em IEEE Trans. Inf. Theory}, vol. 53, no. 1, pp. 39--55, Jan. 2007.

%%%%%Layered decoder references
\bibitem{L1}
M. Mansour, N. Shanbhag, ``High-throughput LDPC decoders,'' {\em IEEE Trans. Very Large Scale Integr. (VLSI) Syst.}, vol. 11, no. 6, pp. 976--996, Dec. 2003.

%\bibitem{L10}
%M. Fossorier, ``Quasicyclic low-density parity-check codes from circulant permutation matrices", {\em IEEE Trans. Inform. Theory}, vol. 50, no. 8, pp. 1788-1793, Aug. 2004.

\bibitem{L2}
D. Hocevar, ``A reduced complexity decoder architecture via layered decoding of LDPC codes,'' in {\em Proc. IEEE Workshop Signal Processing and Systems (SIPS.04)}, Austin, TX, Oct. 2004, pp. 107--112.

\bibitem{XiaoSchedule}
H. Xiao and A. H. Banihashemi, ``Graph-based message-passing schedules for decoding LDPC codes," {\em IEEE Trans. Commun.}, vol. 52, no. 12, pp. 2098--2105, Dec. 2004.

\bibitem{NouhSchedule}
A. Nouh and A. H. Banihashemi, ``Reliability-based schedule for bit-flipping decoding of low-density parity-check codes," {\em IEEE Trans. Commun.}, vol. 52, no. 12, pp. 2038--2040, Dec. 2004.

\bibitem{L4}
T. Brack, M. Alles, F. Kienle, N. Wehn, ``A synthesizable IP core for WiMAX 802.16e LDPC code decoding,'' in {\em Proc. IEEE 17th Int. Symp. Personal Indoor and Mobile Radio Communications}, Sept. 2006, pp. 1--5.

\bibitem{L9}
Z. Wang, Z. Cui, ``Low-complexity high-speed decoder design for quasi-cyclic LDPC codes,'' {\em IEEE Trans. VLSI Syst.}, vol. 15, no. 1, pp. 104--114, Jan. 2007.

\bibitem{L6}
K. Gunnam, G. Choi, M. Yeary, M. Atiquzzaman, ``VLSI architectures for layered decoding for irregular LDPC codes of WiMax,'' in {\em Proc. IEEE Int. Conf. Commun. (ICC)}, June 2007, pp. 4542--4547.

%\bibitem{L5}
%G. Gentile, M. Rovini, L. Fanucci, ``Low-complexity architectures of a decoder for IEEE 802.16e LDPC codes'', {\em Proc. Euromicro Conf. Digital System Design (DSD)}, pp. 369-375, 2007-Aug.



\bibitem{L8}
E. Sharon, S. Litsyn, J. Goldberger, ``Efficient serial message-passing schedules for LDPC decoding,'' {\em IEEE Trans. Inf. Theory}, vol. 53, no. 11, pp. 4076--4091, Nov. 2007.



\bibitem{L3}
C.-H. Liu, S.-W. Yen, C.-L. Chen, H.-C. Chang, C.-Y. Lee, Y.-S. Hsu, S.-J. Jou, ``An LDPC decoder chip based on self-routing network for IEEE 802.16e applications,'' {\em IEEE J. Solid-State Circuits}, vol. 43, no. 3, pp. 684--694, March 2008.

\bibitem{L5}
K. Zhang, X. Huang and Z. Wang,
``High-throughput layered decoder implementation for quasi-cyclic LDPC codes,'' {\em IEEE J. Sel. Areas Commun.}, vol. 27, no. 6, pp. 985--994, August 2009.

\bibitem{L7}
Z. Cui, Z. Wang, X. Zhang, ``Reduced-complexity column-layered decoding and implementation for LDPC codes,'' {\em IET Commun.}, vol. 5, no. 15, pp. 2177--2186, 2011.

\bibitem{InfromedDynamic_wesel_2010}
A. I. V. Casado, M. Griot and R. D. Wesel, ``LDPC decoders with informed dynamic scheduling,'' {\em IEEE Trans. Commun.}, vol. 58, no. 12, pp. 3470--3479, December 2010.

\bibitem{M2I2}
H. Lee and Y. Ueng, ``LDPC decoding scheduling for faster convergence and lower error floor," {\em IEEE Trans. Commun.}, vol. 62, no. 9, pp. 3104--3113, Sept. 2014.

\bibitem{vasic_horizontal}
N. Raveendran and B. Vasic, ``Trapping set analysis of horizontal layered decoder," in {\em Proc. Int.
Conf. Commun. (ICC)}, Kansas City, MO, 2018, pp. 1--6.



%%******************
\bibitem{Angarita_MS_2014}
F. Angarita, J. Valls, V. Almenar and V. Torres, ``Reduced-complexity min-sum algorithm for decoding LDPC codes with low error-floor," {\em IEEE Trans. Circuits Syst. I}, vol. 61, no. 7, pp. 2150--2158, July 2014.

\bibitem{BackTrack_hard_2011}
X. Chen, J. Kang, S. Lin and V. Akella, ``Hardware implementation of a backtracking-based reconfigurable decoder for lowering the error floor of quasi-cyclic LDPC codes," {\em IEEE Trans. Circuits Syst. I}, vol. 58, no. 12, pp. 2931--2943, Dec. 2011.


\bibitem{IDS_kim_2012let}
S. Kim, ``Trapping set error correction through adaptive informed dynamic scheduling decoding of LDPC codes," {\em IEEE Commun. Lett.}, vol. 16, no. 7, pp. 1103--1105, July 2012.

%%%%%%%%

%\bibitem{PEXIT}
%G. Liva and M. Chiani, ``Protograph LDPC Codes Design Based on EXIT Analysis," {\em Proc. IEEE Global Telecommun. Conf.,} Washington, DC, 2007, pp. 3250-3254.

%\bibitem{DE_Blocklayered}
%M. Jang, B. Shin, W. Park, J. No, D. Shin, ``Convergence Speed Analysis of Layered Decoding of Block-Type LDPC Codes." {\em IEICE Trans. Commun.}, vol. e92-b, no. 7, pp. 2484-2487, Jul. 2009.

%\bibitem{Thorpe2003}
%J. Thorpe,``Low-density parity-check (LDPC) codes constructed from
%protographs,'' {\em JPL INP, Tech. Rep.}, Aug. 2003, 42-154.

%\bibitem{Liva2007ExitProtog}
%G. Liva and M. Chiani, ``Protograph LDPC codes design based on EXIT
%analysis,'' in {\em Proc. IEEE GLOBECOM Conf.}, Nov. 2007, pp. 3250-3254.
%%%%%%%


%\bibitem{Sina}
% S. Tolouei and A. H. Banihashemi, ``Fast and accurate error floor estimation of quantized iterative decoders for %variable--regular
%LDPC codes,''{\em IEEE Comm. Lett.}, vol. 18, pp. 1283-–1286, Aug. 2014.




\bibitem{butler_numerical}
B. K. Butler and P. H. Siegel, ``Numerical issues affecting LDPC error floors,''  in {\em Proc. IEEE Global Telecommun. Conf.}, Anaheim, CA, 2012, pp. 3201--3207.

\bibitem{Meyer}
C. D. Meyer, {\em Matrix Analysis and Applied Linear Algebra}. Philadelphia, PA, USA: SIAM, 2000.

\bibitem{varga}
R. S. Varga, {\em Matrix Iterative Analysis}, 2nd ed, Berlin: Springer, 2000.

\bibitem{noutsos_slide}
D. Noutsos, ``Perron-Frobenius theory and some extensions", Como,
Italy, May 2008, [Presentation Slides]. Available: http://www.math.uoi.
gr/ dnoutsos/Papers-pdf-files/slide-perron.pdf.

\bibitem{horn}
R. A. Horn, C. R. Johnson, {\em Matrix Analysis}, Cambridge, U.K.: Cambridge Univ. Press, 1985.

\bibitem{wimax_standard}
{\em IEEE Standard for Local and Metropolitan Area Networks—Part 16:
Air Interface for Fixed and Mobile Broadband Wireless Access Systems
Amendment 2: Physical and Medium Access Control Layers for Combined
Fixed and Mobile Operation in Licensed Bands and Corrigendum
1}, IEEE Standard 802.16e-2005 and 802.16-2004/Cor 1-2005, Feb. 2006.

\bibitem{Ali-Thesis}
A. Farsiabi, ``Code-independent error floor estimation techniques for flooding and layered decoders of
{LDPC} codes,'' Ph.D. dissertation, Dept. of Systems and Computer Engineering, Carleton Univ., Ottawa, ON, Canada,
2020.

\bibitem{FB-TCOM}
A. Farsiabi and A. H. Banihashemi, ``Error floor analysis of column layered decoders,'' to be submitted to {\em IEEE Trans. Commun.}.

\end{thebibliography}
\end{document}